\documentclass[12pt,twoside]{article} 
\usepackage[T1]{fontenc} 

\usepackage{amsthm, amsmath, amssymb, mathrsfs, upgreek, enumerate, mathtools, comment, natbib, graphicx, mathabx, booktabs, threeparttable, relsize, exscale, epstopdf, scrextend, multirow, xr-hyper, pdfpages, bm, caption, subcaption, tikz, adjustbox}
\usepackage[headheight=110pt,margin=1.25in]{geometry}
\usepackage[algoruled]{algorithm2e}
\usepackage[section]{placeins} 

\usepackage{fancyhdr, setspace} 
\usepackage[pagebackref=false]{hyperref} 
\hypersetup{
    colorlinks=true,
    citecolor=blue,
    filecolor=black,
    linkcolor=black,
    urlcolor=blue,
    bookmarksopen=true,
    pdfstartview=FitH
}

\newtheorem{theorem}{Theorem}
\newtheorem{assump}{Assumption}

\newtheorem{proposition}{Proposition}

\newtheorem{sectheorem}{Theorem}[section]

\newtheorem{seclemma}{Lemma}[section]

\theoremstyle{definition}
\newtheorem{definition}{Definition}

\newtheorem{example}{Example}

\newtheorem{remark}{Remark}

\newtheorem{secdefinition}{Definition}[section]

\newtheorem{secremark}{Remark}[section]

\def\Snospace~{\S{}}

\makeatletter
\def\thm@space@setup{
  \thm@preskip=15pt \thm@postskip=15pt 
}
\makeatother

\def\indep{\perp\!\!\!\perp}

\newcommand{\argmin}{\operatornamewithlimits{argmin}}

\newcommand{\cov}{\text{Cov}}
\newcommand{\var}{\text{Var}}
\newcommand{\E}{{\bf E}}
\newcommand{\R}{\mathbb{R}}
\newcommand{\F}{\mathcal{F}}
\newcommand{\N}{\mathcal{N}}
\newcommand{\M}{\mathcal{M}}
\newcommand{\C}{\mathcal{C}}
\newcommand{\prob}{{\bf P}}
\newcommand{\plimarrow}{\stackrel{p}\longrightarrow}
\newcommand{\dlimarrow}{\stackrel{d}\longrightarrow}
\newcommand{\ind}{\bm{1}}
\providecommand{\abs}[1]{\lvert#1\rvert} 
\providecommand{\norm}[1]{\lVert#1\rVert}

\newcommand*{\medcup}{\mathbin{\scalebox{1.5}{\ensuremath{\cup}}}}

\providecommand{\abs}[1]{\lvert#1\rvert} 
\providecommand{\norm}[1]{\lVert#1\rVert}

\renewcommand{\qed}{\hfill \mbox{\raggedright \rule{0.08in}{0.08in}}} 
\renewenvironment{proof}[1][\proofname]{{\noindent\sc#1. }}{\qed\vspace{15pt}} 

\title{\bf\sc Graph Neural Networks for Causal Inference Under Network Confounding\thanks{This paper was previously circulated under the title ``Unconfoundedness with Network Interference.'' We thank Ruonan Xu; seminar audiences at Duke, Singapore Management University, UCSD, Econometrics in Rio 2024, GraphEx2023, and the 2023 North American Summer Meetings; and the editor and anonymous referees for helpful comments.}}

\author{Michael P.\ Leung\thanks{Department of Economics, University of California, Santa Cruz. E-mail: leungm@ucsc.edu.} \and Pantelis Loupos\thanks{Graduate School of Management, University of California, Davis. E-mail: ploupos@ucdavis.edu.}}

\begin{document}
\maketitle
\onehalfspacing
 
\begin{abstract}

  {\sc Abstract.} This paper studies causal inference with observational data from a single large network. We consider a nonparametric model with interference in both potential outcomes and selection into treatment. Specifically, both stages may be the outcomes of simultaneous equations models, allowing for endogenous peer effects. This results in high-dimensional network confounding where the network and covariates of all units constitute sources of selection bias. In contrast, the existing literature assumes that confounding can be summarized by a known, low-dimensional function of these objects. We propose to use graph neural networks (GNNs) to adjust for network confounding. When interference decays with network distance, we argue that the model has low-dimensional structure that makes estimation feasible and justifies the use of shallow GNN architectures. 

  \bigskip


 
\end{abstract}

\addcontentsline{toc}{part}{Main Paper}
\newpage

\section{Introduction}\label{sintro}

Treatment assignment is said to be unconfounded if it is as good as random within subpopulations of observationally equivalent units. When the stable unit treatment value assumption (SUTVA) is plausible, units with identical covariates are naturally considered observationally equivalent. However, when units are connected through a network, they may differ in other observed dimensions that may confound causal inference under network interference. These dimensions may include, for example, the number of type-$x$ neighbors, the number of type-$x$ neighbors with $m$ neighbors of type $y$, and so on through higher-order neighbors. 

Existing work typically adjusts for confounding only on a small subset of these dimensions such as the average covariates of neighbors. However, it may be difficult to justify a particular choice in practice. Neighbor covariates can influence selection into treatment in more complex ways not adequately captured by the mean or other statistic of convenience. 

This paper considers a more general formulation of unconfoundedness. Our objective is estimation of and inference on reduced-form measures of treatment and spillover effects using data from a single large network. We model selection as the outcome of a simultaneous equations model, which allows for peer effects. As a result, treatments are functions of the entire network $\bm{A}$ and all unit covariates $\bm{X}$, and it is generally not possible to summarize confounding by a simple low-dimensional function of these objects. Our unconfoundedness condition instead considers units observationally equivalent if they occupy identical positions in the network, meaning that they match on all observed neighborhood and higher-order neighborhood characteristics.

Methods that rule out interference in selection may result in biased estimates. For example, consider the causal effect of vaccine adoption on illness. With peer effects in adoption, vaccinated individuals tend to have more vaccinated direct and indirect social contacts, and a simple comparison of adopters and non-adopters may overstate vaccine efficacy, even after controlling for neighbors' covariates. 

Because peer effects in selection induce high-dimensional network confounding, they are challenging to accommodate in a nonparametric setting. Most of the literature on interference, including recent work using instrumental variables for identification, rules out selection peer effects \citep{ditraglia2023identifying,hoshino2023causal,kang2016peer}. Papers that allow for them typically rely on semiparametric, game-theoretic models of selection, which substantially reduce the dimensionality of the problem but may be subject to model misspecification \citep{hoshino2023treatment,jackson2020adjusting,kim2020analysis,lin2021selection}.\footnote{Two exceptions are \cite{balat2023multiple}, who study partial identification under strategic complementarities, and \cite{imai2021causal}, who study two-stage randomized trials with noncompliance. Both consider nonparametric selection models but require data consisting of many independent networks.}

To account for network confounding, an initial idea might be to apply double machine learning using the lasso for high-dimensional estimation \citep[e.g.][]{chernozhukov2018double}. Under SUTVA, implementation of the lasso requires the specification of a basis $\{P_k(X_i)\}_{k=1}^d$ for unit-level covariates $X_i$. However, in our setting, a unit $i$'s ``covariates'' correspond to its network position $(i,\bm{X},\bm{A})$ and it is unclear how to choose a basis $\{P_k(i,\bm{X},\bm{A})\}_{k=1}^d$ for such an object. Nothing in the standard toolbox for high-dimensional estimation suggests that it is possible to nonparametrically adjust for network confounding of this sort, which is presumably why it has been avoided in the literature.

\subsection{Contributions}

Our primary contribution is to propose a nonparametric model with endogenous peer effects and make the case that estimation of common reduced-form estimands remains feasible despite high-dimensional network confounding. We utilize the model of approximate neighborhood interference (ANI) proposed by \cite{leung2022causal}, which posits that interference decays with network path distance. Leung shows that ANI allows for endogenous peer effects but focuses on randomized assignment. In observational settings, there may also be peer effects in selection. We therefore relax his assumption of independent treatment assignment to allow for ANI in treatment selection. Since both stages are simultaneous equations models, this creates significant new complications in the form of high-dimensional network confounding.

We observe that graph neural networks (GNNs) can provide a flexible nonparametric basis for functions of network position. We propose to use doubly robust estimators with first-stage nuisance functions approximated by GNNs. Whereas a variety of conventional machine learners can be employed for the first stage under SUTVA, none are well suited to our setting, and the observation that GNNs can fill this gap is novel.

Since the first-stage nuisance functions are nonparametric, high-dimensional functions of network position, some form of low-dimensional structure is required for estimation to be feasible. Additionally, GNNs have been found empirically to perform best with shallow architectures, which correspond to relatively low-dimensional parameterizations, in contrast to deep architectures popular with convolutional neural networks \citep{alon2021bottleneck,li2018deeper}. Our key insight is to draw a novel connection between ANI and approximate sparsity conditions in the lasso literature. We argue that this connection justifies the use of shallow GNN architectures in our setting. 

To understand the idea, let $\N(i,K)$ denote $i$'s {\em $K$-neighborhood}, the set of units whose path distance from $i$ is at most $K$, and $(\bm{X}_{\N(i,K)}, \bm{A}_{\N(i,K)})$ denote the restriction of $(\bm{X},\bm{A})$ to $\N(i,K)$. The key parameter of a GNN is its {\em depth} or number of layers $L$, which determines the {\em receptive field} $(\bm{X}_{\N(i,L)}, \bm{A}_{\N(i,L)})$ used to predict $i$'s outcome. For example, a one-layer GNN only uses $i$'s 1-neighborhood $(\bm{X}_{\N(i,1)}, \bm{A}_{\N(i,1)})$ to predict its outcome, rather than the entirety of $(\bm{X},\bm{A})$. Accordingly, the choice of $L$ depends on prior information about the function being estimated. Under high-dimensional network confounding, it would seem that the propensity score depends nontrivially on the entirety of $(\bm{X},\bm{A})$, which would require a relatively large choice of $L$.

Under ANI, however, interference decays with distance, which implies that a unit $i$'s outcome and treatment are primarily determined by $(\bm{X}_{\N(i,L)}, \bm{A}_{\N(i,L)})$ for relatively small $L$. This is reminiscent of approximate sparsity, under which the regression function primarily depends on a small subset of regressors. As a result, our first-stage nuisance functions can be approximated by lower-dimensional analogs that only depend on the $L$-neighborhood, and these can be directly estimated with shallow $L$-layer GNNs. Our formal result provides primitive conditions on interference rationalizing small choices of $L$ of order $\log n$.

Existing work on network interference assumes the correlation structure is characterized by a dependency graph \citep{emmenegger2022treatment,ogburn2022causal}, which requires dependence between units to be zero after a known network distance. This is analogous to $m$-dependence and incompatible with endogenous peer effects. We establish that the data in our setting is $\psi$-dependent \citep{kojevnikov2021limit}, a correlation structure more similar to mixing and near-epoch dependence in that it allows correlation to decay smoothly with distance. Under high-level conditions on the GNNs, we show that the doubly robust estimator is approximately normally distributed in large networks. We propose a HAC-type variance estimator based on \cite{kojevnikov2021limit} and \cite{kojevnikov2021bootstrap} which we prove is non-negative and consistent. We also provide a bandwidth that adjusts for estimation error in the first-stage machine learners. 

We substantiate the theory in a simulation study and empirical illustration revisiting the microfinance diffusion application of \cite{he2024measuring}. We show how our estimands can capture different dimensions of diffusion complementary to their ``average diffusion at the margin'' measure. Our theoretical framework allows for more complex diffusion processes without requiring the econometrician to prespecify the maximum number of within-period rounds of diffusion. By including richer controls that account for network confounding, we find more attenuated diffusion effects.

\subsection{Related Literature}\label{srelit}

There is a large literature on interference primarily focusing on randomized control trials \citep[e.g.][]{athey2018exact,li2022random,toulis2013estimation}. We contribute to a growing recent literature on unconfoundedness, much of which operates in a partial interference setting where units are partitioned into disjoint groups with no interference across groups \citep[e.g.][]{liu2019doubly,qu2022efficient}. 

\cite{veitch2019using} consider network interference and propose to use as controls ``node embeddings,'' which are learned functions of the graph. Node embeddings can be obtained from a variety of methods, but it is unclear how to justify a particular choice. GNNs learn node embeddings in an end-to-end fashion (see \autoref{sgnn}), and our behavioral model provides justification for their use.

\cite{auerbach2022identification} studies identification conditions distinct from unconfoundedness but proposes a related strategy of ``matching'' on certain network statistics. He provides conditions under which pairwise differencing using unit pairs matched on a novel codegree statistic eliminates selection bias.

Prior to the GNN literature, graph kernels were the dominant method for graph learning tasks \citep{morris2021weisfeiler}. These are to kernel regression as GNNs are to sieve estimation, so graph kernels require a user-specified measure of similarity between regressors, in this case, between two graphs. \cite{auerbach2023local} propose a graph kernel estimator using a novel similarity measure based on graph isomorphism. Since there is no known algorithm for isomorphism testing with polynomial runtime in the network size (\autoref{sgnnprop} discusses GNNs' relationship to this problem), many graph kernel approaches amount to specifying an ``embedding,'' a mapping from networks to Euclidean space \citep{kriege2020survey}. As noted by \cite{wu2020comprehensive}, embeddings are predetermined functions of the network, whereas GNNs produce learnable embeddings.

Finally, our paper contributes to a recent econometric literature using neural networks as sieve estimators for regression functions \citep{athey2021using,farrell2021deep,kaji2020adversarial}. Whereas other machine learners can theoretically be employed in these settings, our problem cannot be solved with standard methods. In this respect, our work relates to \cite{pollmann2021causal} who employs convolutional neural networks to construct counterfactuals for spatial treatments.

The next section defines the model and estimators. We introduce GNNs in \autoref{sgnn} and characterize the asymptotic properties of our estimators in \autoref{slarge}. In \autoref{srfield}, we provide conditions under which the nuisance functions have low-dimensional structure. Section \ref{smc} reports results from a simulation study, and \autoref{sapp} presents the empirical application. Finally, \autoref{sconclude} concludes.

We represent an undirected network $\bm{A}$ as an $n\times n$ binary adjacency matrix with $ij$th entry $A_{ij} \in \{0,1\}$ representing a link between units $i$ and $j$. We assume no self-links, so $A_{ii}=0$. Let $\ell_{\bm{A}}(i,j)$ denote the {\em path distance} between $i,j$ in $\bm{A}$, defined as the length of the shortest path between them, if one exists, and $\infty$ otherwise. The {\em $K$-neighborhood} of a unit $i$ in $\bm{A}$ is denoted by $\N(i,K) = \{j\in\N_n\colon \ell_{\bm{A}}(i,j) \leq K\}$ and its size by $n(i,K) = \abs{\N(i,K)}$. We refer to the elements of $\N(i,K)\backslash\{i\}$ for $K=1$ as $i$'s {\em neighbors} and the elements of the same set for $K>1$ as $i$'s {\em higher-order} neighbors. A unit $i$'s {\em degree} is $n(i,1)-1$, the number of neighbors.

\section{Setup}

Let $\N_n = \{1, \dots, n\}$ be the set of units connected through the network $\bm{A}$. Each unit $i\in\N_n$ is endowed with unobservables $(\varepsilon_i, \nu_i) \in \R^{d_\varepsilon} \times \R^{d_\nu}$ and observables $X_i \in \R^{d_x}$. The model primitives determine outcomes and treatments according to
\begin{equation}
  Y_i = g_n(i,\bm{D},\bm{X},\bm{A},\bm{\varepsilon}) \quad\text{and}\quad D_i = h_n(i,\bm{X},\bm{A},\bm{\nu}), \label{YD}
\end{equation}

\noindent respectively, where $\bm{X} = (X_i)_{i=1}^n$ is the matrix with $i$th row equal to $X_i'$; $\bm{Y}$, $\bm{D}$, $\bm{\varepsilon}$, and $\bm{\nu}$ are similarly defined; and $\{(g_n,h_n)\}_{n\in\mathbb{N}}$ is a sequence of function pairs such that each $g_n(\cdot)$ has codomain $\R$ and $h_n(\cdot)$ codomain $\{0,1\}$. The econometrician observes $(\bm{Y}, \bm{D}, \bm{X}, \bm{A})$, and our analysis treats $(\bm{A}, \bm{X}, \bm{\varepsilon}, \bm{\nu})$ as random.\footnote{The asymptotic theory in \autoref{slarge} conditions on $(\bm{X},\bm{A})$ to avoid imposing additional assumptions on its dependence structure. A design-based analysis would additionally condition on $\bm{\varepsilon}$, but this would generally preclude consistent estimation of the nonparametric functions in the doubly robust estimator defined in \autoref{sest}.} 

The timing of the model is as follows. First, nature draws the primitives $(\bm{A},\bm{X},\bm{\varepsilon},\bm{\nu})$. Next, units select into treatment, potentially based on other units' decisions, and $h_n(\cdot)$ is the reduced-form outcome of that process. Finally, $g_n(\cdot)$ is the reduced form of the subsequent process that generates outcomes. Because $g_n(\cdot)$ and $h_n(\cdot)$ may depend on the primitives of all units, the setup allows $\bm{Y}$ and $\bm{D}$ to be outcomes of simultaneous equations models with endogenous peer effects, as shown in the next examples.

\begin{example}[Linear-in-Means]\label{elim}
  Consider the outcome model 
  \begin{equation*}
    Y_i = \alpha + \beta \frac{\sum_{j=1}^n A_{ij} Y_j}{\sum_{j=1}^n A_{ij}} + \frac{\sum_{j=1}^n A_{ij} Z_j'}{\sum_{j=1}^n A_{ij}} \gamma + Z_i'\delta + \varepsilon_i, 
  \end{equation*}
  
  \noindent where $Z_i = (D_i, X_i')'$ \citep{bramoulle2009identification,manski1993identification}.  The coefficient $\beta$ captures endogenous peer effects, the influence of neighbors' outcomes on own outcomes, while $\gamma$ captures exogenous peer effects, the influence of neighbors' treatments and covariates. Letting $\tilde{\bm{A}}$ denote the row-normalized adjacency matrix and $\bm{1}$ the $n$-dimensional vector of ones, if $\bm{A}$ is connected, the reduced form of the model can be written in matrix form as
  \begin{equation*}
    \bm{Y} =\frac{\alpha}{1-\beta} \ind + \bm{Z}\delta + \sum_{k=0}^\infty \beta^{k+1} \tilde{\bm{A}}^{k+1}\bm{Z}\gamma + \sum_{k=0}^\infty \beta^k \tilde{\bm{A}}^k \bm{\varepsilon}. 
  \end{equation*}

  \noindent This characterizes $Y_i$ as a function $g_n(i,\bm{D},\bm{X},\bm{A},\bm{\varepsilon})$. 
\end{example}

\begin{example}[Binary Game]\label{ebg}
  Consider a binary analog of \autoref{elim} for selection into treatment
  \begin{equation}
    D_i = \bm{1}\left\{ \alpha + \beta \frac{\sum_{j=1}^n A_{ij} D_j}{\sum_{j=1}^n A_{ij}} + \frac{\sum_{j=1}^n A_{ij} X_j'}{\sum_{j=1}^n A_{ij}}\gamma + X_i'\delta + \nu_i > 0 \right\}. \label{Dcomp}
  \end{equation}

  \noindent Unlike \autoref{elim}, there may exist multiple equilibria. The equilibrium selection mechanism can be represented as a reduced-form mapping from the primitives $(\bm{X},\bm{A},\bm{\nu})$ to outcomes $\bm{D}$. This characterizes $D_i$ as a function $h_n(i, \bm{X}, \bm{A}, \bm{\nu})$.\footnote{Our notation represents the equilibrium selection mechanism as a deterministic function of structural primitives, whereas the literature typically represents it as a distribution over equilibria conditional on structural primitives. These two formulations are equivalent because each $\nu_i$ can contain payoff-irrelevant random variables that can be used to generate a distribution over equilibria \citep[see][Remark 1]{leung2021normal}.}

  The previous model corresponds to a game of complete information, whereas in a game of incomplete information, as modeled by \cite{xu2018social} for instance, a unit $i$'s information set is $(\nu_i, \bm{X}, \bm{A})$. The corresponding analog of \eqref{Dcomp} replaces each $D_j$ with $\sigma_j(\bm{X},\bm{A})$, the equilibrium belief that $D_j=1$. Beliefs are typically determined by an equilibrium selection mechanism that is degenerate in the sense that it only depends on public information $(\bm{X}, \bm{A})$. This characterizes $D_i$ as a function $h_n(i,\bm{X},\bm{A},\nu_i)$.
\end{example}

\begin{example}[Diffusion]\label{ediff}
  \cite{he2024measuring} study the following two-period diffusion model. Let $D_i$ denote $i$'s decision to adopt microfinance in period 0 and $Y_i$ its decision in period 1. Their equations (2.4) and (3.6) posit that
  \begin{equation*}
    Y_i = g_n(\bm{D}_{\N(i,K)}, \varepsilon_i) \quad\text{and}\quad D_i = \ind\{W_i'\gamma > \nu_i\},
  \end{equation*}

  \noindent where $W_i$ is a known function of $(\bm{X},\bm{A})$ and $K$ is the maximum distance that adoption decisions can diffuse through the network between periods 0 and 1. We provide a more detailed comparison of our models in \autoref{sapp}.
\end{example}

Given specification \eqref{YD}, we define potential outcomes as
\begin{equation*}
  Y_i(\bm{d}) = g_n(i,\bm{d},\bm{X},\bm{A},\bm{\varepsilon}).
\end{equation*}

\noindent Confounding may arise because $Y_i(\bm{d})$ is potentially correlated with $D_i$ due to high-dimensional observables $(\bm{X},\bm{A})$ and the unobservables $\bm{\varepsilon}$ and $\bm{\nu}$ may be correlated. We restrict the second source of confounding.

\begin{assump}[Unconfoundedness]\label{aigno}
  For any $n\in\mathbb{N}$, $\bm{\varepsilon} \indep \bm{\nu} \mid \bm{X}, \bm{A}$.
\end{assump}

As discussed below, unconfoundedness conditions used in the existing literature additionally limit the first source of confounding to known summary statistics of $(\bm{X},\bm{A})$. Ours is more analogous to standard formulations of unconfoundedness under SUTVA since we do not impose an index restriction on observed confounding.

Because the econometrician only observes a single network, a large-sample theory requires restrictions on interference to obtain some form of weak dependence. We next specify a nonparametric model of decaying interference that accommodates the previous examples. For any $S \subseteq \mathcal{N}_n$, let $\bm{D}_S = (D_i)_{i\in S}$, and similarly define $\bm{X}_S$ and other such submatrices. Let $\bm{A}_S = (A_{ij})_{i,j \in S}$ denote the subnetwork of $\bm{A}$ on $S$, formally the submatrix of $\bm{A}$ restricted to $S$. Recall that $\N(i,s)$ is the $s$-neighborhood of $i$ in $\bm{A}$.

\begin{assump}[ANI]\label{aani}
  There exists a sequence of functions $\{(\gamma_n(\cdot),\eta_n(\cdot))\}_{n\in\mathbb{N}}$ with $\gamma_n,\eta_n\colon \R_+\rightarrow \R_+$ such that $\sup_{n\in\mathbb{N}} \max\{\gamma_{n}(s), \eta_{n}(s)\} \stackrel{s\rightarrow\infty}\longrightarrow 0$ and, for any $n\in\mathbb{N}$,
  \begin{multline}
    \max_{i\in\N_n} \E\big[\lvert g_n(i, \bm{D}, \bm{X}, \bm{A}, \bm{\varepsilon}) \\ - g_{n(i,s)}(i, \bm{D}_{\N(i,s)}, \bm{X}_{\N(i,s)}, \bm{A}_{\N(i,s)}, \bm{\varepsilon}_{\N(i,s)}) \rvert \mid \bm{D}, \bm{X}, \bm{A} \big] \leq \gamma_{n}(s) \quad\text{a.s.} \label{aniy}
  \end{multline}

  \noindent and
  \begin{equation}
    \max_{i\in\N_n} \E\big[\abs{h_n(i, \bm{X}, \bm{A}, \bm{\nu}) - h_{n(i,s)}(i, \bm{X}_{\N(i,s)}, \bm{A}_{\N(i,s)}, \bm{\nu}_{\N(i,s)})} \mid \bm{X}, \bm{A} \big] \leq \eta_n(s) \quad\text{a.s.} \label{anid}
  \end{equation}
\end{assump}

\noindent This is analogous to the model of approximate neighborhood interference proposed by \cite{leung2022causal} but imposed on both the outcome and selection models. Whereas $g_n(i,\dots)$ is unit $i$'s realized outcome, $g_{n(i,s)}(i,\dots)$ is its outcome under a counterfactual ``$s$-neighborhood model.'' In the latter case, we fix all model primitives and treatments at their realized values, drop units outside of $\N(i,s)$ from the model, and direct the remaining units to interact according to the process $g_{n(i,s)}(\cdot)$ to produce counterfactual $s$-neighborhood outcomes.\footnote{This formulation of ANI is related to an estimation strategy proposed by \cite{xu2018social} for binary games on networks with incomplete information. His idea is to approximate an agent $i$'s strategy in the $n$-agent game with its strategy in the counterfactual game restricted to $i$'s $s$-neighborhood.} The error from approximating the observed outcome with the $s$-neighborhood counterfactual is bounded by $\gamma_n(s)$, which decays with the neighborhood radius $s$. This formalizes the idea that $Y_i$ is primarily determined by units relatively proximate to $i$, so that the further a unit is from $i$, the less it influences $i$'s outcome. The second equation imposes the analogous requirement on $D_i$. It is trivially satisfied if $D_i$ is i.i.d., in which case $h_n(i, \bm{X}, \bm{A}, \bm{\nu})$ is only a function of $\nu_i$ for i.i.d.\ $\{\nu_i\}_{i=1}^n$.

\begin{example}\label{eexp}
  For the linear-in-means model in \autoref{elim}, an argument similar to Proposition 1 of \cite{leung2022causal} shows that \eqref{aniy} holds with $\sup_n \gamma_n(s) \leq C\abs{\beta}^s$ for some $C>0$. For the binary game in \autoref{ebg}, an argument similar to Proposition 2 of \cite{leung2022causal} establishes \eqref{anid} with $\sup_n \eta_n(s)$ decaying at an exponential rate with $s$. Finally, for the \cite{he2024measuring} diffusion model in \autoref{ediff}, $Y_i$ only depends on $\bm{D}$ through $\bm{D}_{\N(i,K)}$, so \eqref{aniy} holds with $\gamma_n(s) = c\,\ind\{s < K\}$ for some universal constant $c$. In their empirical application, they use own covariates as controls, so $W_i=X_i$, in which case \eqref{anid} holds with $\eta_n(s) = 0$ for all $s$.
\end{example}

\subsection{Previous Approaches}\label{srelit2}

The first contribution of this paper is to propose a potential outcomes model and formulation of unconfoundedness with richer microfoundations relative to the literature. The second contribution is to show that inference remains possible despite the high-dimensional network confounding unique to our setting. We next elaborate on the first contribution by comparing our model to the literature. The second contribution will be the focus of the remainder of the paper.

The standard SUTVA model and unconfoundedness condition correspond to
\begin{equation}
  Y_i = g(D_i,X_i,\varepsilon_i) \quad\text{and}\quad \varepsilon_i \indep D_i \mid X_i. \label{SUTVA}
\end{equation}

\noindent To generalize this setup to allow for network interference, the literature proceeds as follows. Define
\begin{equation}
  T_i = f_n(i, \bm{D}, \bm{A}) \quad\text{and}\quad W_i = q_n(i, \bm{X}, \bm{A}) \label{YKnbhd}
\end{equation}

\noindent where $f_n(\cdot)$ and $q_n(\cdot)$ are known vector-valued functions. The {\em effective treatment} \citep{manski2013identification} or {\em exposure mapping} \citep{aronow_estimating_2017} $T_i$ is a low-dimensional function of the treatment assignment vector. The {\em network controls} $W_i$ are low-dimensional functions of the covariates. The literature commonly employs the {\em neighborhood interference} model and unconfoundedness condition
\begin{equation}
  Y_i = g(T_i,W_i,\varepsilon_i) \quad\text{and}\quad \varepsilon_i \indep T_i \mid W_i, \label{SOOKnbhd}
\end{equation}

\noindent which is a direct generalization of \eqref{SUTVA} \citep{emmenegger2022treatment,forastiere2021identification,ogburn2022causal}. Here $T_i$ entirely summarizes interference while $W_i$ summarizes confounding.

Common examples of $T_i$ and $W_i$ are 
\begin{equation}
  T_i = \begin{pmatrix} D_i & \sum_{j=1}^n A_{ij} D_j \end{pmatrix} \quad\text{and}\quad W_i = \begin{pmatrix} X_i & \sum_{j=1}^n A_{ij} & \frac{\sum_{j=1}^n A_{ij}X_j}{\sum_{j=1}^n A_{ij}} \end{pmatrix}. \label{TWex}
\end{equation}

\noindent Versions of these are employed in the numerical illustrations of the previous three references. The specification of $T_i$ implies that $Y_i$ depends on $\bm{D}$ only through two statistics: own treatment and the number of treated neighbors. Under model \eqref{SOOKnbhd}, variation in the first component identifies a direct treatment effect and variation in the second a spillover effect. Like most exposure mappings used in the literature, this only depends on $\bm{D}_{\N(i,1)}$, so the outcome model \eqref{SOOKnbhd} implies no interference beyond the 1-neighborhood. Likewise, the choice of $W_i$ implies no confounding beyond 1-neighborhood covariates. More generally, one could restrict the outcome model to depend only on the $K$-neighborhood treatments $\bm{D}_{\N(i,K)}$ for some fixed threshold $K$. Unlike \autoref{aani}, this does not allow for endogenous peer effects. 

Whereas \cite{leung2022causal} focuses on randomized experiments, we study observational data on economic agents who choose to select into treatment. It then becomes important to specify a model of selection rationalizing the econometrician's choice of controls. \cite{sanchez2022spillovers} is the first to provide such a model. Under neighborhood interference \eqref{SOOKnbhd} and an exposure mapping similar to \eqref{TWex}, he shows that it is sufficient to set $W_i=X_i$, that is, to solely control for own covariates. Since much of the literature utilizes controls such as \eqref{TWex}, this raises the question of what model of selection justifies their use or more broadly the use of ``network controls'' that depend more generally on $(\bm{X},\bm{A})$.

Our model \eqref{YD} provides an answer. The presence of complex interference in both the outcome and treatment stages induces selection on $(\bm{X},\bm{A})$, so that it is generally insufficient to control only for a simple summary statistic such as \eqref{TWex}. Our outcome model is considerably more general than \eqref{SOOKnbhd} because we do not require existence of a low-dimensional function $T_i$ of $(\bm{D},\bm{A})$ summarizing interference. Our unconfoundedness condition (\autoref{aigno}) is likewise considerably weaker than \eqref{SOOKnbhd} because we do not require existence of a low-dimensional function $W_i$ of $(\bm{X},\bm{A})$ summarizing confounding.

\subsection{Estimand}\label{eee}

Following much of the literature, we focus on estimands defined by exposure mappings \citep{leung2024causal,savje2021causal}, though the core idea of accounting for high-dimensional network confounding using GNNs may potentially be applied to other estimands. Recall from the previous subsection the definition of the exposure mapping $T_i = f_n(i, \bm{D}, \bm{A})$, where $\{f_n\}_{n\in\mathbb{N}}$ is a sequence of functions with common codomain $\mathcal{T}$, a discrete subset of $\R^{d_t}$. Let $\M_n \subseteq \N_n$ be a subset of the units and $m_n = \abs{\M_n}$. We study the estimand
\begin{multline*}
  \tau(t,t') = \frac{1}{m_n} \sum_{i\in\M_n} \big( \mu_t(i,\bm{X},\bm{A}) - \mu_{t'}(i,\bm{X},\bm{A}) \big) \\ \text{for}\quad \mu_t(i,\bm{x},\bm{a}) = \E[Y_i \mid T_i=t, \bm{X}=\bm{x}, \bm{A}=\bm{a}] \quad\text{and}\quad t,t' \in \mathcal{T}.
\end{multline*}

\noindent This compares average outcomes of units under two different values of the exposure mapping while adjusting for high-dimensional network confounders. The comparison is restricted to a subpopulation $\M_n$, the choice of which can be important for ensuring overlap, as discussed in the following examples. Depending on the choice of $f_n(\cdot)$, $t$, and $t'$, $\tau(t,t')$ may be intended to measure an average treatment or spillover effect, as illustrated in the examples below.

\begin{example}\label{eATE}
  Let $T_i=D_i$ and $\M_n=\N_n$. Then $\tau(1,0)$ compares average outcomes of treated and untreated units using the full network, which is intended to measure the direct effect of the treatment.
\end{example}

\begin{example}\label{ehastn}
  Consider the exposure mapping
  \begin{equation*}
    T_i = \begin{pmatrix} D_i & \ind\left\{\sum_{j=1}^n A_{ij}D_j > 0\right\} \end{pmatrix}.
  \end{equation*}

  \noindent For $t=(0,1)$ and $t'=(0,0)$, $\tau(t,t')$ compares the average outcomes of untreated units with and without at least one treated neighbor, which is intended to capture a spillover effect. For $t=(1,0)$ and $t'=(0,0)$, it compares the average outcomes of treated and untreated units with no treated neighbors. For overlap, we need to exclude units with no neighbors since a treated neighbor occurs with probability zero for such units. This is accomplished by choosing $\M_n$ to be the subset of units whose degree $n(i,1)-1$ lies in some desired set excluding zero. That is, 
  \begin{equation}
    \M_n = \{i \in \N_n\colon n(i,1)-1 \in \bm{\Gamma}\} \quad\text{for some}\quad \bm{\Gamma} \subseteq \R_+\backslash\{0\}. \label{Mn}
  \end{equation}
\end{example}

\begin{example}\label{eextn}
  We obtain a more granular version of \autoref{ehastn} by setting 
  \begin{equation}
    T_i = \begin{pmatrix} D_i & \sum_{j=1}^n A_{ij} D_j \end{pmatrix} \label{Ti5}
  \end{equation}

  \noindent and $\M_n$ according to \eqref{Mn} with $\bm{\Gamma} = \{\gamma\}$ for some $\gamma \in \mathbb{N}$. If we choose $\gamma=3$, $t=(0,2)$ and $t'=(0,0)$, then $\tau(t,t')$ takes the subpopulation of untreated units with degree three and compares those with two versus zero treated neighbors. For this choice of $t,t'$, we require $\gamma\geq 2$ for overlap.
\end{example}

Our large-sample results pertain to the following subpopulations and exposure mappings, which include the previous examples. 

\begin{assump}[Exposure Mappings]\label{aemap}
  Let $\M_n$ be given by \eqref{Mn} for some possibly unbounded interval $\bm{\Gamma} \subseteq \R_+$. For any $t \in \mathcal{T}$, there exist $d \in \{0,1\}$ and a possibly unbounded interval $\bm{\Delta} \subseteq \bm{\Gamma}$ such that for all $i$
  \begin{equation*}
    \bm{1}\{T_i=t\} = \bm{1}\left\{D_i = d,\,\, \sum_{j=1}^n A_{ij}D_j \in \bm{\Delta}\right\}.
  \end{equation*}
\end{assump}

\noindent In \autoref{ehastn} for $t=(0,1)$, this holds for $d=0$, $\bm{\Delta} = (0,\infty)$, and $\bm{\Gamma}$ given in the example. In \autoref{eextn} with $t=(0,2)$, this holds for $d=0$, $\bm{\Delta}=[1.5,2.5]$, and $\bm{\Gamma}=[2.5,3.5]$. 

We restrict to this class of mappings for two reasons. First, it includes the most widely used examples in the literature, which are those presented above. Second, $T_i$ is a complex and discontinuous function of treatments, which are in turn complex functions of the structural primitives. We require additional structure on $T_i$ to characterize the dependence structure of the data and apply a CLT.

\cite{leung2024causal} provides conditions under which $\tau(t,t')$ has a causal interpretation. The following result is a direct consequence of his Theorem A.1. Let $p_{i,t}(\cdot \mid \bm{x},\bm{a})$ denote the conditional probability mass function of $\bm{D} \mid T_i=t, \bm{X}=\bm{x}, \bm{A}=\bm{a}$.

\begin{proposition}\label{pexp}
  Let $t\geq t'$ in the usual partial order. Suppose the exposure mapping $f_n(i,\cdot)$ is componentwise nondecreasing for all $i$ and satisfies \autoref{aemap} (as in Examples \ref{eATE}--\ref{eextn}). Further suppose $\{D_i\}_{i=1}^n$ is independent conditional on $(\bm{X},\bm{A})$. Under \autoref{aigno}, for all $i\in\N_n$ there exists a monotone coupling $\bm{D}_{i,t}^* \stackrel{a.s.}\geq \bm{D}_{i,t'}^*$ independent of $\bm{D}_{-\N(i,1)}$ with $\bm{D}_{i,s}^* \sim p_{i,s}(\cdot \mid \bm{X}, \bm{A})$ for all $s \in \{t,t'\}$ such that 
  \begin{equation*}
    \tau(t,t') = \frac{1}{n} \sum_{i=1}^n \E\big[ Y_i(\bm{D}_{i,t}^*, \bm{D}_{-\N(i,1)}) - Y_i(\bm{D}_{i,t'}^*, \bm{D}_{-\N(i,1)}) \mid \bm{X}, \bm{A} \big]. 
  \end{equation*}
\end{proposition}

\noindent The result represents $\tau(t,t')$ as a convex average of unit-level effects of the form $Y_i(\bm{d}_{i,t}, \bm{d}_{-\N(i,1)}) - Y_i(\bm{d}_{i,t'}, \bm{d}_{-\N(i,1)})$ such that the exposure mapping evaluates to $t$ under the assignment vector $(\bm{d}_{i,t}, \bm{d}_{-\N(i,1)})$ and likewise for $t'$. In \autoref{eextn} with $t=(0,2)$, $\bm{d}_{i,t}$ is the vector with dimension $n(i,1)$ such that $i$ is untreated and has exactly two treated neighbors. Furthermore, $\bm{d}_{i,t},\bm{d}_{i,t'}$ are partially ordered, which means that increasing the exposure mapping from $t$ to $t'$ pushes more neighborhood units into treatment.

Treatments are conditionally independent if treatment adoption follows a nonparametric game of incomplete information $D_i = h_n(i, \bm{X}, \bm{A}, \nu_i)$ as in \autoref{ebg} and private information $\nu_i$ is independently distributed across units conditional on $(\bm{X},\bm{A})$. Structural analyses commonly assume private information is i.i.d.\ and independent of public information \citep[e.g.][]{lin2021selection,xu2018social}.

\subsection{Estimator}\label{sest}

Define the nuisance functions
\begin{equation*}
  p_t(i,\bm{X},\bm{A}) = \prob(T_i=t \mid \bm{X}, \bm{A}) \quad\text{and}\quad \mu_t(i,\bm{x},\bm{a}) = \E[Y_i \mid T_i=t, \bm{X}=\bm{x}, \bm{A}=\bm{a}], 
\end{equation*}

\noindent which we refer to respectively as the (generalized) propensity score \citep{imbens2000role} and outcome regression. Let $\hat{p}_t(i,\bm{X},\bm{A})$ and $\hat{\mu}_t(i,\bm{x},\bm{a})$ denote their respective GNN estimators, which will be defined in \autoref{sgnnest}. We use a standard doubly robust estimator for multi-valued treatments
\begin{equation*}
  \hat\tau(t,t') = \frac{1}{m_n} \sum_{i\in\M_n} \hat{\tau}_i(t,t'),
\end{equation*}

\noindent where
\begin{multline*}
  \hat{\tau}_i(t,t') = \frac{\bm{1}\{T_i=t\} (Y_i - \hat{\mu}_t(i,\bm{X},\bm{A}))}{\hat{p}_t(i,\bm{X},\bm{A})} + \hat{\mu}_t(i,\bm{X},\bm{A}) \\ - \frac{\bm{1}\{T_i=t'\} (Y_i - \hat{\mu}_{t'}(i,\bm{X},\bm{A}))}{\hat{p}_{t'}(i,\bm{X},\bm{A})} - \hat{\mu}_{t'}(i,\bm{X},\bm{A}).
\end{multline*}

To define the variance estimator, we need some notation. Define the bandwidth
\begin{equation}
  b_n = \lceil \tilde b_n \rceil \quad\text{for}\quad \tilde b_n = \left\{ \begin{array}{cc} \frac{1}{4} \mathcal{L}(\bm{A}) & \text{if } \mathcal{L}(\bm{A}) < 2\frac{\log n}{\log \delta(\bm{A})}, \\ \mathcal{L}(\bm{A})^{1/4} & \text{otherwise}, \end{array} \right. \label{ourb}
\end{equation}

\noindent where $\lceil\cdot\rceil$ rounds up to the nearest integer, $\delta(\bm{A}) = n^{-1} \sum_{i=1}^n (n(i,1)-1)$ is the average degree, and $\mathcal{L}(\bm{A})$ is the average path length.\footnote{We assume $\delta(\bm{A})>1$, as is typical in practice. By the average path length, we mean the average over all unit pairs in the largest component of $\bm{A}$. A component is a connected subnetwork such that all units in the subnetwork have infinite path distance to non-members of the subnetwork.} Define $\bm{K}^U$ and $\bm{K}^{PD}$ as $m_n\times m_n$ matrices with respective $ij$th entries $\mathbf{1}\{\ell_{\bm{A}}(i,j)\leq b_n\}$ and 
\begin{equation}
  \frac{\abs{\mathcal{N}(i,b_n/2) \cap \mathcal{N}(j,b_n/2)}}{\abs{\mathcal{N}(i,b_n/2)}^{1/2} \abs{\mathcal{N}(j,b_n/2)}^{1/2}} \label{PSDK}
\end{equation}

\noindent where the rows and columns of both matrices range over $i,j \in \mathcal{M}_n$. We propose the estimator
\begin{align*}
  \hat\sigma^2 = \max\{\hat\sigma^2_U, \hat\sigma^2_{PD}\} \quad&\text{where}\quad \hat\sigma^2_q = \frac{1}{m_n} \tilde{\bm{\tau}}' \bm{K}^q \tilde{\bm{\tau}} \quad\text{for}\quad q \in \{U, PD\}, \\
  &\tilde{\bm{\tau}} = \big(\hat{\tau}_i(t,t') - \hat\mu_t(i,\bm{X},\bm{A}) + \hat\mu_{t'}(i,\bm{X},\bm{A})\big)_{i\in\M_n}.
\end{align*}

\noindent This takes the larger of two HAC estimators using different kernels. Since $\bm{K}^{PD}$ is positive semidefinite, $\hat\sigma^2$ is always non-negative.\footnote{Let $\bm{G}$ be the matrix obtained from $\bm{K}^U$ by dividing each row of the latter by the square root of its row sum and replacing $b_n$ with $b_n/2$. Then $\bm{K}^{PD} = \bm{G}'\bm{G}$, which is positive semidefinite.} 

\begin{remark}
  The bandwidth \eqref{ourb} is similar to the proposal of \cite{leung2022causal} but with constants adjusted to account for the first-stage estimates. The estimator $\hat\sigma^2_U$ is a network HAC estimator with uniform kernel \citep{kojevnikov2021limit}, while $\hat\sigma^2_{PD}$ is inspired by an estimator proposed by \cite{kojevnikov2021bootstrap}.\footnote{Our weights \eqref{PSDK} differ from the kernel weights in equation (4.2) of \cite{kojevnikov2021bootstrap} in the denominator. We make this change to ensure that the weights are uniformly bounded above and tend to one with the sample size when units are connected. \cite{kojevnikov2021bootstrap} has to assume these properties (his Assumption 4.1) which is an implicit restriction on $\bm{A}$.} It was previously found in simulations that sloped kernels, including \eqref{PSDK}, can substantially overreject in simulations compared to the (non-sloped) uniform kernel under network dependence \citep{leung2022causal}. We thus propose to take the larger of the two alternatives.
\end{remark}

\begin{remark}
  Our large-sample theory conditions on $(\bm{X},\bm{A})$. Under conventional unconditional or ``superpopulation'' asymptotics, the correct analogs of the HAC estimators $\hat\sigma^2_U$ and $\hat\sigma^2_{PD}$ would be obtained by redefining $\tilde{\bm{\tau}} = (\hat{\tau}_i(t,t') - \hat\tau(t,t'))_{i\in\M_n}$. That is, we center at an estimate of the unconditional mean. Under our conditional asymptotics, the superpopulation estimators can be shown to be asymptotically conservative using an argument similar to Theorem 4 of \cite{leung2022causal}. Because we condition on observables $(\bm{X},\bm{A})$, we obtain an asymptotically exact estimator by instead centering at an estimate of the conditional mean, similar to \cite{abadie2014inference} and \cite{jin2024tailored}. This is not possible in the design-based setting of \cite{leung2022causal} because his asymptotics are conditional on unobserved potential outcomes.
\end{remark}

\section{Graph Neural Networks}\label{sgnn}

Consider the problem of estimating an unknown scalar function of network position $F^*(i, \bm{X}, \bm{A})$. A GNN estimator of $F^*(\cdot)$ is a parameterized function that maps $(\bm{X},\bm{A})$ to a vector of unit-level predictions $(\hat{F}(i, \bm{X}, \bm{A}))_{i=1}^n \in \R^n$. In \autoref{sarch} we define the standard GNN architecture. In \autoref{sgnnest}, we present GNN estimators for the nuisance functions. In \autoref{sinvar}, we discuss the significance of a nonparametric shape restriction, permutation invariance, that GNNs impose.

\subsection{Architecture}\label{sarch}

The standard GNN architecture consists of nested, parameterized, vector-valued functions called {\em neurons} that are arranged in $L$ {\em hidden layers} with $n$ neurons per layer. Let $h_i^{(l)}$ denote the $i$th neuron in hidden layer $l$, which is typically a vector of dimension or {\em width} $H$ that is common across $i$ and $l$. This is often interpreted as unit $i$'s {\em node embedding}, a Euclidean representation of its network position. As we progress to higher-order layers, say $h_i^{(l)}$ to $h_i^{(l+1)}$, $i$'s embedding becomes richer in a sense discussed below.

Connections between neurons in different layers are determined by $\bm{A}$ through the ``message-passing'' architecture
\begin{equation}
  h_i^{(l)} = \Phi_{0l}\left( h_i^{(l-1)} \,\,,\,\, \Phi_{1l}\big(h_i^{(l-1)}, \{h_j^{(l-1)}\colon A_{ij}=1\}\big) \right), \label{GNNlayer}
\end{equation}

\noindent where $\Phi_{0l}(\cdot),\Phi_{1l}(\cdot)$ are parameterized functions with codomain $\R^H$, examples of which are provided below. We initialize $h_i^{(0)} = X_i$, interpreted as the naive node embedding that incorporates no network information. These are the usual controls under SUTVA. In subsequent layers, unit $i$'s node embedding is updated as a function of its neighbors' embeddings in the previous layer and therefore incorporates increasingly more network information as $l$ increases. This is a modeling strategy common to modern neural network architectures, that of representing a complex object in the problem domain as a Euclidean vector with learnable parameters arranged hierarchically, be it a node's network position (GNNs), a subregion of an image (CNNs), or the meaning of a word or sentence (transformers).

The final hidden layer $h_i^{(L)}$ is passed to an output layer that aggregates the vector into a scalar prediction $\Phi_o(h_i^{(L)})$ for some parameterized function $\Phi_o\colon \R^H \rightarrow \R$. The GNN ultimately returns the vector of predictions $(\Phi_o(h_i^{(L)}))_{i=1}^n \in \R^n$. The GNN estimator, formally defined below, amounts to a regression of $Y_i$ on $i$'s node embedding $h_i^{(L)} \in \R^H$, which itself depends on unknown parameters through $\{(\Phi_{0l}, \Phi_{1l})\}_{l=1}^L$. 

\begin{remark}
  The {\em depth} $L$ of a GNN determines its {\em receptive field}, the $L$-neighborhood subgraph $(\bm{X}_{\N(i,L)}, \bm{A}_{\N(i,L)})$ used to predict $Y_i$. To see this, it helps to understand why a GNN layer \eqref{GNNlayer} is often referred to as a ``round of message passing.'' In this metaphor, $h_j^{(l)}$ is the information, or message, held by unit $j$ at step $l$ of the process. Messages are successively diffused to neighbors of $j$ at the next step $l+1$ and neighbors of neighbors at step $l+2$, etc., since at each step, each unit aggregates the messages of its neighbors. This is reminiscent of DeGroot learning \citep{degroot1974reaching} but with more general aggregation functions. Since we initialize the process at $h_i^{(0)}=X_i$, the final hidden layer $h_i^{(L)}$ is only a function of $(\bm{X}_{\N(i,L)}, \bm{A}_{\N(i,L)})$.
  Several comments are in order.
\end{remark}

\begin{remark}
  The second argument of the {\em aggregation function} $\Phi_{1l}(\cdot)$ is the ``multiset'' (a set with possibly repeating elements) consisting of the node embeddings of the ego's 1-neighborhood. Because multisets are by definition unordered, the labels of the units are immaterial, so the output of each layer is {\em permutation invariant} in the sense that the output remains the same under any labeling of the units.
\end{remark}

\begin{remark}
  Both $\Phi_{0l}(\cdot)$ and $\Phi_{1l}(\cdot)$ depend on the dimension of their input vector(s) but not the number of units $n$. Accordingly, for any fixed architecture, a GNN can take as input a network of any size. The parameter space is determined not by the size of the input network but rather by $L$ and the parameters of $\Phi_{0l}(\cdot), \Phi_{1l}(\cdot), \Phi_o(\cdot)$.
\end{remark}

The choices of $\Phi_{0l}(\cdot)$ and especially $\Phi_{1l}(\cdot)$ define different GNN architectures, two of which we discuss next. First we need to introduce the multilayer perceptron (MLP), which defines a sieve space commonly used in the literature. The GNN output layer $\Phi_o(\cdot)$ is typically an MLP, and $\Phi_{0l}(\cdot),\Phi_{1l}(\cdot)$ are commonly defined using MLPs.

\begin{definition}\label{dMLP}
  A {\em multilayer perceptron (MLP)} is a parameterized function mapping a vector to a scalar. It is characterized by $L$ hidden layers, each with width $H$, an output layer, and an {\em activation function} $\sigma\colon \R\rightarrow\R$. In the first hidden layer, the input vector $x$ is mapped to a vector $(\sigma(\alpha_{1,h} + x'\beta_{1,h}))_{h=1}^H \in \R^H$ where $\alpha_h$ is a scalar and $\beta_{1,h}$ a vector of the same dimension as the input. In layer $l$, the input is the output of the previous layer, denoted $x_{l-1}$, and this is mapped to a vector $(\sigma(\alpha_{l,h} + x_{l-1}'\beta_{l,h}))_{h=1}^H$. In the output layer, the output $x_L$ of the final hidden layer $L$ is mapped to the scalar $\alpha_o + x_L'\beta_o$, which is the output of the MLP. In the special case of $L=0$, the MLP is referred to as a {\em linear layer}. A common choice of $\sigma$ is {\em ReLU activation} $\sigma(x) = \max\{x,0\}$.
\end{definition}

\begin{example}[GIN]\label{eGNNarch}
  Theoretical results on GNNs commonly pertain to the ``graph isomorphism network'' (GIN) architecture, which has a nonparametric interpretation. The specification is
  \begin{equation*}
    h_i^{(l)} = \phi_{0l}\left( h_i^{(l-1)} \,\,,\,\, \sum_{j=1}^n A_{ij} \phi_{1l}(h_j^{(l-1)}) \right), 
  \end{equation*}

  \noindent where $\phi_{0l}(\cdot),\phi_{1l}(\cdot)$ lie in some sieve space, typically chosen to be MLPs. The use of sum aggregation in the second argument is motivated by the key insight that any injective function $\Phi_{1l}(S)$ of a multiset $S$ can be written as $g(\sum_{s\in S} f(s))$ for some functions $f,g$ when $X_i$ has countable support \citep{xu2018powerful}. By approximating the unknown $f$ and $g$ using sieves, this architecture can approximate a large nonparametric function class, as discussed in \autoref{sgnnprop}. 
\end{example}

\begin{example}[PNA]\label{epna}
  Our simulations and empirical application utilize the ``principal neighborhood aggregation'' architecture due to \cite{corso2020principal}, which generalizes GINs by including multiple aggregation functions: 
  \begin{equation*}
    h_i^{(l)} = \phi_{0l}\left( h_i^{(l-1)}\,\,,\,\, \Gamma(\{\phi_{1l}(h_i^{(l-1)}, h_j^{(l-1)})\colon A_{ij}=1\}) \right)
  \end{equation*}

  \noindent where $\phi_{0l}(\cdot),\phi_{1l}(\cdot)$ lie in some sieve space, typically chosen to be MLPs, and $\Gamma(\cdot)$ is a possibly vector-valued function. The theoretical motivation is that the representation in \autoref{eGNNarch} using sum aggregation no longer holds when the support of $X_i$ is uncountable, so using multiple aggregators can result in more powerful architectures \citep{corso2020principal}.

  For an example of $\Gamma(\cdot)$, let $\mu(\cdot), \sigma(\cdot), \Sigma(\cdot), \min(\cdot),$ and $\max(\cdot)$ be respectively the mean, standard deviation, sum, min, and max functions, defined component-wise for multisets of vectors. Then setting $\Gamma(\cdot) = \Gamma_1(\cdot)$ for
  \begin{equation*}
    \Gamma_1(\cdot) = \begin{pmatrix} \mu(\cdot) & \sigma(\cdot) & \Sigma(\cdot) & \min(\cdot) & \max(\cdot) \end{pmatrix}
  \end{equation*}

  \noindent results in an architecture utilizing five aggregation functions.

  \cite{corso2020principal} combine multiple aggregators with ``degree scalars'' that multiply each aggregation function by a function of the size of the multiset input $n(i,1)-1$. The simplest example is the identity scalar, which maps any multiset to unity. This trivially multiplies each aggregation function in $\Gamma_1(\cdot)$ above, but it is useful to consider non-identity scalars. Let $\abs{\cdot}$ be the function that takes as input a multiset and outputs its size. \cite{corso2020principal} define logarithmic amplification and attenuation scalers
  \begin{equation*}
    S(\cdot,\alpha) = \left( \frac{\log (\abs{\cdot}+1)}{\delta} \right)^\alpha, \quad \delta = \frac{1}{n} \sum_{i=1}^n \log\left( \sum_{j=1}^n A_{ij}+1 \right), \quad \alpha\in [-1,1].
  \end{equation*}

  \noindent The choice of $\alpha$ defines whether the scalar ``amplifies'' ($\alpha=1$) or ``attenuates'' ($\alpha=-1$) the aggregation function, and $\alpha=0$ is the identity scalar. The purpose of the logarithm is to prevent small changes in degree from amplifying gradients in an exponential manner with each successive GNN layer. An aggregation function that augments $\Gamma_1(\cdot)$ with logarithmic amplification and attenuation is
  \begin{equation*}
    \Gamma_2(\cdot) = \begin{pmatrix} S(\cdot,0) & S(\cdot,1) & S(\cdot,-1) \end{pmatrix} \,\,\bigotimes\,\, \Gamma_1(\cdot),
  \end{equation*}

  \noindent where $\bigotimes$ denotes the tensor product, resulting in 15 aggregation functions.
\end{example}

\subsection{GNN Estimator}\label{sgnnest}

To construct a GNN estimator, the econometrician must specify $L$, $\{(\Phi_{0l},\Phi_{1l})\}_{l=1}^L$, and $\Phi_o(\cdot)$. The latter is often taken to be an MLP (\autoref{dMLP}). For the hidden layers, the GIN architecture in \autoref{eGNNarch} uses sum aggregation and has a nonparametric justification, while \autoref{epna} generalizes GINs to allow for multiple aggregation functions. We use the latter in our numerical illustrations. \autoref{rL} in \autoref{sapp} provides suggestions for choosing $L$. 

Let $\F_{\text{GNN}}(L)$ denote the set of all GNNs with $L$ layers ranging over all possible functions $\{(\Phi_{0l},\Phi_{1l})\}_{l=1}^L$ and $\Phi_o(\cdot)$ in the chosen classes. Letting $d_{kl}$ be the number of parameters in $\Phi_{kl}(\cdot)$ for $k \in \{0,1\}$ and $d_o$ the number of parameters in $\Phi_o(\cdot)$, the number of parameters in a given $F \in \F_\text{GNN}(L)$ is $d_o + \sum_{l=1}^L (d_{0l}+d_{1l})$. Recall that each such $F$ takes as input $(\bm{X}, \bm{A})$ and outputs a vector of unit-level scalar predictions $\R^n$. We let $F(i,\bm{X},\bm{A}) = \Phi_o(h_i^{(L)})$ denote the $i$th component of $F(\bm{X}, \bm{A})$. A {\em GNN estimator} is a function in this set that minimizes a loss function $\ell(\cdot)$: 
\begin{equation}
  \hat{F}_{\text{GNN}} \in \argmin_{F\in\F_{\text{GNN}}(L)} \sum_{i=1}^n \ell(Y_i, F(i,\bm{X},\bm{A})). \label{GNNestimator}
\end{equation}

When $Y_i$ is real-valued, we may use squared-error loss 
\begin{equation*}
  \ell(Y_i, F(i,\bm{X},\bm{A})) = 0.5(Y_i - F(i,\bm{X},\bm{A}))^2,
\end{equation*}

\noindent in which case $\hat{F}_{\text{GNN}}(\bm{X},\bm{A})$ estimates $F^*(\bm{X},\bm{A}) = (\E[Y_i \mid \bm{X}, \bm{A}])_{i=1}^n$. For binary $Y_i$, we may use the logistic loss 
\begin{equation*}
  \ell(Y_i, F(i,\bm{X},\bm{A})) = -Y_i F(i,\bm{X},\bm{A}) + \log (1+\text{exp}(F(i,\bm{X},\bm{A}))),
\end{equation*}

\noindent in which case $\hat{F}_{\text{GNN}}(\bm{X},\bm{A})$ estimates the log odds $F^*(\bm{X},\bm{A}) = (\log (\E[Y_i \mid \bm{X}, \bm{A}]/(1-\E[Y_i \mid \bm{X}, \bm{A}])))_{i=1}^n$.

Returning to the doubly robust estimator in \autoref{sest}, to estimate the outcome regression with $\R$-valued outcomes, we restrict the sum in \eqref{GNNestimator} to the set of units $i\in\mathcal{M}_n$ for which $T_i=t$ and use squared-error loss to obtain $\hat{\mu}_t(i,\bm{X},\bm{A}) = \hat{F}_{\text{GNN}}(i,\bm{X},\bm{A})$. To estimate the generalized propensity score, we replace $Y_i$ in \eqref{GNNestimator} with $\bm{1}\{T_i=t\}$ and use logistic loss to obtain 
\begin{equation*}
  \hat{p}_t(i,\bm{X},\bm{A}) = \frac{\text{exp}(\hat{F}_{\text{GNN}}(i,\bm{X},\bm{A}))}{1 + \text{exp}(\hat{F}_{\text{GNN}}(i,\bm{X},\bm{A}))}.
\end{equation*}

\subsection{Invariance}\label{sinvar}

Modern neural network architectures often incorporate prior information in the form of input symmetries to reduce the dimensionality of the parameter space \citep{bronstein2021geometric}. Convolutional neural networks (CNNs), widely used in image recognition, process grid-structured inputs and impose translation invariance. GNNs process graph-structured inputs and impose permutation invariance, as previously discussed. We next explain the purpose of this shape restriction and its economic content.

Define a {\em permutation} as an bijection $\pi\colon \N_n \rightarrow \N_n$. Abusing notation, we write $\pi(\bm{X}) = (X_{\pi(i)})_{i=1}^n$, which permutes the rows of matrix $\bm{X}$ according to $\pi$, and similarly define $\pi(\bm{D})$ and permutations of other such arrays. Likewise, we write $\pi(\bm{A}) = (A_{\pi(i)\pi(j)})_{i,j\in\N_n}$, which permutes the rows and columns of the matrix $\bm{A}$. We say a function $F(i,\bm{X},\bm{A})$ is {\em (permutation-)invariant} if for any bijection $\pi\colon \N_n\rightarrow\N_n$ and input $(i,\bm{X},\bm{A})$, we have $F(i,\bm{X},\bm{A}) = F(\pi(i),\pi(\bm{X}),\pi(\bm{A}))$.

\begin{example}\label{einvar}
  The function $F(i,\bm{X},\bm{A}) = i + X_i + \sum_{k=1}^n A_{ik} X_k$ is not invariant because it depends directly on the label $i$. Consider a permutation $\pi$ that swaps the labels of units $i$ and $j$ and swaps the labels of units 1 and $n$. Then $F(\pi(i),\bm{X},\bm{A}) = j + X_j + \sum_{k=1}^n A_{jk} X_k$, while $F(j,\pi(\bm{X}),\pi(\bm{A})) = j + X_i + \sum_{k=1}^n A_{ik} X_k$ since the permutation swaps $i$ and $j$ but swapping the labels 1 and $n$ does not change the summation. Therefore $F$ is not invariant. However, it would have been invariant had $F(i,\bm{X},\bm{A}) = X_i + \sum_{k=1}^n A_{ik} X_k$.
\end{example}

Invariance appears to be necessary for estimation because the nuisance functions would otherwise depend on the unit label $i$. Computing $\hat\tau(t,t')$ would then require estimating $n$ distinct propensity scores $(p_t(i, \bm{X}, \bm{A}))_{i=1}^n$, which is infeasible when the data consists of $n$ units. The literature on neighborhood interference \eqref{SOOKnbhd} avoids this problem by imposing the additional restriction that $p_t(i, \bm{X}, \bm{A}) = p(W_i)$ for some function $p(\cdot)$ that does not depend on $i$. That is, units are observationally equivalent if they have identical controls $W_i$, and observationally equivalent units have identical probabilities of being assigned to an exposure mapping realization of $t$.

The analogous requirement for our setting is that assignment probabilities are equivalent for units that have isomorphic network positions, which is just another way of saying that the propensity score is invariant. This is weaker than the restriction employed under neighborhood interference. 

\begin{example}
  Consider the network in \autoref{intrans} where each unit $i$ has a binary covariate $X_i$ that is an indicator for its color being gray, and let $W_i$ be given as in \eqref{TWex}. Then $W_4=W_5$, but units 4 and 5 are not isomorphic. Whereas the literature requires units 4 and 5 to have identical propensity scores, invariance does not. If we remove unit 3 from the graph, then 4 and 5 have isomorphic positions, and invariance implies that they have equal propensity scores.
\end{example}

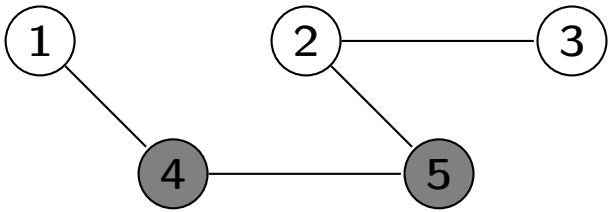
\begin{figure}[ht]
  \centering
  \begin{tikzpicture}[->,>=latex,shorten >=1pt,auto,node distance=2.5cm,
    thick,main
    node/.style={circle,fill=white!20,draw,font=\sffamily\Large\bfseries}]
    \node[main node] (1) {1};
    \node[main node,fill=gray] (2) [below right of=1] {4};
    \node[main node] (3) [above right of=2] {2};
    \node[main node,fill=gray] (4) [below right of=3] {5};
    \node[main node] (5) [above right of=4] {3};
    \draw[-] (1) -- (2);
    \draw[-] (2) -- (4);
    \draw[-] (3) -- (4);
    \draw[-] (3) -- (5);
  \end{tikzpicture}
  \caption{Units 4 and 5 are not isomorphic.}
  \label{intrans}
\end{figure}

If the propensity score is invariant, this reduces our problem from estimating $n$ separate scores to estimating only one because, for any $i$, there exists a permutation $\pi_i$ (in particular the one that only permutes units $1$ and $i$) such that $p_t(i, \bm{X}, \bm{A}) = p_t(1, \pi_i(\bm{X}), \pi_i(\bm{A}))$ and similarly for $\mu_t(\cdot)$. The right-hand side is a function $p_t(1, \cdot, \cdot)$ that does not depend on $i$, so evaluating $i$'s propensity score is now only a matter of supplying the correct $i$-specific inputs $(\pi_i(\bm{X}), \pi_i(\bm{A}))$.

The next result demonstrates that invariance is a weak requirement holding under minimal exchangeability conditions on the structural primitives.

\begin{proposition}\label{pexch}
  Suppose for any $n\in\mathbb{N}$ and permutation $\pi$,
  \begin{multline*}
    f_n(i, \bm{D}, \bm{A}) = f_n(\pi(i), \pi(\bm{D}), \pi(\bm{A})), \\
    g_n(i, \bm{D}, \bm{X}, \bm{A}, \bm{\varepsilon}) = g_n(\pi(i), \pi(\bm{D}), \pi(\bm{X}), \pi(\bm{A}), \pi(\bm{\varepsilon})), \quad\text{and} \\
    h_n(i, \bm{X}, \bm{A}, \bm{\nu}) = h_n(\pi(i), \pi(\bm{X}), \pi(\bm{A}), \pi(\bm{\nu}))
  \end{multline*}

  \noindent a.s., and $(\bm{A}, \bm{X}, \bm{\varepsilon}, \bm{\nu}) \stackrel{d}= (\pi(\bm{A}), \pi(\bm{X}), \pi(\bm{\varepsilon}), \pi(\bm{\nu}))$. Then for any $t\in\mathcal{T}$, $p_t(\cdot)$ and $\mu_t(\cdot)$ are invariant functions.
\end{proposition}
\begin{proof}
  See \autoref{sproofs}.
\end{proof}

\noindent The first three equations impose invariance on $(f_n,g_n,h_n)$, meaning that the functions do not depend directly on unit labels. This is the case for all exposure mappings $f_n$ used in the literature. Applied to $(g_n,h_n)$, invariance only says that unit identities do not influence behavior beyond the model primitives $(\bm{A}, \bm{X}, \bm{\varepsilon}, \bm{\nu})$. The final requirement says that these model primitives are themselves distributionally invariant, which is a weak condition since unit labels carry no intrinsic meaning in our setting. Distributional invariance fails to hold in design-based settings since these condition upon $\bm{\varepsilon}$, among other primitives. Then $i$'s outcome regression depends on $\varepsilon_i$, so it is not invariant with respect to observables alone.

\section{Asymptotic Theory}\label{slarge}

We characterize the asymptotic properties of the estimators in \autoref{sest} under a sequence of models and GNNs sending $n\rightarrow\infty$. Along this sequence, the following objects may change, subject to the conditions imposed below: the exposure mapping $f_n$ from \eqref{YKnbhd}, the structural functions $(g_n,h_n)$ from \eqref{YD}, the distribution of the model primitives $(\bm{A}, \bm{X}, \bm{\varepsilon}, \bm{\nu})$, and the GNN functions $\{(\Phi_o, \Phi_{0l}, \Phi_{1l})\}_{l=1}^L$ and number of layers $L$ defined in \autoref{sarch}. Define 
\begin{multline*}
  \varphi_{t,t'}(i) = \frac{\bm{1}\{T_i=t\} (Y_i - \mu_t(i,\bm{X},\bm{A}))}{p_t(i,\bm{X},\bm{A})} + \mu_t(i,\bm{X},\bm{A}) \\ - \frac{\bm{1}\{T_i=t'\} (Y_i - \mu_{t'}(i,\bm{X},\bm{A}))}{p_{t'}(i,\bm{X},\bm{A})} - \mu_{t'}(i,\bm{X},\bm{A}) - \tau(t,t'), 
\end{multline*}

\noindent whose average over $i\in\mathcal{M}_n$ is the doubly robust moment condition, and let
\begin{equation*}
  \sigma_n^2 = \var\left( \frac{1}{\sqrt{m_n}} \sum_{i\in\M_n} \varphi_{t,t'}(i) \,\bigg|\, \bm{X}, \bm{A} \right).
\end{equation*}

\begin{assump}[Moments]\label{amom}
  (a) There exists $M<\infty$ and $p>4$ such that for any $n\in\mathbb{N}$, $i\in\N_n$, and $\bm{d}\in\{0,1\}^n$, $\E[\abs{Y_i(\bm{d})}^p \mid \bm{X}, \bm{A}] < M$ a.s. (b) There exists $[\underline{\pi},\overline{\pi}] \subset (0,1)$ such that $\hat{p}_t(i,\bm{X},\bm{A}), p_t(i,\bm{X},\bm{A}) \in [\underline{\pi},\overline{\pi}]$ and $m_n/n \geq \underline{\pi}$ a.s.\ for all $n\in\mathbb{N}$, $i\in\M_n$, $t\in\mathcal{T}$. (c) $\liminf_{n\rightarrow\infty} \sigma_n^2 > 0$ a.s. 
\end{assump}

\noindent Part (b) requires sufficient overlap for the propensity scores. Under \autoref{aemap}, this typically holds if $\bm{\Gamma}$ is a bounded set. Part (b) further imposes overlap on the propensity score estimator, which is common in the double machine learning literature \citep[e.g.][]{chernozhukov2018double,farrell2015robust,farrell2021deep}. It also requires that $\M_n$ is a nontrivial subset of $\N_n$. Part (c) is a standard non-degeneracy condition.

Let $Z_i = (i, Y_i, T_i, \bm{X}, \bm{A})$, and for any $x\in\R$, define the functions $\Psi_{p_t}(Z_i, x) = \ind\{T_i=t\} (Y_i - \mu_t(i,\bm{X},\bm{A}))(x^{-1} - p_t(i,\bm{X},\bm{A})^{-1})$ and $\Psi_{\mu_t}(Z_i, x) = (x - \mu_t(i,\bm{X},\bm{A})) (1 - \ind\{T_i=t\} p_t(i,\bm{X},\bm{A})^{-1})$, both of which are mean zero for any $x$.

\begin{assump}[GNNs]\label{agnn}
  The following hold for any $t \in \mathcal{T}$.
  \begin{enumerate}[(a)]
    \item (MSE) Both $m_n^{-1} \sum_{i\in\M_n} (\hat{p}_t(i,\bm{X},\bm{A}) - p_t(i,\bm{X},\bm{A}))^2$ and $m_n^{-1} \sum_{i\in\M_n} (\hat{\mu}_t(i,\bm{X},\bm{A}) - \mu_t(i,\bm{X},\bm{A}))^2$ are $o_p(1)$, and their product is $o_p(n^{-1})$.
    \item (SE) For any $\varrho \in \{p_t, \mu_t\}$, $m_n^{-1} \sum_{i\in\M_n} \Psi_\varrho(Z_i, \hat{\varrho}(i,\bm{X},\bm{A})) = o_p(n^{-1/2})$.
  \end{enumerate}
\end{assump}

\noindent Part (a) requires GNNs to converge at a faster than $n^{-1/4}$ rate, while part (b) is a ``stochastic equicontinuity'' condition. These are familiar requirements in the i.i.d.\ setting \citep[e.g.\ Assumption 3 of][]{farrell2015robust}. Primitive conditions are beyond the scope of the existing literature for reasons discussed in \autoref{sgnnadd}, but we provide two results that may prove useful for verifying \autoref{agnn}. In \autoref{srfield}, we show that the nuisance functions can be well-approximated by $L$-neighborhood analogs for $L=O(\log n)$, making estimation a low-dimensional problem. In \autoref{sgnnprop}, we reframe and combine several theoretical results in the GNN literature to show that the nuisance functions can be approximated by GNNs if they lie in a certain nonparametric subclass of invariant functions. 

The next assumption is used to show that $\{\varphi_{t,t'}(i)\}_{i=1}^n$ is $\psi$-dependent \citep[see \autoref{dpsidep}, due to][]{kojevnikov2021limit}, which enables the application of a CLT. Define 
\begin{equation*}
  \N^\partial(i,s) = \{j\in\N_n\colon \ell_{\bm{A}}(i,j) = s\} \quad\text{and}\quad \delta_n^\partial(s; k) = \frac{1}{n} \sum_{i=1}^n \abs{\N^\partial(i,s)}^k,
\end{equation*}

\noindent respectively $i$'s $s$-neighborhood boundary and the $k$th moment of the $s$-neighborhood boundary size. Let 
\begin{align}
  &\Delta_n(s, m; k) = \frac{1}{n} \sum_{i=1}^n \max_{j\in\N^\partial(i,s)} \abs{\N(i,m) \backslash \N(j,s-1)}^k, \nonumber\\ 
  &c_n(s, m; k) = \inf_{\alpha>1} \Delta_n(s, m; k\alpha)^{1/\alpha} \delta_n^\partial(s; \alpha/(\alpha-1))^{1-1/\alpha}, \quad\text{and} \nonumber\\
  &\psi_n(s) = \max_{i\in\N_n} \left( \gamma_n(s/4) + \eta_n(s/4) \big(1 + n(i,1) + \Lambda_n(i,s/4)\, n(i,s/4)\big)\right), \label{thepsi}
\end{align}

\noindent where $\Lambda_n(i,s/4)$ is a constant defined in the next assumption. The second quantity is roughly a measure of the density of $s$-neighborhoods. The third quantity is a measure of dependence in that it bounds the covariance between $\varphi_{t,t'}(i)$ and $\varphi_{t,t'}(j)$ when $\ell_{\bm{A}}(i,j) \leq s$. Lastly, define 
\begin{equation*}
  G_n(i, \bm{d}_{\N(i,s)}) = \E[g_{n(i,s)}(i, \bm{d}_{\N(i,s)}, \bm{X}_{\N(i,s)}, \bm{A}_{\N(i,s)}, \bm{\varepsilon}_{\N(i,s)}) \mid \bm{X}, \bm{A}].
\end{equation*}

\begin{assump}[Weak Dependence]\label{apsi}
  (a) $\{(\varepsilon_i,\nu_i)\}_{i=1}^n$ is independently distributed conditional on $(\bm{X},\bm{A})$.
  (b) For any $n\in\mathbb{N}$, $i\in\N_n$, $s\geq 0$, and $\bm{d},\bm{d}'\in\{0,1\}^n$,
  \begin{equation*}
    \abs{G_n(i, \bm{d}_{\N(i,s)}) - G_n(i, \bm{d}_{\N(i,s)}')} \leq \Lambda_n(i,s) \sum_{j\in \N(i,s)} \abs{d_j-d_j'} \quad\text{a.s.}
  \end{equation*}

  \noindent for some constant $\Lambda_n(i,s)$. (c) $\sup_{n\in\mathbb{N}} \max_{s\geq 1} \psi_n(s) < \infty$ a.s. (d) For $p$ in \autoref{amom}(a), some positive sequence $v_n \rightarrow\infty$ and any $k \in \{1,2\}$, 
  \begin{equation}
    \frac{1}{n^{k/2}} \sum_{s=0}^\infty c_n(s, v_n; k) \psi_n(s)^{1-(2+k)/p} = o(1) \quad\text{and}\quad n^{3/2} \psi_n(v_n)^{1-1/p} = o(1) \quad\text{a.s.} \label{2ND}
  \end{equation}
\end{assump}

\noindent Parts (a) and (b) are used to establish that $\{\varphi_{t,t'}(i)\}_{i=1}^n$ is $\psi$-dependent. This is a nontrivial task because $(Y_i,T_i)$ are complex functions of treatments which are complex functions of unobservables. Part (b) is a Lipschitz condition that holds if potential outcomes are uniformly bounded. In particular we can take $\Lambda_n(s) = 2M$ where $M$ is the uniform bound on the ranges of $\{g_n\}_{n\in\mathbb{N}}$. Part (a) can be relaxed to allow for some forms of $\psi$-dependent unobservables under additional smoothness conditions on $g_n$ and $h_n$ \citep[][Proposition 2.5]{kojevnikov2021limit}. The simulation study in \autoref{smc} provides evidence that our methods can handle network-dependent unobservables.

The key requirement is part (d). This corresponds to Condition ND of \cite{kojevnikov2021limit}, which they utilize to establish a CLT. It requires covariances to decay with network distance (controlled by $\psi_n(s)$) faster than network neighborhoods densify (controlled by $c_n(s, m; k)$). This is analogous to spatial CLTs that require mixing coefficients to decay with spatial distance $s$ faster than a particular polynomial rate, which corresponds to the rate at which $s$-neighborhoods fill with units in Euclidean space. We verify \eqref{2ND} for different classes of graphs in \autoref{svapsi}. 

\begin{theorem}\label{tclt}
  Under Assumptions \ref{aigno}--\ref{apsi},
  \begin{equation*}
    \sigma_n^{-1/2} \sqrt{m_n}\big(\hat\tau(t,t') - \tau(t,t')\big) \dlimarrow \mathcal{N}(0,1).
  \end{equation*}
\end{theorem}
\begin{proof}
  See \autoref{sproofs}.
\end{proof}

The next assumption is used to characterize the asymptotic properties of $\hat\sigma^2$. 

\begin{assump}[HAC]\label{ahac}
  (a) For some $M>0$ and all $n\in\mathbb{N}$, $i\in\mathcal{N}_n$, and $t\in\mathcal{T}$, $\abs{ \max\{Y_i, \hat{\mu}_t(i,\bm{X},\bm{A})\}} < M$ a.s. (b) For $p$ in \autoref{amom} and $b_n\rightarrow\infty$, \\ $n^{-1} \sum_{s=0}^\infty c_n(s, b_n; 2) \psi_n(s)^{1-4/p} = o(1)$ a.s. (c) $n^{-1} \sum_{i=1}^n n(i,b_n)^2 = O_p(\sqrt{n})$. 
\end{assump}

\noindent Part (a) strengthens \autoref{amom}(a) to uniformly bounded outcomes. Part (b) is Assumption 4.1(iii) of \cite{kojevnikov2021limit}. Part (c) is used to handle the estimation error in the nuisance functions. Under spatial dependence, one can derive an explicit upper bound on the rate at which the HAC bandwidth $b_n$ must diverge because $b_n$-neighborhoods grow at a known rate $b_n^d$ in $\R^d$. This is not the case for networks; \autoref{ahac}(b)--(c) only implicitly characterize the rate for $b_n$ as a function of the network topology, in particular the sizes of network $b_n$-neighborhoods. It then remains to derive an explicit bandwidth that satisfies these conditions for reasonable classes of graphs. In \autoref{sbandchoice}, we show that \eqref{ourb} does.

\begin{theorem}\label{thac}
  Under \autoref{ahac} and the assumptions of \autoref{tclt}, $\hat\sigma^2 = \sigma_n^2 + o_p(1)$.
\end{theorem}
\begin{proof}
  See \autoref{sproofs}.
\end{proof}

\section{Receptive Field}\label{srfield}

This section provides conditions under which the nuisance functions have low-dimensional structure that makes estimation feasible and justifies the use of shallow GNN architectures. As discussed in \autoref{sarch}, the number of layers $L$ in a GNN determines its receptive field $(\bm{X}_{\N(i,L)}, \bm{A}_{\N(i,L)})$ used to construct $i$'s estimate. In practice, it is common to choose a small value, which results in a receptive field that excludes most of the network. This has been found empirically to achieve better predictive performance than deep architectures with large $L$, as we will discuss in \autoref{smc}. We next provide conditions under which the the GNN ``estimand'' $p_t(i,\bm{X}_{\N(i,L)}, \bm{A}_{\N(i,L)})$ can approximate the target $p_t(i,\bm{X},\bm{A})$ for relatively small values of $L$.

Consider the sequence of models and GNNs in \autoref{slarge}, and let $L=L_n$ be a possibly diverging sequence of GNN depths. Recall from \autoref{sgnnest} that a GNN has $d_o + \sum_{l=1}^L (d_{0l}+d_{1l})$ parameters. We provide sufficient conditions for the following properties to hold for some $L = O(\log n)$.
\begin{enumerate}[(a)]
  \item The error from approximating the high-dimensional propensity score with its $L$-neighborhood analog is small in that
    \begin{equation}
      \frac{1}{m_n} \sum_{i\in\M_n} \big(p_t(i,\bm{X},\bm{A}) - p_t(i,\bm{X}_{\N(i,L)}, \bm{A}_{\N(i,L)})\big)^2 = o_p(n^{-1/2}) \label{dasgoal}
    \end{equation}
    and similarly for the outcome regression $\mu_t(\cdot)$. 

  \item $\F_\text{GNN}(L)$ is low-dimensional in that the number of GNN parameters is
    \begin{equation}
      d_o + \sum_{l=1}^L (d_{0l}+d_{1l}) = o(\sqrt{n}).
      \label{ld}
    \end{equation}
\end{enumerate}

\noindent In \autoref{sgnnadd}, we discuss how these properties aid in verifying \autoref{agnn}. We next provide intuition by drawing an analogy to approximate sparsity conditions in the lasso literature. 

\subsection{Approximate Sparsity}

Let $h(X_i) = \E[Y_i \mid X_i]$, and consider a lasso regression of $Y_i$ on a vector of basis functions $P(X_i)$. For the lasso prediction $P(X_i)'\hat\beta$ to be a good estimate of $h(X_i)$, we require
\begin{equation}
  \frac{1}{n} \sum_{i=1}^n (P(X_i)'\hat\beta - h(X_i))^2 = o_p(n^{-1/2}). \label{lassoO}
\end{equation}

\noindent To verify this, it is common to impose approximate sparsity, which consists of the following two conditions \citep[e.g.][]{belloni2014inference}. 
\begin{enumerate}[(a)]
  \item There exists $\beta$ such that $n^{-1}\sum_{i=1}^n (P(X_i)'\beta - h(X_i))^2 = o_p(n^{-1/2})$.
  \item $\norm{\beta}_0 = o(\sqrt{n})$.
\end{enumerate}

\noindent That is, $h(\cdot)$ can be approximated by a simpler function $P(\cdot)'\beta$, and estimating the latter is a low-dimensional problem. These conditions mirror \eqref{dasgoal} and \eqref{ld} and simplify the task of showing \eqref{lassoO} to establishing $n^{-1} \sum_{i=1}^n (P(X_i)'(\hat\beta - \beta))^2 = o_p(n^{-1/2})$ for low-dimensional $\beta$.

\begin{example}\label{eas}
  Suppose $h(X_i) = \sum_{j=1}^\infty P_j(X_i)\theta_j$ with $\abs{\theta_j} \stackrel{j\rightarrow\infty}\longrightarrow 0$. That is, one can order the regressors $P_1(X_i), \ldots, P_m(X_i)$ such that their corresponding true regression coefficients decay to zero. Then the outcome depends primarily on the first few regressors despite $m$ being potentially high-dimensional. This satisfies (a) and (b) above given a sufficiently quick rate of decay for the $\theta_j$s \citep[][\S 4.1.2]{belloni2014inference}.
\end{example}

The main idea in our setting is that the dependence of $Y_i$ and $D_i$ on other units decays with network distance under ANI (\autoref{aani}). That is, these quantities are primarily functions of $(\bm{X}_{\N(i,L)}, \bm{A}_{\N(i,L)})$ for some small radius $L$, which is analogous to \autoref{eas}. We may then approximate $p_t(i,\bm{X},\bm{A})$ with the lower-dimensional estimand $p_t(i,\bm{X}_{\N(i,L)},\bm{A}_{\N(i,L)})$, which an $L$-layer GNN directly estimates.

\subsection{Primitive Conditions}

For $p_t(i,\bm{X},\bm{A})$ to approximate $p_t(i,\bm{X}_{\N(i,L)},\bm{A}_{\N(i,L)})$, we also require a form of conditional independence since the latter drops from the conditioning event the subnetwork external to the $L$-neighborhood. While ANI says that outcomes and treatments are mainly functions of units in a relatively small neighborhood, their primitives may be correlated with those of units far from the neighborhood. We consider it reasonable to assume, in the spirit of ANI, that this correlation instead decays with distance, which is the substance of the next condition.

\begin{assump}[Approximate CI]\label{aani2}
  There exist a sequence of functions $\{\lambda_n(\cdot)\}_{n\in\mathbb{N}}$ with $\lambda_n\colon\R_+\rightarrow\R_+$ and a linear function $r_\lambda\colon\R_+\rightarrow\R_+$ such that $\sup_{n\in\mathbb{N}} \lambda_n(s) \stackrel{s\rightarrow\infty}\longrightarrow 0$, $r_\lambda(s) \geq s$ for all $s\in\R_+$, and 
  \begin{multline*}
    \abs{\E[f(\bm{\varepsilon}_{\N(i,s)}, \bm{\nu}_{\N(i,s)}, \bm{X}_{\N(i,s)}, \bm{A}_{\N(i,s)}) \mid \bm{X}, \bm{A}] \\ - \E[f(\bm{\varepsilon}_{\N(i,s)}, \bm{\nu}_{\N(i,s)}, \bm{X}_{\N(i,s)}, \bm{A}_{\N(i,s)}) \mid \bm{X}_{\N(i,r_\lambda(s))}, \bm{A}_{\N(i,r_\lambda(s))} ]} \leq \lambda_n(s) 
  \end{multline*}

  \noindent a.s.\ for any $n\in\mathbb{N}$, $i\in\N_n$, $s\geq 0$, and $\R$-valued, bounded, measurable function $f(\cdot)$.
\end{assump}

\noindent When $r_\lambda$ is the identity function, the assumption requires that the unobservables of an $s$-neighborhood are approximately conditionally independent of the network outside of this neighborhood, where the approximation error is shrinking with the radius $s$. More generally, it allows the $s$-neighborhood to be approximately conditionally independent of the network outside the greater $r_\lambda(s)$-neighborhood for $r_\lambda(s) \geq s$.

\begin{example}
  Under \autoref{aigno}, we can represent $\varepsilon_i = v_{n,1}(i, \bm{X}, \bm{A}, \bm{U})$ and $\nu_i = v_{n,2}(i, \bm{X}, \bm{A}, \bm{V})$ for some independent random vectors $\bm{U},\bm{V}$ that are independent of $(\bm{X},\bm{A})$ \citep[][Proposition 8.20]{kallenberg1997foundations}. Reparameterize model \eqref{YD} as $Y_i = g_n(i, \bm{D}, \bm{X}, \bm{A}, \bm{\varepsilon}) \equiv \tilde g_n(i, \bm{D}, \bm{X}, \bm{A}, \bm{U})$ and similarly $D_i = \tilde h_n(i, \bm{X}, \bm{A}, \bm{U})$. Suppose ANI holds for the reparameterized model, that is, replacing $(g_n,h_n)$ with $(\tilde g_n, \tilde h_n)$ and $(\bm{\varepsilon},\bm{\nu})$ with $(\bm{U},\bm{V})$ in the statement of \autoref{aani}. This jointly imposes a version of ANI on the outcome, selection, and functions $v_{n,1},v_{n,2}$. Since $(\bm{U},\bm{V}) \indep (\bm{X},\bm{A})$ by supposition, \autoref{aani2} holds for the reparameterized model (substituting $(\bm{U},\bm{V})$ for $(\bm{\varepsilon},\bm{\nu})$) with $r_\lambda(s)=s$ and $\lambda_n(s)=0$ for all $s$.
\end{example}

\begin{example}
  Consider a dyadic network formation model analogous to the one used by \cite{sanchez2022spillovers} where $A_{ij} = \bm{1}\{V(X_i,X_j,\zeta_{ij})>0\}$ for $\R$-valued $V(\cdot)$ and $\{\zeta_{ij}\}_{i<j}$ is a set of i.i.d.\ random variables independent of $(\bm{X},\bm{\varepsilon},\bm{\nu})$. Suppose there exist functions $H_\varepsilon, H_\nu$ and independent vectors $\bm{U} = (U_i)_{i=1}^n$ and $\bm{V} = (V_i)_{i=1}^n$ that are independent of all other primitives such that $(\varepsilon_i,\nu_i) = (H_\varepsilon(U_i,X_i), H_\nu(V_i,X_i))$. Then \autoref{aani2} holds with $r_\lambda(s)=s$ and $\lambda_n(s)=0$ for all $s$. Since \autoref{aani2} only requires approximate independence, it may be possible to verify when links are weakly dependent as in some models of strategic network formation \citep[e.g.][]{leung2021normal}.
\end{example}

\begin{theorem}\label{crc}
  Suppose Assumptions \ref{aani}, \ref{aemap}, and \ref{aani2} hold, $\sup_n \max\{\lambda_n(s), \gamma_n(s), \eta_n(s)\} = O(e^{-\alpha s})$ as $s\rightarrow\infty$ for some $\alpha>0$, and $n^{-1} \sum_{i=1}^n n(i,1)^2 = O_p(1)$. Then \eqref{dasgoal} holds if
  \begin{equation*}
    L \geq r_\lambda\big(((4-\epsilon)\alpha)^{-1} \log n+1\big)
  \end{equation*}

  \noindent for some $\epsilon \in (0,4)$. Further suppose $\sup_{n\in\mathbb{N}} n^{-1} \sum_{i=1}^n n(i,s)^2 = O_p(e^{\xi s})$ as $s\rightarrow\infty$ for some $\xi < 2\alpha$. Under Assumptions \ref{aigno}, \ref{amom}(b), and \ref{ahac}(a), the analog of \eqref{dasgoal} holds for $\mu_t(\cdot)$ if instead $L \geq r_\lambda(((2-\epsilon)(\alpha-\xi/2))^{-1} \log n)$ for some $\epsilon \in (0,2)$.
\end{theorem}
\begin{proof}
  See \autoref{sproofs}. 
\end{proof}

\noindent The result allows $L$ to be of order $\log n$, which quantifies the sense in which the GNN can be shallow. The lower bounds on $L$ given in the theorem are not feasible choices, being dependent on unknowns $r_\lambda(\cdot)$ and $\alpha$. This is similar to how finite-sample bounds for the lasso require restrictions on the penalty parameter involving unknown constants. In \autoref{smc}, we illustrate the performance of different choices of $L$ in simulations and provide suggestions for choosing $L$ in practice.

The first half of the theorem concerns the propensity score. Exponential decay of the interference bounds in \autoref{aani} holds for the models discussed in \autoref{eexp}. Real-world networks are typically sparse, usually formalized as $n^{-1} \sum_{i=1}^n n(i,1) = O_p(1)$, which the theorem mildly strengthens to a second-moment condition. The second half of the theorem concerns $\mu_t(\cdot)$. The requirement $n^{-1} \sum_{i=1}^n n(i,s)^2 = O_p(e^{\xi s})$ for $\xi<2\alpha$ says $s$-neighborhoods grow at a rate $\xi$ not too much larger than the rate $\alpha$ at which interference decays. A similar type of condition is required for a central limit theorem, as discussed in \autoref{slarge} and \autoref{sveri}. 

The last result provides sufficient conditions for \eqref{ld}.

\begin{proposition}\label{pld}
  Suppose the GNN architecture is given by either Example \ref{eGNNarch} or \ref{epna}, and we choose $L = O(\log n)$ GNN layers. For some $\kappa<1/4$, suppose each element of $\{\phi_{kl}(\cdot)\colon k\in\{0,1\}, l=1, \ldots, L\} \cup \{\Phi_o(\cdot)\}$ is an MLP with width $O(n^\kappa \log^2 n)$ and depth $O(\log n)$, uniformly over $l$. Then \eqref{ld} holds.
\end{proposition}
\begin{proof}
  See \autoref{sproofs}. The rates for the MLP width and depth correspond to those of Theorem 1 of \cite{farrell2021deep}.
\end{proof}

\section{Simulation Study}\label{smc}

We design a simulation study with three objectives in mind. The first is to illustrate the finite-sample properties of our proposed estimators for different choices of $L$. The second is to demonstrate that shallow GNNs can perform well even on ``wide'' networks that otherwise may require many layers in the absence of the low-dimensional structure established in \autoref{srfield}. The third is to compare the performance of GNNs to that of standard machine learners utilizing prespecified controls based on \eqref{SOOKnbhd}. 

\subsection{Design}

We simulate $\bm{A}$ from two random graph models. The random geometric graph model sets $A_{ij} = \ind\{\norm{\rho_i-\rho_j} \leq r_n\}$ for $\{\rho_i\}_{i=1}^n \stackrel{iid} \sim \mathcal{U}([0,1]^2)$ and $r_n = (5/(\pi n))^{1/2}$, where $\pi$ is the transcendental number. The Erd\H{o}s-R\'{e}nyi model sets $A_{ij} \stackrel{iid}\sim \text{Ber}(5 / n)$. Both have limiting average degree equal to five. The former model results in ``wide'' networks with high average path lengths that grow at a polynomial rate with $n$, while the latter results in low average path lengths of $\log n$ order. For $n=2000$, the average path length is about 39.5 for random geometric graphs and 4.9 for Erd\H{o}s-R\'{e}nyi graphs.

We draw mutually independent sequences $\{X_i\}_{i=1}^n \stackrel{iid}\sim \mathcal{U}(\{0,0.25,0.5,0.75,1\})$, $\{\varepsilon_i\}_{i=1}^n \stackrel{iid}\sim \N(0,1)$, and $\{\nu_i\}_{i=1}^n \stackrel{iid}\sim \N(0,1)$ independently of $\bm{A}$. For some vectors $\bm{D}=(D_i)_{i=1}^n$ and $\bm{\nu}=(\nu_i)_{i=1}^n$, define
\begin{equation*}
  V_i(\bm{D},\bm{\nu};\theta) = \alpha + \beta \frac{\sum_{j=1}^n A_{ij} D_j}{\sum_{j=1}^n A_{ij}} + \delta \frac{\sum_{j=1}^n A_{ij} X_j}{\sum_{j=1}^n A_{ij}} + \gamma X_i + \nu_i + \frac{\sum_{j=1}^n A_{ij} \nu_j}{\sum_{j=1}^n A_{ij}}
\end{equation*}

\noindent where $\theta = (\alpha,\beta,\delta,\gamma)$. We generate $\{Y_i\}_{i=1}^n$ from a linear-in-means model with $Y_i=V_i(\bm{Y},\bm{\varepsilon};\theta_y)$ and $\theta_y=(0.5,0.8,10,-1)$. We generate $\{D_i\}_{i=1}^n$ according to \autoref{ebg}, so that $D_i = \bm{1}\{V_i(\bm{D},\bm{\nu};\theta_d) > 0\}$ with $\theta_d=(-0.5,1.5,1,-1)$. The equilibrium selection mechanism is myopic best-response dynamics starting from the initial condition $\{D_i^0\}_{i=1}^n$ for $D_i^0 = \bm{1}\{V_i(\bm{0},\bm{\nu};\theta_d) > 0\}$.

The design induces a greater degree of dependence than what our assumptions allow. The error term $\nu_i + \sum_{j=1}^n A_{ij} \nu_j / \sum_{j=1}^n A_{ij}$ is not conditionally independent across units unlike what \autoref{apsi}(a) requires. Also, back-of-the-envelope calculations indicate that peer effects are sufficiently large in magnitude that \autoref{apsi}(d) is violated. Section \ref{samc} presents results for a design that satisfies our assumptions.

The estimand is $\tau(1,0)$ in \autoref{eATE}, whose true value is zero. About 57 percent of units select into treatment, so the effective sample size used to estimate the outcome regressions is around $n/2$ since $\mu_t(i,\bm{X}, \bm{A})$ is estimated only with observations for which $T_i=t$. We report results for $n=1000,2000,4000$.

\subsection{Nuisance Functions}\label{smcne}

For the GNNs, we use the PNA architecture in \autoref{epna} with aggregator $\Gamma_2(\cdot)$ defined in the example and $L=1,2,3$. Both $\phi_{0l}$ and $\phi_{1l}$ are one-layer MLPs (\autoref{dMLP}) with widths $H=3,6,9$. We optimize the GNNs with full-batch gradient descent using the default PyTorch implementation of the ``Adam'' gradient descent algorithm \citep{NEURIPS2019_9015}. We randomize initial parameter values and utilize a learning rate of 0.01. For $\phi_{1l}(\cdot)$, we use a linear layer (\autoref{dMLP}), which is the default for the {\tt PNAConv} class in the PyTorch Geometric package \citep{fey2019fast}. For $\phi_{0l}(\cdot)$, we use a one-layer MLP with ReLU activation. Finally, $\Phi_o(\cdot)$ is a linear layer.

We compare GNNs to conventional machine learners using the prespecified controls given in \eqref{TWex}, which are analogous to those used in the numerical illustrations of \cite{emmenegger2022treatment} and \cite{forastiere2021identification}. We report results for two different machine learners. The first is an MLP with the same number of layers $L$ and width $H$ as the GNNs and trained using the same gradient descent algorithm and learning rate. The second is a random forest with 500 trees and minimum leaf size 10. 

\begin{table}[ht]
\small
\caption{Simulation Results: Random Geometric Graph}
\begin{threeparttable}
\begin{tabular}{llrrrrrrrrr}
\toprule
& & \multicolumn{3}{c}{$L=1$} & \multicolumn{3}{c}{$L=2$} & \multicolumn{3}{c}{$L=3$} \\
\cmidrule(lr){3-5} \cmidrule(lr){6-8} \cmidrule(lr){9-11}
\multicolumn{2}{c}{$n$}& 1000 & 2000 & 4000 & 1000 & 2000 & 4000 & 1000 & 2000 & 4000 \\
\multicolumn{2}{c}{$\sum_{i=1}^n D_i$} & 568 & 1138 & 2275 & 568 & 1138 & 2275 & 568 & 1138 & 2275 \\
\multicolumn{2}{c}{$H$} & 3 & 6 & 9 & 3 & 6 & 9 & 3 & 6 & 9 \\
\midrule
\multirow[t]{5}{*}{GNN} & $\hat\tau$ & 0.1801 & 0.1178 & 0.1034 & 0.1164 & 0.0439 & 0.0339 & 0.1531 & 0.0498 & 0.0319 \\
 & CI & 0.9190 & 0.9246 & 0.9204 & 0.9138 & 0.9302 & 0.9406 & 0.9062 & 0.9096 & 0.9318 \\
 & CI+ & 0.9412 & 0.9470 & 0.9404 & 0.9362 & 0.9478 & 0.9550 & 0.9308 & 0.9344 & 0.9496 \\
 & SE & 0.4826 & 0.3191 & 0.2250 & 0.4640 & 0.2851 & 0.2021 & 0.4896 & 0.2712 & 0.1921 \\
 & SE+ & 0.5191 & 0.3445 & 0.2432 & 0.4976 & 0.3071 & 0.2178 & 0.5244 & 0.2922 & 0.2070 \\
\cmidrule(lr){1-11}
\multirow[t]{2}{*}{oracle} & CI & 0.9374 & 0.9372 & 0.9240 & 0.9424 & 0.9450 & 0.9444 & 0.9370 & 0.9450 & 0.9460 \\
 & SE & 0.5263 & 0.3322 & 0.2277 & 0.5351 & 0.3035 & 0.2073 & 0.5675 & 0.3039 & 0.2039 \\
\cmidrule(lr){1-11}
\multirow[t]{2}{*}{naive}  & CI & 0.6832 & 0.6874 & 0.6858 & 0.6646 & 0.6970 & 0.6996 & 0.6606 & 0.6754 & 0.6822 \\
 & SE & 0.2748 & 0.1800 & 0.1258 & 0.2574 & 0.1562 & 0.1090 & 0.2692 & 0.1503 & 0.1042 \\
\cmidrule(lr){1-11}
\multirow[t]{3}{*}{MLP} & $\hat\tau$ & 0.2114 & 0.1932 & 0.1804 & 0.2114 & 0.1932 & 0.1804 & 0.2114 & 0.1932 & 0.1804 \\
 & CI & 0.9124 & 0.9002 & 0.8682 & 0.9124 & 0.9002 & 0.8682 & 0.9124 & 0.9002 & 0.8682 \\
 & SE & 0.4617 & 0.3289 & 0.2330 & 0.4617 & 0.3289 & 0.2330 & 0.4617 & 0.3289 & 0.2330 \\
\cmidrule(lr){1-11}
\multirow[t]{3}{*}{forest} & $\hat\tau$ & 0.1926 & 0.1461 & 0.1112 & 0.1926 & 0.1461 & 0.1112 & 0.1926 & 0.1461 & 0.1112 \\
 & CI & 0.8652 & 0.8764 & 0.8898 & 0.8652 & 0.8764 & 0.8898 & 0.8652 & 0.8764 & 0.8898 \\
 & SE & 0.3866 & 0.2800 & 0.2039 & 0.3866 & 0.2800 & 0.2039 & 0.3866 & 0.2800 & 0.2039 \\
\bottomrule
\end{tabular}
\begin{tablenotes}[para,flushleft]
  \footnotesize 5k simulations. Estimand is $\tau(1,0)=0$. CI $=$ empirical coverage of 95\% CIs. SE $=$ our standard errors. SE+ $=$ \cite{gao2025causal} SE. Oracle $=$ GNNs with true SE. Naive $=$ GNNs with i.i.d.\ SE. Average bandwidth $\eqref{ourb}=3$ for each SE cell.
\end{tablenotes}
\end{threeparttable}
\label{simresultsRGG}
\end{table}

\begin{table}[ht]
\small
\caption{Simulation Results: Erd\H{o}s-R\'{e}nyi}
\begin{threeparttable}
\begin{tabular}{llrrrrrrrrr}
\toprule
& & \multicolumn{3}{c}{$L=1$} & \multicolumn{3}{c}{$L=2$} & \multicolumn{3}{c}{$L=3$} \\
\cmidrule(lr){3-5} \cmidrule(lr){6-8} \cmidrule(lr){9-11}
\multicolumn{2}{c}{$n$} & 1000 & 2000 & 4000 & 1000 & 2000 & 4000 & 1000 & 2000 & 4000 \\
\multicolumn{2}{c}{$\sum_{i=1}^n D_i$} & 591 & 1183 & 2366 & 591 & 1183 & 2366 & 591 & 1183 & 2366 \\
\multicolumn{2}{c}{$H$} & 3 & 6 & 9 & 3 & 6 & 9 & 3 & 6 & 9 \\
\midrule
\multirow[t]{5}{*}{GNN} & $\hat\tau$ & 0.0975 & 0.0559 & 0.0486 & 0.0854 & 0.0338 & 0.0258 & 0.1050 & 0.0420 & 0.0268 \\
 & CI & 0.9004 & 0.9196 & 0.9138 & 0.8994 & 0.9266 & 0.9328 & 0.8874 & 0.9152 & 0.9264 \\
 & CI+ & 0.9594 & 0.9746 & 0.9702 & 0.9574 & 0.9778 & 0.9770 & 0.9534 & 0.9712 & 0.9750 \\
 & SE & 0.2086 & 0.1382 & 0.0982 & 0.2074 & 0.1316 & 0.0939 & 0.2161 & 0.1311 & 0.0927 \\
 & SE+ & 0.2604 & 0.1738 & 0.1232 & 0.2586 & 0.1655 & 0.1176 & 0.2695 & 0.1652 & 0.1163 \\
\cmidrule(lr){1-11}
\multirow[t]{2}{*}{oracle}  & CI & 0.9350 & 0.9400 & 0.9282 & 0.9342 & 0.9466 & 0.9438 & 0.9258 & 0.9436 & 0.9378 \\
 & SE & 0.2500 & 0.1485 & 0.1028 & 0.2493 & 0.1425 & 0.0988 & 0.2613 & 0.1450 & 0.0979 \\
\cmidrule(lr){1-11}
\multirow[t]{2}{*}{naive} & CI & 0.7778 & 0.7846 & 0.7748 & 0.7702 & 0.7946 & 0.7942 & 0.7520 & 0.7772 & 0.7872 \\
 & SE & 0.1507 & 0.0984 & 0.0692 & 0.1478 & 0.0921 & 0.0649 & 0.1545 & 0.0921 & 0.0641 \\
\cmidrule(lr){1-11}
\multirow[t]{3}{*}{MLP} & $\hat\tau$ & 0.1415 & 0.1311 & 0.1217 & 0.1415 & 0.1311 & 0.1217 & 0.1415 & 0.1311 & 0.1217 \\
 & CI & 0.8742 & 0.8272 & 0.7452 & 0.8742 & 0.8272 & 0.7452 & 0.8742 & 0.8272 & 0.7452 \\
 & SE & 0.2045 & 0.1459 & 0.1043 & 0.2045 & 0.1459 & 0.1043 & 0.2045 & 0.1459 & 0.1043 \\
\cmidrule(lr){1-11}
\multirow[t]{3}{*}{forest} & $\hat\tau$ & 0.1515 & 0.1213 & 0.0696 & 0.1515 & 0.1213 & 0.0696 & 0.1515 & 0.1213 & 0.0696 \\
 & CI & 0.8236 & 0.8032 & 0.8584 & 0.8236 & 0.8032 & 0.8584 & 0.8236 & 0.8032 & 0.8584 \\
 & SE & 0.1828 & 0.1291 & 0.0936 & 0.1828 & 0.1291 & 0.0936 & 0.1828 & 0.1291 & 0.0936 \\
\bottomrule
\end{tabular}
\begin{tablenotes}[para,flushleft]
  \footnotesize See table notes of \autoref{simresultsRGG}. Average bandwidth $\eqref{ourb}=2$ for each SE cell.
\end{tablenotes}
\end{threeparttable}
\label{simresultsER}
\end{table}

\subsection{Results}\label{ssimresults}

Tables \ref{simresultsRGG} and \ref{simresultsER} report the results of 5000 simulations for the random geometric graph and Erd\H{o}s-R\'{e}nyi models, respectively. Rows marked $\hat\tau$, SE, and CI report the average point estimate, average standard error, and empirical coverage of 95-percent confidence intervals, respectively. Since the true estimand is zero, the absolute value of $\hat\tau$ also equals the bias. Rows SE+ and CI+ use the \cite{gao2025causal} standard errors.\footnote{In the notation of \autoref{sest}, their variance estimator is $m_n^{-1}\tilde{\bm{\tau}}'\bm{K}_+^U\tilde{\bm{\tau}}$ where $\bm{K}_+^U = \bm{Q} \max\{\bm{\Lambda},\bm{0}\} \bm{Q}'$ and $\bm{Q}\bm{\Lambda}\bm{Q}'$ is the eigendecomposition of $\bm{K}^U$.} The ``oracle'' rows report results using the true standard error, computed by taking the standard deviation of $\hat\tau(1,0)$ across simulation draws. The ``naive'' rows report results using i.i.d.\ standard errors, which illustrate the degree of dependence. All other SE and CI rows use our variance estimator in \autoref{sest}. 

We first compare the standard machine learners to GNNs with $L=1$ since all share the same receptive field in this case. The bias obtained with GNNs is smaller, often by half, particularly in the Erd\H{o}s-R\'{e}nyi case. The most competitive alternative is random forests which achieves a similar bias when the network is the random geometric graph, but in the case of the Erd\H{o}s-R\'{e}nyi model, its bias is over 30 percent larger. This suggests that GNNs can learn a different function of $(\bm{X},\bm{A})$ than common prespecified controls, one that can better adjust for confounding. The improvement in bias using GNNs does not come at an apparent cost to variance as the SEs are comparable across methods.

Second, we compare the GNN estimators across different choices of $L$. The best performance in terms of bias is achieved with $L=2$. This is the case for both random graph models and is particularly notable for the random geometric graph because its average path length is substantially larger than the radius of the receptive field. This demonstrates that GNNs can perform well despite only controlling for $(\bm{X}_{\N(i,2)}, \bm{A}_{\N(i,2)})$, which is possible due to the low-dimensional structure established in \autoref{srfield}. Unsurprisingly, $L=2$ outperforms $L=1$ since the latter only adjusts for 1-neighborhood confounding. In principle, $L=3$ accounts for higher-order network confounds, but the bias turns out to be slightly larger and the coverage slightly worse. 

The choice of $L=2$ is not unusual in the GNN literature. For example, \cite{das2023credit} and \cite{ying2018graph} use $L=2$, while \cite{ma2021graph} use $L=3$. \cite{zhou2021understanding} compute the prediction error of GNNs across several different datasets and architectures with $L=2,4,8,\dots$ and find that $L=2$ has the best performance. The fact that GNN performance often fails to improve, and indeed can worsen, with larger $L$ is well known in the GNN literature.\footnote{\cite{bronstein2020do} writes, ``Significant efforts have recently been dedicated to coping with the problem of depth in graph neural networks, in hope to [sic] achieve better performance and perhaps avoid embarrassment in using the term `deep learning' when referring to graph neural networks with just two layers.''} See \autoref{sbv} for a survey of possible explanations.

\begin{remark}[Choosing $L$ in Practice]\label{rL}
  The network's average path length (APL) and the values of $L$ used in the literature above are useful reference points for the choice of $L$. At the extreme of $L=\text{APL}$, the GNN has a large receptive field encompassing most of the network, so in light of the low-dimensional structure established in \autoref{srfield}, we suggest considering values of $L$ below half the APL (e.g.\ Erd\H{o}s-R\'{e}nyi graphs with $n=2000$ have APLs about 5). However, for networks with high APLs (e.g.\ random geometric graphs with $n=2000$ have APLs about 40), values near this bound may exhibit poor performance given the empirical evidence above. In these cases, one should truncate to values around 3 or 4 at most. To assess robustness, we suggest reporting results for a range of values of $L$ below these upper bounds.
\end{remark}

The oracle CIs achieve coverage close to the nominal level across most sample sizes and architectures, which illustrates the quality of the normal approximation. Our SEs result in some undercoverage, which is common for HAC estimators, though coverage tends toward the nominal level as $n$ grows for $L=2$. The \cite{gao2025causal} SEs are generally larger than the oracle SEs and hence conservative, which improves coverage in smaller samples.

\section{Empirical Application}\label{sapp}

We revisit the analysis of \cite{he2024measuring} on the diffusion of microfinance through rural villages in Karnataka, India. They use data collected by \cite{banerjee2013diffusion} containing twelve dimensions of social relationships, demographic details, and microfinance adoption decisions from 43 villages involved with Bharatha Swamukti Samsthe's (BSS) microfinance program in 2007. BSS initiated the program by meeting with a select group of village ``leaders'' who were asked to spread the word about microfinance. 

Following the analysis of \cite{he2024measuring}, the unit of observation is the household, and household observables $X_i$ include the normalized total number of households within each village and indicators for participation in self-help groups, savings activities, and caste composition. They construct three social networks from the multigraph data: $G_{ee}$ represents connections through material exchanges like borrowing or lending essentials; $G_{sc}$ captures social activities including advice sharing or joint religious attendance; and $G_{all}$ is the union of $G_{ee}$ and $G_{sc}$. We report results with $\bm{A}$ set to each of the three. The average path length is about 3 for all graphs.

We consider three different definitions of the treatment. In the ``leader case,'' $D_i$ is an indicator for whether household $i$ has a leader. In the ``leader-adopter case,'' it is an indicator for whether the household has a leader who adopts microfinance in the first trimester of the study. These definitions are used by \cite{he2024measuring}. We add to these the ``adopter case,'' where $D_i$ is simply an indicator for whether any household member adopts microfinance in the first trimester. The outcome $Y_i$ is an indicator for whether an individual in household $i$ adopts microfinance starting in the first trimester of the study or later.\footnote{In the adopter case, the treatment time period intersects with that of the outcome in the first trimester, which generates unwanted feedback between the outcome and selection models. To avoid this, we should define $Y_i$ as an indicator for adoption {\em after} the first trimester. Fortunately, adoption decisions are never reversed in the data, so this coincides with the present definition of $Y_i$.}

\subsection{Comparison with He and Song (2024)}

We next discuss estimands in the context of the adopter case. \cite{he2024measuring} propose a novel estimand called ``average diffusion at the margin'' (ADM). This measures the expected number of neighbors induced to adopt microfinance as a result of the ego's adoption. To identify the ADM, they assume the following. First, initial adoption decisions are unconfounded (their Assumption 2.1), and in the application, they only use household observables $X_i$ as the controls. Second, their selection model, as described in \autoref{ediff}, is a parametric single-agent discrete choice model. Third, they assume adoption decisions are irreversible in that $Y_i \geq D_i$, which is true in the application. Finally, as discussed in \autoref{ediff}, the econometrician must specify the maximum number of rounds of diffusion that take place between the measurement of $D_i$ and $Y_i$, and they choose $K=1$ in the application.

Our approach has several advantages. First, we do not require knowledge of the number of within-period rounds of diffusion, binary outcomes, or irreversibility of adoption decisions. Second, we use a richer set of network controls that includes covariates of higher-order neighbors, without assuming a known function $W_i$ as in \eqref{SOOKnbhd}. Third, we employ a nonparametric selection model allowing for peer effects in initial adoption. Fourth, our outcome model \eqref{YD} also allows for peer effects in subsequent adoption, as well as higher-order diffusion beyond the $K=1$ neighborhood since outcomes depend on the entire initial adoption vector $\bm{D}$. ANI posits that this dependence decays with distance, which is a feature of most diffusion models, as information from distant units is less likely to diffuse to the ego. 

The cost of imposing less structure than \cite{he2024measuring} is that the ADM may not be identified under our assumptions. Instead, we consider two estimands defined by the following exposure mapping:
\begin{equation*}
  T_i = \left\{\begin{array}{ll} 2 & \text{if } \sum_{j=1}^n A_{ij} D_j > 1 \\ 1 & \text{if } \sum_{j=1}^n A_{ij} D_j = 1 \\ 0 & \text{otherwise} \end{array}\right.
\end{equation*}

\noindent with $\mathcal{M}_n$ defined in \eqref{Mn}. Then $\tau(1,0)$ ($\tau(2,0)$) measures the effect on own adoption of going from 0 to 1 (two or more) adopting neighbor(s) for units with at least one neighbor.\footnote{Assumption 2.1 of \cite{he2024measuring} implies that $\{D_i\}_{i=1}^n$ is independent conditional on observables. Our estimands have a causal interpretation under the same condition by \autoref{pexp}.} This sheds light on a different dimension of diffusion relative to the ADM. Whereas the ADM measures how many alters are affected by the ego's adoption, our estimands quantify the effect of having multiple adopting neighbors on the ego's adoption. We will find that having multiple adopting neighbors has a much larger effect than having only one.

As previously stated, \cite{he2024measuring} define the treatment in two ways. One is a binary indicator for having a leader in the household, the idea being that all leaders were initially informed about microfinance and told to spread the word. However, not all leaders adopted in the first trimester, which perhaps motivates the second definition, a binary indicator for having an adopting leader in the household. In our view, it may be plausible to argue that microfinance adoption in the initial period is as good as random within observable subpopulations, but it is less plausible to make the same case for being a leader, which is likely determined by a complex social process. We therefore consider a third definition, which is simply an indicator for adopting microfinance in the initial period, irrespective of having a leader in the household. Recall that the interpretations of the causal estimands above pertain to the third definition; in our view, the interpretations are less clear when treatment is defined as in the other cases. 

\subsection{Results}

We present estimates of $\tau(1,0)$, $\tau(2,0)$, and the ADM for the three network specifications and three treatment definitions introduced above. We use three different machine learning estimators of the nuisance functions. The first is GNNs using the same PNA architecture, learning rate, and gradient descent algorithm as the simulation study (see \autoref{smcne}). The other estimators use the same prespecified controls \eqref{TWex} and machine learners (MLPs and random forests) as the simulation study. For both neural networks, we set the width to $H=4$ to match the number of household covariates $X_i$.

\begin{table}[t]
\centering
\caption{Estimand $\tau(1,0)$}
\begin{adjustbox}{width=\textwidth}
\begin{threeparttable}
\begin{tabular}{lrrrrrrrc}
\toprule
 & ADM & \multicolumn{3}{c}{GNN} & \multicolumn{3}{c}{MLP} & forest \\
 \cmidrule(lr){2-2} \cmidrule(lr){3-5} \cmidrule(lr){6-8} \cmidrule(lr){9-9}
& & 1 Layer & 2 Layer & 3 Layer & 1 Layer & 2 Layer & 3 Layer \\
\midrule
\multicolumn{8}{l}{Leader case} \\
$G_{ee}$ & -0.052 & 0.003 (0.016) & 0.005 (0.014) & 0.008 (0.015) & 0.012 (0.015) & 0.010 (0.015) & 0.007 (0.015) & 0.003 (0.014) \\
$G_{sc}$ & -0.049 & 0.022 (0.017) & 0.018 (0.016) & 0.047 (0.021) & 0.021 (0.021) & 0.017 (0.020) & 0.021 (0.019) & 0.029 (0.018) \\
$G_{all}$ & -0.050 & 0.001 (0.026) & 0.011 (0.017) & 0.018 (0.019) & 0.005 (0.022) & 0.007 (0.022) & 0.015 (0.019) & 0.020 (0.018) \\
\multicolumn{8}{l}{Leader-adopter case} \\
$G_{ee}$ & 0.215 & 0.092 (0.016) & 0.075 (0.019) & 0.083 (0.027) & 0.111 (0.035) & 0.104 (0.034) & 0.076 (0.025) & 0.057 (0.022) \\
$G_{sc}$ & 0.434 & 0.074 (0.020) & 0.089 (0.018) & 0.075 (0.022) & 0.092 (0.032) & 0.073 (0.026) & 0.090 (0.022) & 0.038 (0.018) \\
$G_{all}$ & 0.435 & 0.079 (0.020) & 0.084 (0.018) & 0.066 (0.024) & 0.081 (0.026) & 0.076 (0.025) & 0.095 (0.026) & 0.039 (0.018) \\
\multicolumn{8}{l}{Adopter case} \\
$G_{ee}$ & 0.423 & 0.071 (0.014) & 0.055 (0.015) & 0.065 (0.016) & 0.066 (0.016) & 0.066 (0.017) & 0.061 (0.017) & 0.027 (0.014) \\
$G_{sc}$ & 0.622 & 0.031 (0.015) & 0.033 (0.013) & 0.022 (0.014) & 0.033 (0.015) & 0.033 (0.015) & 0.036 (0.015) & 0.009 (0.014) \\
$G_{all}$ & 0.657 & 0.028 (0.014) & 0.033 (0.014) & 0.011 (0.016) & 0.030 (0.016) & 0.032 (0.015) & 0.029 (0.015) & 0.002 (0.015) \\
\bottomrule
\end{tabular}
\begin{tablenotes}[para,flushleft]
  $n=4413$. Standard errors are in parentheses. This table presents the effect of a single neighbor adopting microfinance on own adoption.
\end{tablenotes}
\end{threeparttable}
\end{adjustbox}
\label{exposure1}
\end{table}

\begin{table}[t]
\centering
\centering
\caption{Estimand $\tau(2,0)$}
\begin{adjustbox}{width=\textwidth}
\begin{threeparttable}
\begin{tabular}{lrrrrrrrc}
\toprule
 & ADM & \multicolumn{3}{c}{GNN} & \multicolumn{3}{c}{MLP} & forest \\
  \cmidrule(lr){2-2} \cmidrule(lr){3-5} \cmidrule(lr){6-8} \cmidrule(lr){9-9}
& & 1 Layer & 2 Layer & 3 Layer & 1 Layer & 2 Layer & 3 Layer \\
\midrule
\multicolumn{8}{l}{Leader case} \\
$G_{ee}$ & -0.052 & -0.027 (0.020) & -0.012 (0.014) & -0.036 (0.022) & -0.007 (0.016) & -0.011 (0.016) & -0.019 (0.016) & -0.004 (0.016) \\
$G_{sc}$ & -0.049 & -0.024 (0.022) & -0.001 (0.014) & 0.024 (0.026) & -0.019 (0.019) & -0.022 (0.019) & -0.015 (0.018) & 0.000 (0.016) \\
$G_{all}$ & -0.050 & -0.016 (0.026) & -0.011 (0.017) & 0.010 (0.022) & -0.025 (0.021) & -0.028 (0.021) & -0.016 (0.017) & 0.000 (0.017) \\
\multicolumn{8}{l}{Leader-adopter case} \\
$G_{ee}$ & 0.215 & 0.252 (0.007) & 0.239 (0.014) & 0.095 (0.008) & 0.157 (0.017) & 0.520 (0.024) & 0.549 (0.024) & 0.194 (0.047) \\
$G_{sc}$ & 0.434 & 0.193 (0.049) & 0.135 (0.017) & 0.072 (0.021) & 0.099 (0.053) & 0.049 (0.050) & 0.424 (0.049) & 0.141 (0.037) \\
$G_{all}$ & 0.435 & 0.146 (0.051) & 0.116 (0.018) & 0.337 (0.008) & 0.137 (0.047) & 0.064 (0.056) & 0.193 (0.042) & 0.107 (0.034) \\
\multicolumn{8}{l}{Adopter case} \\
$G_{ee}$ & 0.423 & 0.220 (0.027) & 0.197 (0.020) & 0.220 (0.024) & 0.199 (0.031) & 0.195 (0.027) & 0.201 (0.023) & 0.106 (0.024) \\
$G_{sc}$ & 0.622 & 0.193 (0.023) & 0.188 (0.016) & 0.168 (0.023) & 0.172 (0.021) & 0.168 (0.022) & 0.186 (0.022) & 0.089 (0.019) \\
$G_{all}$ & 0.657 & 0.181 (0.024) & 0.176 (0.017) & 0.158 (0.021) & 0.154 (0.022) & 0.161 (0.021) & 0.166 (0.022) & 0.080 (0.019) \\
\bottomrule
\end{tabular}
\begin{tablenotes}[para,flushleft]
  $n=4413$. Standard errors are in parentheses. This table presents the effect of multiple neighbors adopting microfinance on own adoption.
\end{tablenotes}
\end{threeparttable}
\end{adjustbox}
\label{exposure2}
\end{table}

To compute the estimates, we concatenate the village networks into a single adjacency matrix of size $n=4413$. For the non-ADM estimates, we trim observations with propensity scores outside of $[0.01,0.99]$.\footnote{The only substantial trimming occurs in the leader-adopter case for $\tau(2,0)$. This is because the number of units with two or more treated units is small, no more than 100. The worst case is network $G_{ee}$ for which the smallest post-trimming sample size across $L$ is 1020 for GNNs, 889 for MLPs and 554 for random forests. Fortunately, these are sufficiently large that we draw similar conclusions from the estimates regardless of $L$ or the network.} Standard errors are obtained from the variance estimator defined in \autoref{sest}. Across all network definitions, the bandwidth \eqref{ourb} equals one.

Tables \ref{exposure1} and \ref{exposure2} report results for $\tau(1,0)$ and $\tau(2,0)$ respectively, as well as ADM estimates. First consider $\tau(1,0)$, which contrasts microfinance adoption rates for units with 1 versus 0 initially adopting neighbors. The GNN results are similar across $L$. For the leader case, we obtain precise zeros for almost all estimates. For the leader-adopter case, the GNN estimates are substantially smaller in magnitude than those of the ADM, at most half the magnitude. While the ADM is a different estimand, $\tau(1,0)$ is perhaps the case where they are most logically comparable, and the smaller effect sizes we find may be due to the use of richer network controls. For the adopter case, the GNN estimates are an order of magnitude smaller than the ADM estimates. The MLP estimates are comparable to the GNN estimates, while the random forest estimates are generally smaller, particularly in the adopter case.

The estimand using $\tau(2,0)$ contrasts units with $2+$ versus 0 initially adopting neighbors. The estimates for the leader case are precise zeros, but unlike $\tau(1,0)$, we find sizeable effects for the leader-adopter and adopter cases. For the latter, the estimates are around 20 percentage points. The MLP estimates are comparable to the GNN estimates, while the random forest estimates are substantially smaller. Regardless of method, the conclusion from both tables is that the effect of having multiple adopting neighbors is more than triple the effect of having only one, but the magnitudes are less than half the ADM estimates.

\section{Conclusion}\label{sconclude}

Existing work on network interference with unconfoundedness assumes that it suffices to control for a known, low-dimensional function of the network and covariates $W_i$. In this respect, the approaches may be viewed as semiparametric. We propose to use GNNs to nonparametrically learn the function and provide a behavioral model under which it is low-dimensional and estimable with shallow GNNs.

Our contributions are twofold. First, we observe that the standard formulation of unconfoundedness that controls for a known function of the network is limited in terms of microfoundations, ruling out simultaneous peer effects in the outcome and treatment selection stages. We consider a model of approximate neighborhood interference (ANI) that allows for both. The challenging aspect of the model is that it induces high-dimensional network confounding, which is presumably why it has not been previously studied in a fully nonparametric setting. 

Second, we make the case that estimation remains feasible. We draw an analogy between ANI and approximate sparsity conditions in the lasso literature, which posit that a high-dimensional regression function is well-approximated by a lower-dimensional analog. Under ANI and additional conditions, we show that the propensity score and outcome regression can be approximated by low-dimensional functions of the ego's $L$-neighborhood network for relatively small $L$. \cite{leung2022causal} studies the implications of ANI for asymptotic inference in randomized control trials, while we highlight its utility for handling high-dimensional network confounding. 


\newpage
\part{Supplementary Appendix}

\makeatletter
\@addtoreset{section}{part}
\makeatother
\renewcommand{\thesection}{SA.\arabic{section}} 
\setcounter{section}{0}
\numberwithin{equation}{section} 

\section{Theoretical Properties of GNNs}\label{sgnnadd}

\subsection{\autoref{agnn}}

{\bf Mean-Squared Error.} This section discusses verification of \autoref{agnn}, which is currently beyond the scope of the literature. In the notation of the \autoref{agnn}, let $\varrho$ denote either the propensity score $p_t$ or outcome regression $\mu_t$. To verify \autoref{agnn}(a), \autoref{crc} implies that it suffices to show
\begin{equation}
  \frac{1}{m_n} \sum_{i\in\M_n} \big(\hat{\varrho}(i,\bm{X},\bm{A}) - \varrho(i,\bm{X}_{\N(i,L)}, \bm{A}_{\N(i,L)})\big)^2 = o_p(n^{-1/2}) \label{e3rdpt1}
\end{equation}

\noindent for $L = O(\log n)$. This should be more feasible to verify directly given that (a) $\hat{\varrho}$ is an $L$-layer GNN which only uses information from $(\bm{X}_{\N(i,L)}, \bm{A}_{\N(i,L)})$, (b) the GNN is low-dimensional under the conditions of \autoref{pld}, (c) $\varrho$ is invariant like $\varrho$ under the conditions of \autoref{pexch}, and (d) the data $\{(Y_i, T_i, \bm{X}_{\N(i,L)}, \bm{A}_{\N(i,L)})\}_{i=1}^n$ is $\psi$-dependent conditional on $(\bm{X}, \bm{A})$ by an argument similar to \autoref{lpsi}.

Recent work by \cite{wang2024graph} provides sufficient conditions for \eqref{e3rdpt1} (see their section 3.2), but some limitations prevent a direct application to our setup. 
\begin{enumerate}
  \item They impose restrictions on the dependence structure that hold if the data consists of many independent clusters (their \S 4.2.2). This enables the application of i.i.d.\ concentration inequalities. Analogous results do not presently exist for $\psi$-dependent data.

  \item Their analysis is semiparametric because they assume the GNN model is correctly specified up to the MLPs in each layer (their Assumptions I3 and II6). A nonparametric analysis would require a characterization of the nonparametric function class that GNNs can approximate.
\end{enumerate}

\noindent In \autoref{sgnnprop} below, we provide a characterization, drawing heavily from existing theory in the GNN literature. The result says that, for any fixed $n$ and $L$, if the function lies within a certain nonparametric subclass of invariant functions, then there exists a sequence of GNNs converging to it. This is a step toward what is eventually needed, which is a rate of convergence in terms of the GNN parameters \citep[analogous to][\S B.3]{wang2024graph}. 

\bigskip
\noindent {\bf Stochastic Equicontinuity.} \autoref{agnn}(b) is typically established from \autoref{agnn}(a) and additional conditions. It is most straightforward to verify with cross-fitting, but it can be established without cross-fitting for certain machine learners in low-dimensional settings. \cite{farrell2021deep} do so for MLPs and i.i.d.\ data (their Lemma 10). \cite{wang2024graph} provide sufficient conditions for GNNs (their Appendix D).

In our formulation, we explicitly state the stochastic equicontinuity (SE) condition for the propensity score. This is often left implicit in standard i.i.d.\ settings, but it requires careful attention here. Specifically, the condition for $\Psi_{\mu_t}$ is standard and concerns the outcome regression, while the condition for $\Psi_{p_t}$ is the analog for the propensity score. The latter is not explicitly stated in the literature because, when the data is i.i.d.\ and SUTVA holds, it is straightforward to verify from first principles \citep[e.g.][proof of Theorem 3.1]{farrell2015robust}. However, due to the complexity of our setting, verification from first principles is not apparently possible without stronger restrictions. 

To see this, we discuss two methods of verifying SE for $\Psi_{p_t}$. The first is to observe that $\Psi_{\mu_t}$ and $\Psi_{p_t}$ have the same structure, so techniques used by \cite{wang2024graph} to verify SE for the former can be applied to the latter. The second is to either use cross-fitting or restrict interference in the outcome model to enable verification from first principles.

\bigskip
\noindent {\bf Method 1.} Appendix D of \cite{wang2024graph} verifies SE for $\Psi_{\mu_t}$. Under their assumptions, the arguments are directly applicable to $\Psi_{p_t}$ because they only rely on the following properties that are shared by both functions given uniformly bounded outcomes.
\begin{enumerate}
  \item For any constant $x$, $\Psi_{\mu_t}(Z_i, x)$ and $\Psi_{p_t}(Z_i, x)$ are uniformly bounded and mean zero. 
  \item Their variances are bounded by a universal constant times the mean-squared error of the machine learner.
  \item They are Lipschitz in $x$ with uniformly bounded Lipschitz constant.
\end{enumerate}

\noindent {\bf Method 2.} The next two lemmas show that the SE condition for $\Psi_{p_t}$ holds if either the GNNs are trained using cross-fitting or outcomes follow a generalized neighborhood interference model (ruling out endogenous peer effects in the outcome but not the selection stage).\footnote{By comparison, Corollary 4 of \cite{wang2024graph} verifies the condition using independent clusters (their \S 4.2.2) and SUTVA (their \S 4.1), which is a special case of neighborhood interference. Lemma E.4 of \cite{emmenegger2022treatment} verifies SE using cross-fitting (their Algorithm 1) and neighborhood interference (their equation (1)).}  With a single network, it is not generally possible to cross-fit in a manner that satisfies the conditions of the next lemma, but the purpose of the exercise is to illustrate a technical point discussed in \autoref{rkey} below. Without either restriction, a key step breaks down, and the challenge arises not from dependence per se but from its combination with our richer model of interference.

\begin{seclemma}[Cross-Fitting]\label{lcf}
  Suppose $\{\hat{p}_t(i, \bm{X}, \bm{A})\}_{i=1}^n \indep \bm{Y},\bm{D} \,|\, \bm{X}, \bm{A}$ (for example if the propensity score GNN is trained on independent data). Further suppose
  \begin{equation}
    \limsup_{n\rightarrow\infty} \sum_{s=0}^\infty \delta_n^\partial(s; 2)^{1/2} \psi_n(s)^{1-2/p} < \infty \quad\text{a.s.} \label{3ND}
  \end{equation}

  \noindent for $\delta_n^\partial(\cdot)$ defined prior to \eqref{thepsi}. Under Assumptions \ref{aigno}, \ref{aani}, \ref{aemap}, \ref{amom}(a) and (b), \ref{agnn}(a), and \ref{apsi}(a) and (b), the SE condition for $\Psi_{p_t}$ in \autoref{agnn}(b) holds.
\end{seclemma}
\begin{proof}
  Abbreviate $\mu_i = \mu_t(i,\bm{X},\bm{A})$, $p_i = p_t(i,\bm{X},\bm{A})$, $\hat{p}_i = \hat{p}_t(i,\bm{X},\bm{A})$, and $\ind_i(t) = \ind\{T_i=t\}$. We show that $m_n^{-1/2} \sum_{i\in\M_n} \Psi_{p_t}(Z_i, \hat{p}_i)$ has an $o(1)$ second moment. Since the propensity score GNN is trained on independent data, the second moment equals
  \begin{equation}
    \frac{1}{m_n} \sum_{i\in\M_n} \sum_{j\in\M_n} \E\left[ \E\left[ (Y_i-\mu_i)\ind_i(t) (Y_j-\mu_j)\ind_j(t) \mid \bm{X}, \bm{A} \right] \frac{(\hat{p}_i-p_i) (\hat{p}_j-p_j)}{\hat{p}_i p_i \hat{p}_j p_j} \right]. \label{hofwea42}
  \end{equation}

  Since $\E[(Y_i-\mu_i)\ind_i(t) \mid \bm{X}, \bm{A}] = 0$, the inner conditional expectation equals $\cov((Y_i-\mu_i)\ind_i(t), (Y_j-\mu_j)\ind_j(t) \mid \bm{X}, \bm{A})$. This is the key step of the argument for both lemmas.

  By the argument in the proof of \autoref{lpsi}, $\{(Y_i-\mu_i)\ind_i(t)\}_{i=1}^n$ is conditionally $\psi$-dependent given $(\bm{X},\bm{A})$ (\autoref{dpsidep}) with dependence coefficient $\psi_n(s)$ defined in \eqref{thepsi}. By Corollary A.2 of \cite{kojevnikov2021limit}, which we may apply due to Assumptions \ref{aigno} and \ref{amom}(a),
  \begin{equation*}
    \abs{\cov((Y_i-\mu_i)\ind_i(t), (Y_j-\mu_j)\ind_j(t) \mid \bm{X}, \bm{A})} \leq \psi_n(\ell_{\bm{A}}(i,j))^{1-2/p}.
  \end{equation*}

  Combining the previous equations, by \autoref{amom}(b), there exist universal constants $C,C'>0$ such that
  \begin{align*}
    \eqref{hofwea42} &\leq C \sum_{s=0}^\infty \psi_n(s)^{1-2/p} \frac{n}{m_n} \frac{1}{n} \sum_{i=1}^n \sum_{j=1}^n \ind\{\ell_{\bm{A}}(i,j)=s\} C' \E\left[\abs{\hat{p}_i-p_i}\right] \\
	 &\leq C\,C' \sum_{s=0}^\infty \psi_n(s)^{1-2/p} \frac{n}{m_n} \left( \frac{1}{n} \sum_{i=1}^n \abs{\N^\partial(i,s)}^2 \right)^{1/2} \left( \frac{1}{n} \sum_{i=1}^n \E\left[ (\hat{p}_i-p_i)^2 \right] \right)^{1/2}.
  \end{align*}

  \noindent The last line is $o_p(1)$ by Assumptions \ref{amom}(b) and \ref{agnn}(a) and \eqref{3ND}.
\end{proof}

\begin{secremark}
  A few comments are in order.
  \begin{itemize}
    \item Condition \eqref{3ND} is similar to \autoref{apsi}(d). In \autoref{sveri}, we verify both from the same set of conditions. 

    \item Cross-fitting can be used to verify the SE condition for $\Psi_{\mu_t}$ using the same argument.

    \item A noteworthy special case of cross-fitting is when $\hat{p}_t$ is a non-random function. This is likely important to consider when attempting to extend the argument in Appendix D of \cite{wang2024graph} to $\psi$-dependent data. Since they assume many small independent clusters, they only need a bound for $\max_i \var(\Psi_{p_t}(Z_i, \tilde p_t(i, \bm{X}, \bm{A}))$ for any non-random function $\tilde p_t$. In our setting, we would presumably need to account for covariances due to $\psi$-dependence and bound the more complicated term $\var(m_n^{-1/2} \sum_{i\in\M_n} \Psi_{p_t}(Z_i, \tilde p_t(i, \bm{X}, \bm{A})))$. This is \eqref{hofwea42} with $\hat{p}_t$ replaced by $\tilde{p}_t$.
  \end{itemize}
\end{secremark}

\begin{seclemma}[Neighborhood Interference]\label{lni}
  Suppose outcomes follow the generalized neighborhood interference model $Y_i = g_n(i, T_i, \bm{X}, \bm{A}, \bm{\varepsilon})$ and that \eqref{3ND} holds. Under Assumptions \ref{aigno}, \ref{aani}, \ref{amom}(a) and (b), \ref{agnn}(a), and \ref{apsi}(a), the SE condition for $\Psi_{p_t}$ in \autoref{agnn}(b) holds.
\end{seclemma}
\begin{proof}
  Abbreviate $\mu_i = \mu_t(i,\bm{X},\bm{A})$, $p_i = p_t(i,\bm{X},\bm{A})$, $\hat{p}_i = \hat{p}_t(i,\bm{X},\bm{A})$, and $\ind_i(t) = \ind\{T_i=t\}$. We show that $m_n^{-1/2} \sum_{i\in\M_n} \Psi_{p_t}(Z_i, \hat{p}_i)$ has an $o(1)$ second moment. The second moment equals
  \begin{equation}
    \frac{1}{m_n} \sum_{i\in\M_n} \sum_{j\in\M_n} \E\left[ \E\left[ (Y_i-\mu_i)\ind_i(t) (Y_j-\mu_j)\ind_j(t) \mid \bm{D}, \bm{X}, \bm{A} \right] \frac{(\hat{p}_i-p_i) (\hat{p}_j-p_j)}{\hat{p}_i p_i \hat{p}_j p_j} \right]. \label{hofwea43}
  \end{equation}

  \noindent Under the generalized neighborhood interference model,
  \begin{multline}
    \abs{\E\left[ (Y_i-\mu_i)\ind_i(t) (Y_j-\mu_j)\ind_j(t) \mid \bm{D}, \bm{X}, \bm{A} \right]} \\ = \abs{\cov(Y_i\ind_i(t), Y_j\ind_j(t) \mid \bm{D}, \bm{X}, \bm{A})} \leq \gamma_n(\ell_{\bm{A}}(i,j)/2)^{1-2/p} \label{fjiowee32}
  \end{multline}
  
  \noindent for $\gamma_n$ defined in \autoref{aani}. The equality is shown in \autoref{lstr} below, which uses the generalized neighborhood interference model. The inequality is shown in \autoref{lgamma} below, which does not require this model. 

  Given this result, by \autoref{amom}(b), there exist universal constants $C,C'>0$ such that
  \begin{align*}
    \eqref{hofwea43} &\leq C \sum_{s=0}^\infty \gamma_n(s/2)^{1-2/p} \frac{n}{m_n} \frac{1}{n} \sum_{i=1}^n \sum_{j=1}^n \ind\{\ell_{\bm{A}}(i,j)=s\} C' \E\left[\abs{\hat{p}_i-p_i}\right] \\
	 &\leq C\,C' \sum_{s=0}^\infty \gamma_n(s/2)^{1-2/p} \frac{n}{m_n} \left( \frac{1}{n} \sum_{i=1}^n \abs{\N^\partial(i,s)}^2 \right)^{1/2} \left( \frac{1}{n} \sum_{i=1}^n \E\left[ (\hat{p}_i-p_i)^2 \right] \right)^{1/2}.
  \end{align*}

  \noindent The last line is $o_p(1)$ by Assumptions \ref{amom}(b) and \ref{agnn}(a) and \eqref{3ND}.
\end{proof}

\begin{secremark}\label{rkey}
  The key term is the inner conditional expectation in \eqref{hofwea43}, and the proofs of both lemmas above proceed by bounding the term by a covariance that is decreasing with network distance. Without cross-fitting or generalized neighborhood interference, the problem is that the conditioning variables in the definition of $\mu_i$ are $(T_i, \bm{X}, \bm{A})$, which differ from those in the outer conditional expectation $(\bm{D}, \bm{X}, \bm{A})$. As shown in the next lemma, the term generally equals a covariance term plus a bias. Under the generalized neighborhood interference model, the bias is zero, but in general the term cannot be controlled. Cross-fitting avoids the bias because it enables us to condition only on $(\bm{X},\bm{A})$.
\end{secremark}

\begin{seclemma}\label{lstr}
  Under \autoref{aigno} and the generalized neighborhood interference model of \autoref{lni}, the equality in \eqref{fjiowee32} holds.
\end{seclemma}
\begin{proof}
  Letting $\F_n = (\bm{D}, \bm{X}, \bm{A})$ and $\tilde\mu_i = \E[Y \mid \F_n]$, we have
  \begin{multline*}
    \E\left[ (Y_i-\mu_i)\ind_i(t) (Y_j-\mu_j)\ind_j(t) \mid \F_n \right] \\
    = \E\left[ (Y_i-\tilde\mu_i)\ind_i(t) (Y_j-\tilde\mu_j)\ind_j(t) \mid \F_n \right] + \E\left[ (\tilde\mu_i-\mu_i)\ind_i(t) (\tilde\mu_j-\mu_j)\ind_j(t) \mid \F_n \right] \\
    + \E\left[ (\tilde\mu_i-\mu_i)\ind_i(t) (Y_j-\tilde\mu_j)\ind_j(t) \mid \F_n \right] + \E\left[ (Y_i-\tilde\mu_i)\ind_i(t) (\tilde\mu_j-\mu_j)\ind_j(t) \mid \F_n \right] \\
    = \cov(Y_i\ind_i(t), Y_j\ind_j(t) \mid \F_n) - \underbrace{(\tilde\mu_i-\mu_i)\ind_i(t) (\tilde\mu_j-\mu_j)\ind_j(t)}_\text{bias}.
  \end{multline*}

  \noindent Under the generalized neighborhood interference model and \autoref{aigno}, the bias term is zero because
  \begin{multline*}
    \mu_i\ind_i(t) = \E[g_n(i, t, \bm{X}, \bm{A}, \bm{\varepsilon})\ind_i(t) \mid T_i, \bm{X}, \bm{A}] \\ = \E[g_n(i, t, \bm{X}, \bm{A}, \bm{\varepsilon})\ind_i(t) \mid \bm{D}, \bm{X}, \bm{A}] = \tilde\mu_i\ind_i(t).
  \end{multline*}
\end{proof}

\begin{seclemma}\label{lgamma}
  Under Assumptions \ref{aigno}, \ref{aani}, \ref{amom}(a), and \ref{apsi}(a), $\cov(Y_i, Y_j \mid \bm{D},\bm{X},\bm{A}) \leq C \gamma_n(s/2)^{1-2/p}$ a.s.\ for $p$ given in \autoref{amom}(a) and some universal constant $C>0$.
\end{seclemma}
\begin{proof}
  Let $\F_n'$ be the $\sigma$-algebra generated by $(\bm{D},\bm{X},\bm{A})$. We show that $\{Y_i\}_{i=1}^n$ is conditionally $\psi$-dependent given $\F_n'$ (\autoref{dpsidep}) with dependence coefficient $\gamma_n(s/2)$ \citep[cf.][Proposition 2.3]{kojevnikov2021limit}. Define $(h,h') \in \mathbb{N}\times\mathbb{N}$, $(f, f') \in \mathcal{L}_h \times \mathcal{L}_{h'}$, $s>0$, $(H,H') \in \mathcal{P}_n(h,h';s)$,
  \begin{equation*}
    Y_i^{(s)} = g_{n(i,s)}(i, \bm{D}_{\N(i,s)}, \bm{X}_{\N(i,s)}, \bm{A}_{\N(i,s)}, \bm{\varepsilon}_{\N(i,s)}),
  \end{equation*}

  \noindent $\xi = f((Y_i)_{i\in H})$, $\zeta = f'((Y_i)_{i\in H'})$, $\xi^{(s)} = f( (Y_i^{(s)})_{i \in H} )$, and $\zeta^{(s)} = f'( (Y_i^{(s)})_{i \in H'} )$. By \autoref{apsi}(a), $\xi^{(s/2)} \indep \zeta^{(s/2)} \mid \F_n'$, so
  \begin{align*}
    \abs{\cov(\xi,\zeta \mid \F_n')} &\leq \abs{\cov(\xi-\xi^{(s/2)}, \zeta \mid \F_n')} + \abs{\cov(\xi^{(s/2)}, \zeta-\zeta^{(s/2)} \mid \F_n')} \\ 
			  &\leq 2 \norm{f'}_\infty \E[\abs{\xi-\xi^{(s/2)}} \mid \F_n'] + 2 \norm{f}_\infty \E[\abs{\zeta-\zeta^{(s/2)}} \mid \F_n'] \\
			  &\leq 2 \big(h \norm{f'}_\infty \text{Lip}(f) + h' \norm{f}_\infty \text{Lip}(f')\big) \max_{i\in\mathcal{N}_n} \E[\abs{Y_i - Y_i^{(s/2)}} \mid \F_n'] \\
			  &\leq 2 \big(h \norm{f'}_\infty \text{Lip}(f) + h' \norm{f}_\infty \text{Lip}(f')\big) \gamma_n(s/2),
  \end{align*}

  \noindent the last line by \autoref{aani}. Given $\psi$-dependence, the claim follows from Corollary A.2 of \cite{kojevnikov2021limit}, which we may apply in light of the moment conditions implied by Assumptions \ref{aigno} and \ref{amom}(a). 
\end{proof}

\subsection{WL Function Class}\label{sgnnprop}

MLPs can approximate any measurable function \citep{hornik1989multilayer}, so given the discussion in \autoref{sinvar}, a natural question is whether GNNs can approximate any measurable, {\em invariant} function of graph-structured inputs. In other words, for GNNs to approximate the propensity score or outcome regression, is it enough to assume that these functions are invariant (and satisfy appropriate regularity conditions)? For reasons related to the graph isomorphism problem, stronger restrictions appear to be necessary. 

To see why, let $\F$ be a set of functions mapping $(\bm{X},\bm{A})$ to $\R$. What properties must $\F$ have for any invariant function to be well approximated by a sequence of functions in $\F$? \cite{chen2019equivalence} show that there must exist $F \in \F$ such that $F$ can separate any pair of non-isomorphic graphs in that $F(\bm{X},\bm{A}) \neq F(\bm{X}',\bm{A}')$ for any non-isomorphic $(\bm{X},\bm{A}), (\bm{X}',\bm{A}')$. A function with this property solves the graph isomorphism problem, a problem for which no known polynomial-time solution exists \citep{kobler2012graph,morris2021weisfeiler}. Since GNNs are an example of $\F$ that can be computed in polynomial time, this strongly suggests that approximating any invariant function is too demanding of a requirement.

\begin{secremark}
  GNNs as we apply them have codomain $\R^n$ rather than $\R$, whereas the \cite{chen2019equivalence} result pertains to GNNs with codomain $\R$. This is because they consider graph prediction problems, meaning the outcome is a scalar at the network level. GNNs are also used for node prediction problems, as in our paper, in which case the codomain is $\R^n$ (one prediction per node). Any vector of node predictions aggregated in an invariant manner results in a graph prediction. Hence, to apply their result to GNNs with codomain $\R^n$, we can take $\F = \{F \circ g\colon F \in \F_{\text{GNN}}, g \in \mathcal{A}\}$, where $\mathcal{A}$ is the set of functions $g\colon \R^n \rightarrow \R$ such that $\pi \circ g = g$ for any permutation (bijection) $\pi$ on  $\R^n$. For instance, sum aggregation $x \mapsto \sum_{i=1}^n x_i$ is an element of $\mathcal{A}$. See for example \S 3.1 of \cite{azizian2021expressive}.
\end{secremark}

To define the subclass of invariant functions that GNNs can approximate, we need to take a detour and discuss graph isomorphism tests. The subclass will be defined by a weaker graph separation criterion than solving the graph isomorphism problem, in particular one defined by the {\em Weisfeiler-Leman (WL) test}. This is a (generally imperfect) test for graph isomorphism on which almost all practical graph isomorphism solvers are based \citep{morris2021weisfeiler}. 

Given a labeled graph $(\bm{X},\bm{A})$, the WL test outputs a graph coloring (a vector of labels for each unit) according to the following recursive procedure, whose definition follows \cite{maron2019provably}. At each iteration $t>0$, each unit $i$ is assigned a color $C_t(i)$ from some set $\Sigma$ (e.g.\ the natural numbers) according to
\begin{equation}
  C_t(i) = \Phi\left( C_{t-1}(i), \{C_{t-1}(j)\colon A_{ij}=1\} \right), \label{1WL}
\end{equation}

\noindent where $\Phi(\cdot)$ is a bijective function that takes as input a color and a multiset of neighbors' colors.\footnote{Strictly speaking, this is the 1-WL test.} Intuitively, at each iteration, two units are assigned different colors if they differ in the number of identically colored neighbors, so that at iteration $t$, colors capture some information about a unit's $(t-1)$-neighborhood. Colors are initialized at $t=0$ using a deterministic rule that assigns each $i$ to the same color $C_0(i) \in \Sigma$ if and only if they have the same covariates $X_i$. At each iteration, the number of assigned colors increases, and the algorithm converges when the coloring is the same in two adjacent iterations. This takes at most $n-1$ iterations since there cannot be more than $n$ distinctly assigned colors.

To test whether two labeled graphs are isomorphic, the procedure is run in parallel on both graphs for some number of iterations, typically until convergence. At this point, if there exists a color such that the number of units assigned that color differs in the two graphs, then the graphs are considered non-isomorphic. This procedure correctly identifies isomorphic graphs, but it is underpowered since there exist non-isomorphic graphs considered isomorphic by the WL test \citep{morris2021weisfeiler}. Also, because the number of colors increases each iteration, the test is more powerful when run longer. 

\cite{morris2019weisfeiler} and \cite{xu2018powerful} note the similarity between the GNN architecture \eqref{GNNlayer} and WL test \eqref{1WL}. The former may be viewed as a continuous approximation of the latter, replacing the hash function $\Phi(\cdot)$ with a learnable aggregator $\Phi_{1l}(\cdot)$. They formally show that any GNN has at most the graph separation power of the WL test and furthermore there exist architectures as powerful. 

Returning to the original problem, we now define the class of functions approximated by GNNs in terms of the WL test. Let $\mathcal{S}$ denote the support of $(\bm{X},\bm{A})$. 

\begin{definition}
  For any set of functions $\F$ with domain $\mathcal{S}$, let $\rho(\F)$ be the subset of $\mathcal{S}^2$ such that
  \begin{equation*}
    \big( (x,a), (x',a'') \big) \in \rho(\F) \quad\text{if and only if}\quad F(x,a) = F(x',a') \quad\text{for all}\quad f \in \F.
  \end{equation*}

  \noindent For any two sets of functions $\mathcal{E}, \F$ with domain $\mathcal{S}$, we say that $\mathcal{E}$ is {\em at most as separating as} $\F$ if $\rho(\F) \subseteq \rho(\mathcal{E})$.
\end{definition}

\noindent This is essentially Definition 2 of \cite{azizian2021expressive}. Intuitively, if $\mathcal{E}$ is at most as separating as $\F$, the latter is more complex in the sense that some function in $\F$ can separate weakly more elements of $\mathcal{S}$ than any function in $\mathcal{E}$.

Let $f_{\text{WL},L}$ denote the function of $(\bm{X},\bm{A})$ with codomain $\Sigma^n$ that outputs the vector of node colorings from the WL test run for $L$ iterations. Let $\C(\mathcal{S})$ be the set of continuous functions with domain $\mathcal{S}$. For any $L\in\mathbb{N}$, define the {\em WL function class}
\begin{equation*}
  \F_\text{WL}(L) = \{F^* \in \C(\mathcal{S})\colon \rho(\{F_{\text{WL},L}\}) \subseteq \rho(F^*)\}.
\end{equation*}

\noindent This is the set of continuous functions of $(\bm{X},\bm{A})$ that are at most as separating as the WL test with $L$ iterations. 

The next result says that a function can be approximated by $L$-layer GNNs under the shape restriction that they are elements of the WL function class. This is a stronger shape restriction than invariance because, by construction, the output of the WL test is invariant, so $\F_\text{WL}(L)$ is a subset of the set of all invariant functions.

Consider the GNN architecture in \autoref{eGNNarch} with $\phi_{0l}(\cdot), \phi_{1l}(\cdot), \Phi_o(\cdot)$ being MLPs. For technical reasons, we augment the output layer of the architecture with an additional MLP layer $L+1$ at the output stage, so for an MLP $\Phi^*\colon \R^n \rightarrow \R^n$, the GNN output is $\Phi^*(\Phi_o(h_1^{(L)}), \ldots, \Phi_o(h_n^{(L)}))$. Let $\F_{\text{GNN}*}(L)$ denote the set of such GNNs with $L$ layers, ranging over the parameter space of the MLPs, including their widths and depths. Finally, for any $f \in \F_{\text{GNN}*}(L)$, let $F(i, \bm{X}, \bm{A})$ denote the $i$th component of $F(\bm{X},\bm{A})$.

\begin{sectheorem}\label{tgnnapprox}
  Fix $n,L\in\mathbb{N}$. Suppose that each $X_i$ has the same common, finite support. For any $F^* \in \F_\text{WL}(L)$, there exists a sequence of $L$-layer GNNs $\{F_k\}_{k\in\mathbb{N}} \subseteq \F_{\text{GNN}*}(L)$ such that
  \begin{equation}
    \sup_{(\bm{X},\bm{A}) \in \mathcal{S}} \abs{ F_k(1,\bm{X},\bm{A}) - F^*(1, \bm{X},\bm{A}) } \stackrel{k\rightarrow\infty}\longrightarrow 0. \label{GNNdense}
  \end{equation}
\end{sectheorem}

\noindent In other words, any function in the class $\F_\text{WL}(L)$ can be approximated by a sequence of $L$-layer GNNs in $\F_{\text{GNN}*}(L)$. The result is a consequence of a Stone-Weierstrauss theorem due to \cite{azizian2021expressive} and a version of the \cite{morris2019weisfeiler} and \cite{xu2018powerful} result on the equivalent separation power of GNNs and the WL test. The proof is given below.

The result is essentially Theorem 4 of \cite{azizian2021expressive} but with the distinction that they use $\medcup_L \F_{\text{GNN}*}(L)$ in place of $\F_{\text{GNN}*}(L)$ and $\{F_{\text{WL},\infty}\}$ in place of $\{F_{\text{WL},L}\}$. That is, their theorem states that the set of GNNs ranging over all possible numbers of layers can approximate any continuous function at most as separating as the WL test run until convergence.

\autoref{tgnnapprox} states their result for fixed $L$, and the proof is straightforward from prior results. However, our framing clarifies one of the roles of depth, namely that it determines the strength of the shape restriction implicitly imposed on the function being approximated by GNNs. In particular, because the WL test is more powerful when $L$ is larger, meaning when run for more iterations, \autoref{tgnnapprox} implies that deeper GNNs can approximate weakly richer function classes, or equivalently, impose weaker shape restrictions. We discuss the economic significance of this point in the next subsection.

\bigskip

\begin{proof}[Proof of \autoref{tgnnapprox}]
  Lemma 35 of \cite{azizian2021expressive} (in particular the result for $MGNN_E$) shows that there exists a sequence of GNNs $\{F_k\}_{k\in\mathbb{N}} \subseteq \medcup_L \F_{\text{GNN}*}(L)$ such that \eqref{GNNdense} holds for any $F^*$ in the class
  \begin{equation*}
    \{F^* \in \C(\mathcal{S})\colon \rho(\medcup_L \F_{\text{GNN}*}(L)) \subseteq \rho(F^*)\}.
  \end{equation*}

  \noindent This differs from the statement of our theorem because of the term $\rho(\medcup_L \F_{\text{GNN}*}(L))$, which we wish to replace with $\rho(\{F_{\text{WL},L}\})$. 

  \bigskip
  \noindent {\bf Step 1.} We show that the result holds if we replace $\rho(\medcup_L \F_{\text{GNN}*}(L))$ with $\rho(\F_{\text{GNN}*}(L))$. Then the claim is that, for any {\em fixed} $L$, there exists a sequence of GNNs $\{F_k\}_{k\in\mathbb{N}} \subseteq \F_{\text{GNN}*}(L)$ such that \eqref{GNNdense} holds for any $F^*$ in the class
  \begin{equation}
    \{F^* \in \C(\mathcal{S})\colon \rho(\F_{\text{GNN}*}(L)) \subseteq \rho(F^*)\}. \label{f3j2g0h2j9g}
  \end{equation}

  \noindent In other words, an $L$-layer GNN can arbitrarily approximate any continuous function at most as separating as an $L$-layer GNN. The argument in the proof of Lemma 35 actually applies to \eqref{f3j2g0h2j9g} after some minor changes to notation. The first part of the proof (``We now move to the equivariant case\dots'') up to verifying their equation (26) carries over by redefining the $MGNN_E$ class as having a fixed depth $L$. 

  To show (26), \cite{azizian2021expressive} begin with a GNN $F$ with $L$ layers (their notation uses $T$ in place of $L$) and add an MLP layer that implements their equation (26). For any $F \in \F_{\text{GNN}*}(L)$, consider the mapping
  \begin{equation*}
    (\bm{X},\bm{A}) \mapsto \underbrace{\left( \sum_{i=1}^n F(i,\bm{X},\bm{A}), \dots, \sum_{i=1}^n F(i,\bm{X},\bm{A}) \right)}_{n \text{ times}} \in \R^n
  \end{equation*}

  \noindent (their (26) in our notation). This is an element of $\F_{\text{GNN}*}(L)$ because of the additional MLP layer added to the output of our architecture definition (see the paragraph before the statement of \autoref{tgnnapprox}), so this completes the argument for \eqref{f3j2g0h2j9g}.

  \bigskip
  \noindent {\bf Step 2.} We replace $\rho(\F_{\text{GNN}*}(L))$ with $\rho(\{F_{\text{WL},L}\})$. Theorems VIII.1 and VIII.4 of \cite{grohe2021logic}, which use finiteness of the support of $X_i$, show that these two sets are equivalent. That is, $L$-layer GNNs have the same separation power as the WL test run for $L$ iterations.
\end{proof}

\subsection{Disadvantages of Depth}\label{sbv}

As discussed in \autoref{srfield}, the receptive field is the main consideration when selecting $L$, but \autoref{tgnnapprox} provides a second consideration, which is imposing a weaker implicit shape restriction. It shows that, for GNNs to approximate a target function well, the target must satisfy a shape restriction stronger than invariance, namely that it is at most as separating as the WL test with $L$ iterations. The larger the choice of $L$, the weaker the shape restriction imposed. This cuts against the standard practice of choosing small values of $L$. However, there are several reasons why shallow architectures remain preferable despite this result.

\bigskip
\noindent {\bf Low returns to depth.} For a given graph, how many iterations are required for the WL test to converge? This corresponds to the choice of $L$ for which the shape restriction is weakest. If the number is large for most graphs in practice, then it would suggest gains to choosing large values of $L$.

Unfortunately, the answer is not generally known, being determined by the topology of the input graph in a complex manner. However, there is a range of results bounding the number of iterations required for convergence. For instance, \cite{kiefer2020iteration} construct graphs for which the WL test requires $n-1$ iterations to converge, so such graphs require $n-1$ layers to obtain the weakest shape restriction. This makes the estimation problem extremely high-dimensional, requiring substantially more layers than what is typically required for the receptive field to encompass the entirety of the network. 

Fortunately, theoretical and empirical evidence suggest that such examples are more the exception than the rule and that small choices of $L$ are typically enough to separate many graphs. \cite{babai1980random} show that, with probability approaching one as $n\rightarrow\infty$, in an $n$-unit network drawn uniformly at random from the set of all possible networks, the WL test assigns all units different colors (recall the test must converge at this point) after only {\em two} iterations \citep{morris2021weisfeiler}. Thus, roughly speaking, for large networks, the weakest possible shape restriction is generically achieved with only $L=2$. This might suggest that using a small number of layers is not restrictive in practice. Indeed, \cite{zopf20221} provide empirical evidence on this point, showing that the vast majority of graphs in their dataset can be separated using the WL test after a single iteration.

\bigskip
\noindent {\bf Cost of depth.} Empirically, larger $L$ has been found to result in worse predictive performance, and several explanations have been proposed. The ``oversmoothing'' phenomenon \citep{li2018deeper,oono2020graph} posits that node embeddings tend to become indistinguishable across many units as the number of layers grows. In random geometric graphs (see \autoref{smc}), $L$-neighborhood sizes grow polynomially with $L$, while in Erd\H{o}s-R\'{e}nyi graphs, the growth rate is exponential. Accordingly, a small increase in $L$ can induce a large increase in the number of elements aggregated by $\Phi_{1l}(\cdot)$, so by a law of large numbers intuition, the resulting node embeddings tend to concentrate on the same value. Since node embeddings are meant to represent network positions, which tend to be quite heterogeneous across units, this results in poor predictive performance.

The ``oversquashing'' phenomenon \citep{alon2021bottleneck,topping2022understanding} posits that, as $L$ grows, the GNN aggregates an exceedingly large amount of information due to the growth in neighborhood sizes. This information is stored in node embeddings of relatively small dimension $H$, resulting in information loss, so the effective size of the receptive field remains small as $L$ grows. 

\cite{zhou2021understanding} provide a third explanation, that certain features of common architectures are responsible for variance inflation. In fact, even weaker shape restrictions than those imposed by \autoref{tgnnapprox} are possible using more complex ``$k$-GNN'' architectures, which would theoretically improve bias, but these have greater computational cost and empirically exhibit worse predictive performance and higher variance than the standard architecture \eqref{GNNlayer} \citep{dwivedi2022benchmarking}. These disadvantages may explain in part the common use of the standard architecture with few layers.

\section{Verifying \S8 Assumptions}\label{sveri}

\cite{leung2022causal}, \S A, verifies analogs of Assumptions \ref{apsi}(d) and \ref{ahac} from an older working paper version of \cite{kojevnikov2021limit}. This section repeats the exercise for the published version of the assumptions. We assume throughout that $\max\{\gamma_n(s/2),\psi_n(s)\} \leq \text{exp}(-c(1-4/p)^{-1} s)$ for some $c>0$ and $p$ in \autoref{amom}(a). As in \cite{leung2022causal}, we say a sequence of networks exhibits polynomial neighborhood growth if
\begin{equation*}
  \sup_n \max_{i\in\mathcal{N}_n} \,\abs{\mathcal{N}_{\bm{A}}(i, s)} = C s^d 
\end{equation*}

\noindent for some $C>0$, $d\geq 1$. The sequence exhibits exponential neighborhood growth if 
\begin{equation*}
  \sup_n \max_{i\in\mathcal{N}_n} \,\abs{\mathcal{N}_{\bm{A}}(i, s)} = C e^{\beta s} 
\end{equation*}

\noindent for some $C > 0$ and $\beta = \log \delta(\bm{A})$ \citep[][\S A discusses this choice of $\beta$]{leung2022causal}.

\subsection{\autoref{apsi}(d) and (\ref{3ND})}\label{svapsi}

For polynomial neighborhood growth, choose $v_n = n^{1/(\alpha d)}$ for $\alpha > 2$. The second term in \eqref{2ND} is at most $n^{3/2} \text{exp}(-c\,n^{1/(\alpha d)}) = o(1)$. The first term is 
at most $n^{-1/2} \sum_{s=0}^\infty (Cn^{1/\alpha}) (Cs^d) \text{exp}(-c\,s) = o(1)$ for $k=1$, and for $k=2$, it is at most $n^{-1} \sum_{s=0}^\infty (Cn^{1/\alpha})^2 (Cs^d) \text{exp}(-c\,s) = o(1)$. Finally $\eqref{3ND} \leq \sum_{s=0}^\infty Cs^d \text{exp}(-c\, s) < \infty$.

For exponential neighborhood growth, choose $v_n = \alpha\beta^{-1} \log n$, $\alpha \in (1.5\beta c^{-1}, 0.5)$, with $c$ from the definition of $\psi_n(s)$ above. Such an $\alpha$ exists only if $c>3\beta$, which requires $\psi_n(s)$ to decay sufficiently fast relative to neighborhood growth. The second term in \eqref{2ND} is then at most $n^{3/2} \text{exp}(-c\alpha\beta^{-1} \log n) = n^{1.5-c \alpha\beta^{-1}} = o(1)$. For $k=1$, the first term is at most $n^{-1/2} \sum_{s=0}^\infty C^2\text{exp}(\alpha \log n) \text{exp}((\beta-c)s) = o(1)$, and for $k=2$, it is at most $n^{-1} \sum_{s=0}^\infty C^2\text{exp}(2\alpha \log n) \text{exp}((\beta-c)s) = o(1)$. Finally, $\eqref{3ND} \leq \sum_{s=0}^\infty C\text{exp}((\beta -c)s) < \infty$. 

\subsection{Bandwidth}\label{sbandchoice}

We employ a mix of formal and heuristic arguments to show that the bandwidth \eqref{ourb} satisfies \autoref{ahac}(b)--(c). Under polynomial neighborhood growth, as argued in \S A.2 of \cite{leung2022causal}, $\mathcal{L}(\bm{A}) \approx n^{1/d}$, in which case $b_n = \mathcal{L}(\bm{A})^{1/4} \approx n^{1/(4d)}$. Then \autoref{ahac}(c) holds because $n^{-1} \sum_{i=1}^n n(i,b_n) = Cb_n^d \approx n^{1/4} = o(\sqrt{n})$. \autoref{ahac}(b) holds because, taking $\epsilon=1-4/p$, 
\begin{align}
  n^{-1} \sum_{s=0}^\infty c_n(s,b_n;2)\psi_n(s)^{1-4/p} &\leq C^3 n^{-1} \sum_{s=0}^\infty b_n^{2d} s^d \text{exp}(-c\, s) \label{r910u232gj} \\ &\approx n^{-1} \sqrt{n} \sum_{s=0}^n s^d \text{exp}(-c\, s) = O(n^{-1/2}). \nonumber
\end{align}

Under exponential neighborhood growth, as argued in \S A.2 of \cite{leung2022causal}, $\mathcal{L}(\bm{A}) \approx \log n / \log \delta(\bm{A})$, in which case $b_n \approx 0.25 \log n / \log \delta(\bm{A})$. Then \autoref{ahac}(c) holds because $n^{-1} \sum_{i=1}^n n(i,b_n) = C\text{exp}(\beta b_n) \approx n^{1/4}$. \autoref{ahac}(b) holds because, taking $\epsilon=1-4/p$, 
\begin{align}
  n^{-1} \sum_{s=0}^\infty c_n(s,b_n;2)\psi_n(s)^{1-4/p} &\leq C^3 n^{-1} \sum_{s=0}^\infty \text{exp}(\beta b_n) \text{exp}(\beta s) \text{exp}(-c\, s) \label{fg209h4} \\ &\approx n^{-1} \text{exp}(0.5 \log n) \sum_{s=0}^n \text{exp}((\beta-c) s) = O(n^{-1/2}), \nonumber
\end{align}

\noindent which is $o(1)$ if $c>\beta$, which is weaker than the requirement $c>3\beta$ in \autoref{svapsi}. 

\section{Additional Simulation Results}\label{samc}

\autoref{simresultsRGGiid} presents simulation results from a modified version of the design in \autoref{simresultsRGG} that exhibits sufficient weak dependence required by our assumptions. The results are very similar. In the definition of $V_i(\bm{D},\bm{\nu}; \theta)$ we replace $\nu_i + \sum_{j=1}^n A_{ij}\nu_j / \sum_{j=1}^n A_{ij}$ with a conditionally independent error term with the same distribution, namely $\nu_i + (\sum_{j=1}^n A_{ij})^{-1/2} \tilde\nu_i$ where $\{\tilde\nu_i\}_{i=1}^n \stackrel{iid}\sim \N(0,1)$ is independent of all other structural primitives. The graph exhibits polynomial neighborhood growth in the sense of \autoref{sveri} \citep{leung2022causal}. Since the data-generating process satisfies \autoref{aani} with exponential decay, Assumptions \ref{apsi} and \ref{ahac} are satisfied by the calculations in \autoref{sveri}.

\begin{table}[ht]
\small
\caption{Simulation Results: i.i.d.\ errors}
\begin{threeparttable}
\begin{tabular}{llrrrrrrrrr}
\toprule
& & \multicolumn{3}{c}{$L=1$} & \multicolumn{3}{c}{$L=2$} & \multicolumn{3}{c}{$L=3$} \\
\cmidrule(lr){3-5} \cmidrule(lr){6-8} \cmidrule(lr){9-11}
\multicolumn{2}{c}{$n$} & 1000 & 2000 & 4000 & 1000 & 2000 & 4000 & 1000 & 2000 & 4000 \\
\multicolumn{2}{c}{$\sum_{i=1}^n D_i$} & 540 & 1080 & 2160 & 540 & 1080 & 2160 & 540 & 1080 & 2160 \\
\multicolumn{2}{c}{$H$} & 3 & 6 & 9 & 3 & 6 & 9 & 3 & 6 & 9 \\
\midrule
\multirow[t]{5}{*}{GNN} & $\hat\tau$ & 0.1153 & 0.0415 & 0.0338 & 0.1147 & 0.0311 & 0.0268 & 0.1460 & 0.0417 & 0.0295 \\
 & CI & 0.9144 & 0.9408 & 0.9338 & 0.9062 & 0.9320 & 0.9376 & 0.9022 & 0.9194 & 0.9262 \\
 & CI+ & 0.9496 & 0.9710 & 0.9674 & 0.9450 & 0.9664 & 0.9716 & 0.9380 & 0.9574 & 0.9618 \\
 & SE & 0.2640 & 0.1633 & 0.1135 & 0.2452 & 0.1326 & 0.0915 & 0.2708 & 0.1278 & 0.0859 \\
 & SE+ & 0.2992 & 0.1881 & 0.1316 & 0.2761 & 0.1520 & 0.1055 & 0.3037 & 0.1463 & 0.0989 \\
\cmidrule(lr){1-11}
\multirow[t]{2}{*}{oracle}  & CI & 0.9370 & 0.9496 & 0.9416 & 0.9220 & 0.9516 & 0.9418 & 0.9170 & 0.9502 & 0.9408 \\
 & SE & 0.3266 & 0.1734 & 0.1167 & 0.3092 & 0.1421 & 0.0933 & 0.3265 & 0.1421 & 0.0903 \\
\cmidrule(lr){1-11}
\multirow[t]{2}{*}{naive}  & CI & 0.8566 & 0.8884 & 0.8830 & 0.8410 & 0.8788 & 0.8738 & 0.8318 & 0.8562 & 0.8622 \\
 & SE & 0.2204 & 0.1391 & 0.0967 & 0.1972 & 0.1089 & 0.0752 & 0.2125 & 0.1038 & 0.0701 \\
\cmidrule(lr){1-11}
\multirow[t]{3}{*}{MLP} & $\hat\tau$ & 0.1389 & 0.1129 & 0.1074 & 0.1389 & 0.1129 & 0.1074 & 0.1389 & 0.1129 & 0.1074 \\
 & CI & 0.9056 & 0.8974 & 0.8448 & 0.9056 & 0.8974 & 0.8448 & 0.9056 & 0.8974 & 0.8448 \\
 & SE & 0.2488 & 0.1739 & 0.1222 & 0.2488 & 0.1739 & 0.1222 & 0.2488 & 0.1739 & 0.1222 \\
\cmidrule(lr){1-11}
\multirow[t]{3}{*}{forest} & $\hat\tau$ & 0.1440 & 0.0925 & 0.0683 & 0.1440 & 0.0925 & 0.0683 & 0.1440 & 0.0925 & 0.0683 \\
 & CI & 0.8620 & 0.8778 & 0.8824 & 0.8620 & 0.8778 & 0.8824 & 0.8620 & 0.8778 & 0.8824 \\
 & SE & 0.2134 & 0.1494 & 0.1075 & 0.2134 & 0.1494 & 0.1075 & 0.2134 & 0.1494 & 0.1075 \\
\bottomrule
\end{tabular}
\begin{tablenotes}[para,flushleft]
  \footnotesize See table notes of \autoref{simresultsRGG}. Average bandwidth $\eqref{ourb}=3$ for each SE cell.
\end{tablenotes}
\end{threeparttable}
\label{simresultsRGGiid}
\end{table}

\section{Supporting Lemmas}

Throughout this section we abbreviate $\ind_i(t) = \ind\{T_i=t\}$.

\begin{seclemma}\label{t3rdpt}
  Under Assumptions \ref{aani}, \ref{aemap}, and \ref{aani2}, there exists $C>0$ such that for any $n\in\mathbb{N}$, $i\in\N_n$, and $s$ sufficiently large,
  \begin{multline}
    \abs{p_t(i,\bm{X},\bm{A}) - p_t(i,\bm{X}_{\N(i,r_\lambda(s+1))}, \bm{A}_{\N(i,r_\lambda(s+1))})} \\ \leq C\big( \lambda_n(s+1) + \eta_n(s) (1+n(i,1)) \big) \quad\text{a.s.} \label{papprox}
  \end{multline}

  \noindent Furthermore, if Assumptions \ref{aigno}, \ref{amom}(b), \ref{apsi}(b), and \ref{ahac}(a) hold, then there exists $C>0$ such that for any $n\in\mathbb{N}$, $i\in\N_n$, and $s$ sufficiently large,
  \begin{multline}
    \abs{\mu_t(i,\bm{X},\bm{A}) - \mu_t(i,\bm{X}_{\N(i,r_\lambda(s))}, \bm{A}_{\N(i,r_\lambda(s))})} \\ \leq C\big( \gamma_n(s/2) + \lambda_n(s) + \eta_n(s/2) (1 + n(i,1) + \Lambda_n(i,s/2)\, n(i,s/2)) \big) \quad\text{a.s.}, \label{muapprox}
  \end{multline}

  \noindent where $\Lambda_n(i,s/2)$ is defined in \autoref{apsi}(b).
\end{seclemma}
\begin{proof}
  Fix $i\in\N_n$ such that $n(i,1)-1 = \gamma \in \bm{\Gamma}$.

  {\bf Proof of \eqref{papprox}.} Abbreviate $N_i = \sum_{j=1}^n A_{ij} D_j$. Since $N_i$ is integer-valued and $\bm{\Delta}$ defined in \autoref{aemap} is an interval, by that assumption, there exist $e>0$, $a,b,\alpha \in \R$ and $\beta \in \R \cup \{\infty\}$ with $a<b$ and $\alpha<\beta$ such that for any $\epsilon \in (0,e)$ and $\bm{d} \in \{0,1\}^n$,
  \begin{align}
    \{f_n(i,\bm{d},\bm{A})=t\} &= \left\{ d_i \in [a,b],\,\,\, \sum_{j=1}^n A_{ij} d_j \in [\alpha,\beta] \right\} \nonumber \\
				      &= \left\{ d_i \in [a-\epsilon, b+\epsilon],\,\,\, \sum_{j=1}^n A_{ij} d_j \in [\alpha-\epsilon,\beta+\epsilon] \right\}. \label{Teventequal}
  \end{align}

  \noindent For example, if $T_i = (D_i, \sum_{j=1}^n A_{ij} D_j)$ and $t = (1, 4)$, then this holds for $a=0.5$, $b=1.5$, $\alpha=3.5$, $\beta=4.5$, and $e = 0.1$.

  Fix $s$, and abbreviate $D_j' = h_{n(j,s)}(j,\bm{X}_{\N(j,s)},\bm{A}_{\N(j,s)},\bm{\nu}_{\N(j,s)})$, $N_i' = \sum_{j=1}^n A_{ij}D_j'$, and $\bm{D}'_B = (D_j')_{j\in B}$ for any $B\subseteq \N_n$. Using the first equality of \eqref{Teventequal},
  \begin{multline*}
    p_t(i,\bm{X},\bm{A}) = \prob\left( D_i' + (D_i - D_i') \in [a,b],\,\,\, N_i' + (N_i - N_i') \in [\alpha,\beta] \mid \bm{X}, \bm{A}\right) \\ 
    \leq \prob\left( D_i' \in [a-\epsilon,b+\epsilon],\,\,\, N_i' \in [\alpha-\epsilon,\beta+\epsilon] \mid \bm{X}, \bm{A} \right) \\ + \underbrace{\prob\left(\abs{D_i-D_i'} > \epsilon \mid \bm{X}, \bm{A}\right) + \prob\left(\abs{N_i-N_i'} > \epsilon \mid \bm{X}, \bm{A}\right)}_{R_0}. 
  \end{multline*}

  \noindent By \eqref{Teventequal}, the right-hand side equals
  \begin{equation*}
    \prob\left( D_i' \in [a,b],\,\,\, N_i' \in [\alpha,\beta] \mid \bm{X}, \bm{A}\right) + R_0 = \prob(f_n(i,\bm{D}',\bm{A}) = t \mid \bm{X}, \bm{A}) + R_0, 
  \end{equation*}

  \noindent so that
  \begin{equation}
    p_t(i,\bm{X},\bm{A}) \leq \prob(f_n(i,\bm{D}',\bm{A}) = t \mid \bm{X}, \bm{A}) + R_0. \label{p_tubound}
  \end{equation}

  \noindent By the same argument,
  \begin{align}
    \prob&(f_n(i,\bm{D}',\bm{A}) = t \mid \bm{X}, \bm{A}) = \prob\left( D_i' \in [a,b],\,\,\, N_i' \in [\alpha,\beta] \mid \bm{X}, \bm{A}\right) \nonumber \\ 
	 &\begin{aligned} \leq \prob( D_i \in [a-\epsilon,b+\epsilon],\,\,\, &N_i \in [\alpha-\epsilon,\beta+\epsilon] \mid \bm{X}, \bm{A} ) \\ &+ \prob\left(\abs{D_i-D_i'} > \epsilon \mid \bm{X}, \bm{A}\right) + \prob\left(\abs{N_i-N_i'} > \epsilon \mid \bm{X}, \bm{A}\right) \end{aligned} \nonumber \\
	 &= \prob\left( D_i \in [a,b],\,\,\, N_i \in [\alpha,\beta] \mid \bm{X}, \bm{A} \right) + R_0 \nonumber \\
	 &= p_t(i,\bm{X},\bm{A}) + R_0. \label{p_tlbound}
  \end{align}

  \noindent Combining \eqref{p_tubound} and \eqref{p_tlbound},
  \begin{equation}
    \abs{p_t(i,\bm{X},\bm{A}) - \prob(f_n(i,\bm{D}',\bm{A}) = t \mid \bm{X}, \bm{A})} \leq R_0 \leq \epsilon^{-1} (1 + n(i,1)) \eta_n(s), \label{fbound1}
  \end{equation}

  \noindent the second inequality due to Markov's inequality and \autoref{aani}. 

  Observe that $f_n(i,\bm{D}',\bm{A})$ is a deterministic function of $(\bm{X}_B,\bm{A}_B,\bm{\varepsilon}_B, \bm{\nu}_B)$ for $B={\N(i,s+1)}$ by definition of $D_j'$ and \autoref{aemap}. Then by \autoref{aani2},
  \begin{multline}
    \abs{\prob(f_n(i,\bm{D}',\bm{A}) = t \mid \bm{X}, \bm{A}) \\ - \prob(f_n(i,\bm{D}',\bm{A}) = t \mid \bm{X}_{\N(i,r_\lambda(s+1))}, \bm{A}_{\N(i,r_\lambda(s+1))})} \leq \lambda_n(s+1) \label{fbound2}
  \end{multline}

  \noindent Using \eqref{fbound1} twice (once with the law of iterated expectations) and \eqref{fbound2},
  \begin{equation*}
    \abs{p_t(i,\bm{X},\bm{A}) - p_t(i,\bm{X}_{\N(i,r_\lambda(s+1))},\bm{A}_{\N(i,r_\lambda(s+1))})} \leq \lambda_n(s+1) + 2R_0. 
  \end{equation*}

  {\bf Proof of \eqref{muapprox}.} Noting that $\mu_t(i,\bm{X},\bm{A}) = \E[Y_i \ind_i(t) \mid \bm{X}, \bm{A}] /p_t(i,\bm{X},\bm{A})$, we first bound the numerator. For $B = \N(i,s)$, define $Y_i' = g_{n(i,s)}(i,\bm{D}_B',\bm{X}_B,\bm{A}_B,\bm{\varepsilon}_B)$. By \autoref{lYDani},
  \begin{equation}
    \abs{\E[Y_i \ind_i(t) \mid \bm{X}, \bm{A}] - \E[Y_i' \ind_i(t) \mid \bm{X}, \bm{A}]} \leq \underbrace{\gamma_n(s) + \Lambda_n(i,s) n(i,s) \eta_n(s)}_{R_1}. \label{YYp} 
  \end{equation}

  \noindent Recalling \autoref{aemap}, define $\ind_i(t)' = \bm{1}\{D_i' = d, \sum_{j=1}^n A_{ij}D_j' \in \bm{\Delta}\}$. By \autoref{lindani}, there exists $C'>0$ such that for any $n\in\mathbb{N}$ and $i\in\N_n$,
  \begin{equation*}
    \E[Y_i' \abs{\ind_i(t)-\ind_i(t)'} \mid \bm{X}, \bm{A}] \leq \underbrace{C'(1 + n(i,1)) \eta_n(s)}_{R_2}. 
  \end{equation*}

  \noindent This and \eqref{YYp} yield
  \begin{equation}
    \abs{\E[Y_i \ind_i(t) \mid \bm{X}, \bm{A}] - \E[Y_i' \ind_i(t)' \mid \bm{X}, \bm{A}]} \leq R_1 + R_2. \label{R1R2}
  \end{equation}

  \noindent By \autoref{ahac}(a), $Y_i'\ind_i(t)'$ is a bounded function of $\bm{X}_{\N(i,2s)}$, $\bm{A}_{\N(i,2s)}$, $\bm{\varepsilon}_{\N(i,2s)}$, and $\bm{\nu}_{\N(i,2s)}$, so by \autoref{aani2}
  \begin{equation*}
    \abs{\E[Y_i' \ind_i(t)' \mid \bm{X}, \bm{A}] - \E[Y_i' \ind_i(t)' \mid \bm{X}_{\N(i,r_\lambda(2s))}, \bm{A}_{\N(i,r_\lambda(2s))}]} \leq \lambda_n(2s).
  \end{equation*}

  \noindent Using this and \eqref{R1R2} twice (once with the law of iterated expectations),
  \begin{equation}
    \abs{\E[Y_i \ind_i(t) \mid \bm{X}, \bm{A}] - \E[Y_i \ind_i(t) \mid \bm{X}_{\N(i,r_\lambda(2s))}, \bm{A}_{\N(i,r_\lambda(2s))}]} \leq \underbrace{\lambda_n(2s) + 2(R_1 + R_2)}_{R_1^*}. \label{muapproxnum}
  \end{equation}

  By \eqref{papprox} and \eqref{muapproxnum},
  \begin{align*}
    \mu_t(i,\bm{X},\bm{A}) &= \frac{\E[Y_i \ind_i(t) \mid \bm{X}, \bm{A}]}{p_t(i,\bm{X},\bm{A})} = \frac{\E[Y_i \ind_i(t) \mid \bm{X}_{\N(i,r_\lambda(2s))}, \bm{A}_{\N(i,r_\lambda(2s))}] + R_1^*}{p_t(i,\bm{X}_{\N(i,r_\lambda(2s))}, \bm{A}_{\N(i,r_\lambda(2s))}) + R_2^*} \\
			   &= \mu_t(i,\bm{X}_{\N(i,r_\lambda(2s))}, \bm{A}_{\N(i,r_\lambda(2s))}) + R_3^*
  \end{align*}

  \noindent where, using \autoref{amom}(b),
  \begin{align*}
    &\abs{R_1^*} \leq \lambda_n(2s) + 2\big( \gamma_n(s) + \Lambda_n(i,s) n(i,s) \eta_n(s) + C'(1 + n(i,1)) \eta_n(s) \big), \\
    &\abs{R_2^*} \leq C\big( \lambda_n(2s) + (1+n(i,1))\eta_n(2s-1) \big), \quad\text{and}\\
    &\abs{R_3^*} \leq C'' (\abs{R_1^*}+\abs{R_2^*})
  \end{align*}

  \noindent for some universal $C''>0$. Substituting $s/2$ for $s$ yields the result.
\end{proof}

\begin{seclemma}\label{lYDani}
  Define $B_i = \N(i,s)$, $D_j' = h_{n(j,s)}(j,\bm{X}_{B_j},\bm{A}_{B_j},\bm{\nu}_{B_j})$, $\bm{D}_{B_i}' = (D_j')_{j \in B_i}$, and $Y_i' = g_{n(i,s)}(i,\bm{D}_{B_i}',\bm{X}_{B_i},\bm{A}_{B_i},\bm{\varepsilon}_{B_i})$. Under Assumptions \ref{aigno}, \ref{aani}, and \ref{apsi}(b),
  \begin{equation*}
    \abs{\E[Y_i \ind_i(t) \mid \bm{X}, \bm{A}] - \E[Y_i' \ind_i(t) \mid \bm{X}, \bm{A}]} \leq \gamma_n(s) + \Lambda_n(i,s) n(i,s) \eta_n(s),
  \end{equation*}

  \noindent where $\Lambda_n(i,s)$ is defined in \autoref{apsi}(b) and $\ind_i(t) = \ind\{T_i=t\}$.
\end{seclemma}
\begin{proof}
  By \autoref{aani},
  \begin{equation*}
    \abs{\E[Y_i \ind_i(t) \mid \bm{D}, \bm{X}, \bm{A}] - \E[ g_{n(i,s)}(i,\bm{D}_{B_i},\bm{X}_{B_i},\bm{A}_{B_i},\bm{\varepsilon}_{B_i}) \ind_i(t) \mid \bm{D}, \bm{X}, \bm{A}]} \leq \gamma_n(s).
  \end{equation*}

  \noindent By \autoref{aigno},
  \begin{multline*}
    \E[ g_{n(i,s)}(i,\bm{D}_{B_i},\bm{X}_{B_i},\bm{A}_{B_i},\bm{\varepsilon}_{B_i}) \ind_i(t) \mid \bm{D}=\bm{d}, \bm{X}=\bm{x}, \bm{A}=\bm{a}] \\ = \E[ g_{n(i,s)}(i,\bm{d}_{B_i},\bm{X}_{B_i},\bm{A}_{B_i},\bm{\varepsilon}_{B_i}) \ind_i(t) \mid \bm{X}=\bm{x}, \bm{A}=\bm{a}],
  \end{multline*}

  \noindent which, together with the law of iterated expectations and \autoref{apsi}(b), implies
  \begin{multline*}
    \lvert \E[g_{n(i,s)}(i,\bm{D}_{B_i},\bm{X}_{B_i},\bm{A}_{B_i},\bm{\varepsilon}_{B_i})\ind_i(t) \mid \bm{X}, \bm{A}] \\ - \E[g_{n(i,s)}(i,\bm{D}_{B_i}',\bm{X}_{B_i},\bm{A}_{B_i},\bm{\varepsilon}_{B_i}) \ind_i(t) \mid \bm{X}, \bm{A}] \rvert \\ \leq \Lambda_n(i,s) \sum_{j \in B_i} \E[\abs{D_j-D_j'} \mid \bm{X}, \bm{A}] \leq \Lambda_n(i,s) n(i,s) \eta_n(s),
  \end{multline*}

  \noindent where the last inequality uses \autoref{aani}. Therefore,
  \begin{equation*}
    \E[Y_i \ind_i(t) \mid \bm{X}, \bm{A}] = \E[Y_i' \ind_i(t) \mid \bm{X}, \bm{A}] + R_1 
  \end{equation*}

  \noindent for $\abs{R_1} \leq \gamma_n(s) + \Lambda_n(i,s) n(i,s) \eta_n(s)$.
\end{proof}

\begin{seclemma}\label{lindani}
  Define $Y_i',D_i'$ as in \autoref{lYDani} and $\ind_i(t)' = \bm{1}\{D_i' = d, \sum_{j=1}^n A_{ij}D_j' \in \bm{\Delta}\}$. Under Assumptions \ref{aani}, \ref{aemap}, and \ref{ahac}(a), there exists $C>0$ such that for any $n\in\mathbb{N}$, $i\in\N_n$, and $s\geq 0$,
  \begin{equation*}
    \E[Y_i\abs{\ind_i(t)-\ind_i(t)'} \mid \bm{X}, \bm{A}] \leq C\, (1 + n(i,1)) \eta_n(s).
  \end{equation*}
\end{seclemma}
\begin{proof}
  Recall the definition of $a,b,\alpha,\beta,\epsilon$ prior to \eqref{Teventequal}. Define $N_i = \sum_{j=1}^n A_{ij}D_j$, $N_i' = \sum_{j=1}^n A_{ij}D_j'$, and $\mathcal{C} = \{ \abs{D_i-D_i'} \leq \epsilon, \abs{N_i-N_i'} \leq \epsilon \}$. Then
  \begin{multline}
    \E[Y_i\abs{\ind_i(t)-\ind_i(t)'} \mid \bm{X}=\bm{x}, \bm{A}=\bm{a}] \\ \leq \E[Y_i\abs{\ind_i(t)-\ind_i(t)'} \mid \mathcal{C}, \bm{X}=\bm{x}, \bm{A}=\bm{a}] + C\,\prob(\mathcal{C}^c \mid \bm{X}=\bm{x}, \bm{A}=\bm{a}) \label{boundinds}
  \end{multline}

  \noindent for some universal $C>0$ by \autoref{ahac}(a). By \autoref{aemap},
  \begin{equation*}
    \ind_i(t) = \bm{1}\left\{ D_i \in [a,b],\,\,\, N_i \in [\alpha,\beta] \right\} \quad\text{and}\quad 
    \ind_i(t)' = \bm{1}\left\{ D_i' \in [a,b],\,\,\, N_i' \in [\alpha,\beta] \right\}.
  \end{equation*}

  \noindent Under event $\mathcal{C}$, 
  \begin{align*}
    \bm{1}\big\{ D_i \in [a,b],\,\,\, N_i \in [\alpha,\beta] \big\} &= \bm{1}\{ D_i' + (D_i - D_i') \in [a,b],\,\,\, N_i' + (N_i - N_i') \in [\alpha,\beta] \} \\ 
		      &\leq \bm{1}\big\{ D_i' \in [a-\epsilon,b+\epsilon],\,\,\, N_i' \in [\alpha-\epsilon,\beta+\epsilon] \big\} \\
		      &= \bm{1}\big\{ D_i' \in [a,b],\,\,\, N_i' \in [\alpha,\beta] \big\},
  \end{align*}

  \noindent where the inequality is due to event $\mathcal{C}$ and the last equality is due to \eqref{Teventequal}. By the same argument, $\bm{1}\big\{ D_i' \in [a,b],\,\,\, N_i' \in [\alpha,\beta] \big\} \leq \bm{1}\big\{ D_i \in [a,b],\,\,\, N_i \in [\alpha,\beta] \big\}$, so $\ind_i(t) = \ind_i(t)'$ under event $\mathcal{C}$. Hence, by Markov's inequality and \autoref{aani},
  \begin{equation*}
    \eqref{boundinds} \leq C\epsilon^{-1} (1 + n(i,1)) \eta_n(s)
  \end{equation*}

  \noindent by the argument for \eqref{fbound1}.
\end{proof}

The following notion of weak network dependence is due to \cite{kojevnikov2021limit}. For any $H,H' \subseteq \mathcal{N}_n$, define $\ell_{\bm{A}}(H,H') = \min\{\ell_{\bm{A}}(i,j)\colon i \in H, j \in H'\}$. Let $\{Z_i\}_{i=1}^n \subseteq \R$ be a triangular array, $\bm{Z}_H = (Z_i)_{i \in H}$, $\mathcal{L}_d$ be the set of bounded $\R$-valued Lipschitz functions on $\R^d$, $\text{Lip}(f)$ be the Lipschitz constant of $f \in \mathcal{L}_d$, and
\begin{equation*}
  \mathcal{P}_n(h,h';s) = \left\{ (H,H')\colon H,H' \subseteq \mathcal{N}_n, |H|=h, |H'|=h', \ell_{\bm{A}}(H,H') \geq s \right\}.
\end{equation*}

\begin{secdefinition}\label{dpsidep}
  A triangular array $\{Z_i\}_{i=1}^n$ is {\em conditionally $\psi$-dependent} given a sequence of $\sigma$-algebras $\{\F_n\}_{n\in\mathbb{N}}$ if there exist $C \in (0,\infty)$ and for each $n\in\mathbb{N}$ an $\F_n$-measurable sequence $\{\psi_{n}(s)\}_{s\geq 0}$ with $\psi_{n}(0)=1$ for all $n$ such that
  \begin{equation}
    \abs{\cov(f(\bm{Z}_H), f'(\bm{Z}_{H'}) \mid \mathcal{F}_n)} \leq C hh'(\norm{f}_\infty + \text{Lip}(f)) (\norm{f'}_\infty + \text{Lip}(f')) \psi_{n}(s) \quad\text{a.s.} \label{pwg}
  \end{equation}

  \noindent for all $n,h,h' \in \mathbb{N}$; $s>0$; $f \in \mathcal{L}_h$; $f' \in \mathcal{L}_{h'}$; and $(H,H') \in \mathcal{P}_n(h,h';s)$. We call $\psi_n(s)$ the {\em dependence coefficient} of $\{Z_i\}_{i=1}^n$.
\end{secdefinition}

\begin{seclemma}\label{lpsi}
  Under Assumptions \ref{aigno}, \ref{aani}, \ref{aemap}, \ref{amom}(a) and (b), and \ref{apsi}(a) and (b), for any $t,t' \in \mathcal{T}$, $\{\varphi_{t,t'}(i)\}_{i=1}^n$ is conditionally $\psi$-dependent given the sequence of $\sigma$-algebras generated by $(\bm{X},\bm{A})$ with dependence coefficient $\psi_n(s)$ defined in \eqref{thepsi}.
\end{seclemma}
\begin{proof}
  Let $\F_n$ be the $\sigma$-algebra generated by $(\bm{X}, \bm{A})$, $(h,h') \in \mathbb{N}\times\mathbb{N}$, $(f, f') \in \mathcal{L}_h \times \mathcal{L}_{h'}$, $s>0$, and $(H,H') \in \mathcal{P}_n(h,h';s)$. Define $Z_i = \varphi_{t,t'}(i)$, $\bm{Z}_H = (Z_i)_{i\in H}$, $\xi = f(\bm{Z}_H)$, $\zeta = f'(\bm{Z}_{H'})$, and
  \begin{equation*}
    D_i^{(s)} = h_{n(i,s)}(i, \bm{X}_{\N(i,s)}, \bm{A}_{\N(i,s)}, \bm{\nu}_{\N(i,s)}).
  \end{equation*}

  \noindent For $\bm{D}^{(s)}_{\N(i,s')} = (D_j^{(s)})_{j \in \N(i,s')}$, define
  \begin{align*}
    &\ind_i^{(s)}(t) = \bm{1}\{f_{n(i,s/2)}(i, \bm{D}^{(s/2)}_{\N(i,s/2)}, \bm{A}_{\N(i,s/2)}) = t\}, \\
    &Y_i^{(s)} = g_{n(i,s/2)}(i, \bm{D}^{(s/2)}_{\N(i,s/2)}, \bm{X}_{\N(i,s/2)}, \bm{A}_{\N(i,s/2)}, \bm{\varepsilon}_{\N(i,s/2)}), \\
    &\begin{aligned} Z_i^{(s)} &= \frac{\ind_i^{(s)}(t) (Y_i^{(s)} - \mu_t(i,\bm{X},\bm{A}))}{p_t(i,\bm{X},\bm{A})} + \mu_t(i,\bm{X},\bm{A}) \nonumber\\ &- \frac{\ind_i^{(s)}(t') (Y_i^{(s)} - \mu_{t'}(i,\bm{X},\bm{A})}{p_{t'}(i,\bm{X},\bm{A})} - \mu_{t'}(i,\bm{X},\bm{A}) - \tau(t,t'). \end{aligned}
  \end{align*}

  \noindent Finally, let $\xi^{(s)} = f( (Z_i^{(s)})_{i \in H} )$ and $\zeta^{(s)} = f'( (Z_i^{(s)})_{i \in H'} )$. 
  
  By \autoref{apsi}(a), $(Z_i^{(s/2,\xi)})_{i \in H} \indep (Z_j^{(s/2,\zeta)})_{j \in H'} \mid \F_n$, so 
  \begin{align*}
    \abs{\cov(\xi,\zeta \mid \F_n)} &\leq \abs{\cov(\xi-\xi^{(s/2)}, \zeta \mid \F_n)} + \abs{\cov(\xi^{(s/2)}, \zeta-\zeta^{(s/2)} \mid \F_n)} \\ 
			  &\leq 2 \norm{f'}_\infty \E[\abs{\xi-\xi^{(s/2)}} \mid \F_n] + 2 \norm{f}_\infty \E[\abs{\zeta-\zeta^{(s/2)}} \mid \F_n] \\
			  &\leq 2 \big(h \norm{f'}_\infty \text{Lip}(f) + h' \norm{f}_\infty \text{Lip}(f')\big) \max_{i\in\mathcal{N}_n} \E[\abs{Z_i - Z_i^{(s/2)}} \mid \F_n].
  \end{align*}

  \noindent By Assumptions \ref{aigno} and \ref{amom}(a) and (b), there exists $C>0$ such that for any $n\in\mathbb{N}$ and $i\in\N_n$,
  \begin{equation*}
    \E[\abs{Z_i - Z_i^{(s/2)}} \mid \F_n] \leq C \big( \E[\abs{\ind_i(t) - \ind_i^{(s/2)}(t)} \mid \F_n] + \E[\abs{Y_i - Y_i^{(s/2)}} \mid \F_n] \big).
  \end{equation*}

  \noindent By an argument similar to the proof of \autoref{lYDani},
  \begin{equation*}
    \E[\abs{Y_i - Y_i^{(s)}} \mid \F_n] \leq \gamma_n(s/2) + \Lambda_n(i,s/2) n(i,s/2) \eta_n(s/2).
  \end{equation*}

  \noindent By an argument similar to the proof of \autoref{lindani}, 
  \begin{equation*}
    \E[\abs{\ind_i(t)-\ind_i^{(s)}(t)} \mid \F_n] \leq C' (1 + n(i,1)) \eta_n(s/2)
  \end{equation*}

  \noindent for some universal constant $C'>0$. Hence, for some universal $C''>0$,
  \begin{multline*}
    \max_{i\in\mathcal{N}_n} \E[\abs{Z_i - Z_i^{(s/2)}} \mid \F_n] \\ \leq C'' \max_{i\in\N_n} \left( \gamma_n(s/4) + \eta_n(s/4) \big(1 + n(i,1) + \Lambda_n(i,s/4) n(i,s/4)\big)\right) = C''\psi_n(s).
  \end{multline*}
\end{proof}

\section{Proofs of Main Results}\label{sproofs}

\begin{proof}[Proof of \autoref{pexch}]
  By definition, $T_{\pi(i)} = f_n(\pi(i), \bm{D}, \bm{A})$ and
  \begin{align*}
    &Y_{\pi(i)} = g(\pi(i), \bm{D}, \bm{X}, \bm{A}, \bm{\varepsilon}) = g(\pi(i), (h_n(j, \bm{X}, \bm{A}, \bm{\nu}))_{j=1}^n, \bm{X}, \bm{A}, \bm{\varepsilon}).
  \end{align*}

  \noindent By the invariance assumptions on $f_n,g_n,h_n$,
  \begin{multline*}
    Y_i = g_n(i, \bm{D}, \bm{X}, \bm{A}, \bm{\varepsilon}) = g_n(\pi(i), \pi(\bm{D}), \pi(\bm{X}), \pi(\bm{A}), \pi(\bm{\varepsilon})) \\
    = g_n(\pi(i), (h_n(\pi(j), \pi(\bm{X}), \pi(\bm{A}), \pi(\bm{\nu})))_{j=1}^n, \pi(\bm{X}), \pi(\bm{A}), \pi(\bm{\varepsilon})),
  \end{multline*}

  \noindent and
  \begin{equation*}
    T_i = f_n(\pi(i), \pi(\bm{D}), \pi(\bm{A})) = f_n(\pi(i), (h_n(\pi(j),  \pi(\bm{A}), \pi(\bm{\nu})))_{j=1}^n, \pi(\bm{X}), \pi(\bm{A})),
  \end{equation*}

  \noindent so by the distributional exchangeability assumption, 
  \begin{equation*}
    (Y_i, T_i, \bm{X}, \bm{A}) \stackrel{d}= (Y_{\pi(i)}, T_{\pi(i)}, \pi(\bm{X}), \pi(\bm{A})).
  \end{equation*}

  \noindent It follows that
  \begin{align*}
    \mu_t(i, \bm{x}, \bm{a}) &= \E[Y_i \mid T_i=t, \bm{X}=\bm{x}, \bm{A}=\bm{a}] \\ 
			     &= \E[Y_{\pi(i)} \mid T_{\pi(i)}=t, \pi(\bm{X})=\bm{x}, \pi(\bm{A})=\bm{a}] \\ 
			     &= \E[Y_{\pi(i)} \mid T_{\pi(i)}=t, \bm{X}=\pi(\bm{x}), \bm{A}=\pi(\bm{a})] \\
			     &= \mu_t(\pi(i), \pi(\bm{x}), \pi(\bm{a}))
  \end{align*}

  \noindent and similarly for the generalized propensity score.
\end{proof}

\begin{proof}[Proof of \autoref{tclt}]
  Abbreviate $\ind_i(t) = \ind\{T_i=t\}$, and decompose
  \begin{equation*}
    \sqrt{m_n}(\hat\tau(t,t')-\tau(t,t')) = \frac{1}{\sqrt{m_n}} \sum_{i\in\M_n} \varphi_{t,t'}(i) - R_{1t} + R_{1t'} - R_{2t} + R_{2t'},
  \end{equation*}

  \noindent where 
  \begin{align*}
    R_{1t} &= \frac{1}{\sqrt{m_n}} \sum_{i\in\M_n} \frac{\ind_i(t) (Y_i - \mu_t(i,\bm{X},\bm{A}))}{\hat{p}_t(i,\bm{X},\bm{A}) p_t(i,\bm{X},\bm{A})} (\hat{p}_t(i,\bm{X},\bm{A}) - p_t(i,\bm{X},\bm{A})), \\
    R_{2t} &= \frac{1}{\sqrt{m_n}} \sum_{i\in\M_n} \left( \hat{\mu}_t(i,\bm{X},\bm{A}) - \mu_t(i,\bm{X},\bm{A})\right) \left( 1 - \frac{\ind_i(t)}{\hat{p}_t(i,\bm{X},\bm{A})} \right),
  \end{align*}

  \noindent and likewise for $R_{1t'}$ and $R_{2t'}$. Let $\F_n$ denote the $\sigma$-algebra generated by $(\bm{X},\bm{A})$. By \autoref{lpsi}, $\{\varphi_{t,t'}(i)\}_{i=1}^n$ is conditionally $\psi$-dependent given $\F_n$ in the sense of \autoref{dpsidep} with dependence coefficient $\psi_n(s)$. By Assumptions \ref{amom} and \ref{apsi}(c) and (d), we may apply Theorem 3.2 of \cite{kojevnikov2021limit} to $n^{-1/2} \sum_{i=1}^n \sqrt{n/m_n} \varphi_{t,t'}(i) \ind\{i\in\M_n\}$ to obtain
  \begin{equation*}
    \sigma_n^{-1} \frac{1}{\sqrt{m_n}} \sum_{i\in\M_n} \varphi_{t,t'}(i) \dlimarrow \mathcal{N}(0,1).
  \end{equation*}

  It remains to show that the remainder terms $R_{1t},R_{2t}$ are $o_p(1)$. For $R_{1t}$, this is a consequence of Assumptions \ref{amom}(b) and \ref{agnn}(b). For $R_{2t}$, we can follow the arguments used in the proof of Theorem 3.1 of \cite{farrell2015robust}. Decompose $R_{2t} = R_{21t} + R_{22t}$ for
  \begin{align*}
    &R_{21t} = \frac{1}{\sqrt{m_n}} \sum_{i\in\M_n} (\hat\mu_t(i,\bm{X},\bm{A}) - \mu_t(i,\bm{X},\bm{A})) \left( 1 - \frac{\ind_i(t)}{p_t(i,\bm{X},\bm{A})} \right), \\
    &\begin{aligned} R_{22t} = \frac{1}{\sqrt{m_n}} \sum_{i\in\M_n} &(\hat\mu_t(i,\bm{X},\bm{A}) - \mu_t(i,\bm{X},\bm{A})) \\ &(\hat{p}_t(i,\bm{X},\bm{A}) - p_t(i,\bm{X},\bm{A})) \frac{\ind_i(t)}{\hat{p}_t(i,\bm{X},\bm{A})p_t(i,\bm{X},\bm{A})}. \end{aligned}
  \end{align*}

  \noindent By Assumptions \ref{amom}(b) and \ref{agnn}(b), $R_{21t} = o_p(1)$. Meanwhile
  \begin{multline*}
    \abs{R_{22t}} \leq \max_i (\hat{p}_t(i,\bm{X},\bm{A}) p_t(i,\bm{X},\bm{A}))^{-1} \\
    \times \left( \frac{1}{\sqrt{m_n}} \sum_{i\in\M_n} (\hat\mu_t(i,\bm{X},\bm{A}) - \mu_t(i,\bm{X},\bm{A}))^2 \frac{1}{\sqrt{m_n}} \sum_{i\in\M_n} (\hat{p}_t(i,\bm{X},\bm{A}) - p_t(i,\bm{X},\bm{A}))^2 \right)^{1/2}.
  \end{multline*}

  \noindent The first term on the right-hand side is $O_p(1)$ by \autoref{amom}(b). The second term is $o_p(1)$ by \autoref{agnn}(a).
\end{proof}

\begin{proof}[Proof of \autoref{thac}]
  We first note some useful properties of the PSD kernel \eqref{PSDK}. Like the uniform kernel, it is positive if and only if $\ell_{\bm{A}}(i,j) \leq b_n$. Additionally, as $b_n\rightarrow\infty$ the $ij$th entry tends to $\abs{\mathcal{N}(i,\infty) \cap \mathcal{N}(j,\infty)} \abs{\mathcal{N}(i,\infty)}^{-1/2} \abs{\mathcal{N}(j,\infty)}^{-1/2}$, which is zero if $i,j$ are disconnected. If they are instead connected, the numerator equals $\abs{\mathcal{N}_{\bm{A}}(i,\infty)}$, since $i,j$ are in the same network component, and the denominator equals the same quantity, so the limit is one for connected units.

  {\bf Step 1.} Fix any $q \in \{U, PD\}$, recalling the notation prior to \eqref{ourb}. Let $K_{ij}$ denote the $ij$th entry of $\bm{K}^q$. We first show that the HAC estimator well approximates its oracle analog with known nuisance functions:
  \begin{multline*}
    \abs{\hat\sigma_q^2 - \tilde\sigma_q^2} \plimarrow 0 \quad\text{for}\quad \tilde\sigma_q^2 = \frac{1}{m_n} \sum_{i\in\M_n} \sum_{j\in\M_n} \tilde\varphi_{t,t'}(i) \tilde\varphi_{t,t'}(j) K_{ij}, \\
    \tilde\varphi_{t,t'}(i) = \varphi_{t,t'}(i) - \mu_t(i,\bm{X},\bm{A}) + \mu_{t'}(i,\bm{X},\bm{A}).
  \end{multline*}

  \noindent For $\hat\varphi_{t,t'}(i) = \hat\tau_i(t,t') - \hat\mu_t(i,\bm{X},\bm{A}) + \hat\mu_{t'}(i,\bm{X},\bm{A})$, 
  \begin{multline}
    \abs{\hat\sigma_q^2 - \tilde\sigma_q^2} = \bigg| \frac{1}{m_n} \sum_{i\in\M_n} \big( \hat\varphi_{t,t'}(i) - \tilde\varphi_{t,t'}(i) \big) \sum_{j\in\M_n} \big( \hat\varphi_{t,t'}(j) + \tilde\varphi_{t,t'}(j) \big) K_{ij} \bigg| \\
    \leq \left( \frac{1}{n} \sum_{i=1}^n \big( \hat\varphi_{t,t'}(i) - \tilde\varphi_{t,t'}(i) \big)^2 \frac{1}{n} \sum_{i=1}^n n(i,b_n)^2 \right)^{1/2} \\ \times \frac{n}{m_n} \max_{j\in\N_n} \abs{\hat\varphi_{t,t'}(j) + \tilde\varphi_{t,t'}(j)}, \label{v0239h}
  \end{multline}

  \noindent the first line because $K_{ij} = K_{ji}$ and the second line because $\abs{K_{ij}} \leq \bm{1}\{\ell_{\bm{A}}(i,j) \leq b_n\}$. By Assumptions \ref{amom}(b) and \ref{ahac}(a), the last line is $O_p(1)$. 

  By \autoref{ahac}(c), it remains to show that 
  \begin{equation}
    \frac{1}{n} \sum_{i=1}^n \big( \hat\varphi_{t,t'}(i) - \tilde\varphi_{t,t'}(i) \big)^2 = o_p(n^{-1/2}), \label{bjhoin45o}
  \end{equation}

  \noindent Define
  \begin{multline*}
    \Delta_i(t) = \frac{\ind\{T_i=t\} (Y_i-\hat\mu_t(i,\bm{X},\bm{A})) (p_t(i,\bm{X},\bm{A}) - \hat{p}_t(i,\bm{X},\bm{A}))}{\hat{p}_t(i,\bm{X},\bm{A}) p_t(i,\bm{X},\bm{A})} \\ - (\hat\mu_t(i,\bm{X},\bm{A})-\mu_t(i,\bm{X},\bm{A})) \frac{\ind\{T_i=t\}}{p_t(i,\bm{X},\bm{A})}.
  \end{multline*}

  \noindent The left-hand side of \eqref{bjhoin45o} equals
  \begin{equation*}
    \frac{1}{n} \sum_{i=1}^n \big( \Delta_i(t) - \Delta_i(t') \big)^2.
  \end{equation*}

  \noindent Using Assumptions \ref{ahac}(a) and \ref{agnn}(b) and \autoref{tclt}, this is $o_p(n^{-1/2})$, which establishes \eqref{bjhoin45o}. 

  {\bf Step 2.} We apply Proposition 4.1 of \cite{kojevnikov2021limit} to show $\abs{\tilde\sigma_U^2 - \sigma_n^2} \plimarrow 0$. First, $\E[\tilde\varphi_{t,t'}(i) \mid \bm{X},\bm{A}]=0$ under \autoref{aigno}, as required by their setup.  Their Assumption 2.1 is a consequence of our \autoref{lpsi} and \autoref{apsi}(c). Their Assumption 4.1(i) is satisfied due to our \autoref{ahac}(a). Their Assumption 4.1(ii) is a consequence of their Proposition 4.2. Lastly, their Assumption 4.1(iii) corresponds to our \autoref{ahac}(b). 

  {\bf Step 3.} We apply Proposition 4.1 of \cite{kojevnikov2021bootstrap} to show $\abs{\tilde\sigma_{PD}^2 - \sigma_n^2} \plimarrow 0$. First, $\E[\tilde\varphi_{t,t'}(i) \mid \bm{X},\bm{A}]=0$ under \autoref{aigno}, as required by his setup. His Assumption 4.2 is a consequence of our \autoref{ahac}(b) because, by (A.13) of \cite{kojevnikov2021limit},
  \begin{equation}
    \frac{1}{n^2} \sum_{s=0}^n H_n(s,b_n) \psi_n(s)^{1-4/p} \leq 4 \frac{1}{n} \sum_{s=0}^\infty c_n(s, b_n; 2) \psi_n(s)^{1-4/p}, \label{g329we}
  \end{equation}

  \noindent which is $o(1)$ a.s.\ by \autoref{ahac}(b). The argument then follows from the proof of Claim A.1 of his paper, in particular the steps involving $B_{n,1}$. The only difference is that we replace his kernel weights with $K_{ij}^{PD}$. This has two minor effects on the argument. First, we can take his $\bar{\omega}$ (the uniform upper bound on $\omega_{\bm{A}}(i,j)$) to be 1 without having to invoke additional assumptions. 

  Second, we can directly verify for our weights that the right-hand side of his equation (A.8) is $o(1)$ without additional assumptions. In our context, this expression is
  \begin{equation}
    \sum_{s=0}^n \psi_n(s)^{1-4/p} \frac{1}{n} \sum_{i=1}^n \sum_{j\in \mathcal{N}^\partial(i,s)} \abs{K_{ij}^{PD}-1}, \label{A.8}
  \end{equation}

  \noindent which equals
  \begin{multline*}
    \sum_{s=0}^\infty \sum_{i=0}^\infty \sum_{j=0}^\infty f_n(i,j,s) \quad\text{for}\\ f_n(i,j,s) = \psi_n(s)^{1-4/p} \abs{K_{ij}^{PD}-1} \ind\{j\in \mathcal{N}^\partial(i,s)\} n^{-1} \ind\{i,j \in \mathcal{N}_n, s \in [0,n]\}.
  \end{multline*}

  \noindent First observe that, for any $(i,j,s)$, $f_n(i,j,s) \rightarrow 0$ as $n\rightarrow\infty$. This is because $j\in \mathcal{N}^\partial(i,s)$ implies $\ell_{\bm{A}}(i,j) \leq s$, in which case $K_{ij}^{PD} \rightarrow 1$, as shown at the start of this proof. Second,
  \begin{multline*}
    \abs{f_n(i,j,s)} \leq \abs{g_n(i,j,s)} \quad\text{for}\\ g_n(i,j,s) = 2\psi_n(s)^{1-4/p} \ind\{j\in \mathcal{N}^\partial(i,s)\} n^{-1} \ind\{i,j \in \mathcal{N}_n, s \in [0,n]\},
  \end{multline*}

  \noindent since the weights are uniformly bounded by 1. By Assumptions \ref{apsi}(c) and \ref{ahac}(b), the series $\sum_{s=0}^\infty \sum_{i=0}^\infty \sum_{j=0}^\infty g_n(i,j,s)$ converges to a finite constant. Then \eqref{A.8} is $o(1)$ by the generalized Lebesgue dominated convergence theorem \citep[e.g.][Theorem 19]{royden2017real}, viewing the series as an integral with respect to the product measure constructed from three counting measures over $(s,i,j)$.
\end{proof}

\begin{proof}[Proof of \autoref{crc}]
  We apply \autoref{t3rdpt}. The average (over $i$) of the square right-hand side of \eqref{papprox} is at most of order $e^{-2\alpha s} n^{-1} \sum_{i=1}^n n(i,1)^2$, which is $o_p(n^{-1/2})$ if $s \geq ((4-\epsilon)\alpha)^{-1} \log n$. From the left-hand side of \eqref{papprox}, this choice of $s$ corresponds to $L \geq r_\lambda (((4-\epsilon)\alpha)^{-1} \log n + 1)$.

  The average (over $i$) of the square of the right-hand side of \eqref{muapprox} is at most of order $e^{-\alpha s} (n^{-1} \sum_{i=1}^n n(i,1)^2 + n^{-1} \sum_{i=1}^n \Lambda_n(i,s/2)^2 n(i,s/2)^2)$. Under \autoref{ahac}(a), we satisfy \autoref{apsi}(b) by choosing $\Lambda_n(i,s/2) = 2M$. Then the assumptions imply that this is $o_p(n^{-1/2})$ if $s \geq ((2-\epsilon)(\alpha-\xi/2))^{-1} \log n$. From the left-hand side of \eqref{muapprox}, this corresponds to $L \geq r_\lambda(((2-\epsilon)(\alpha-\xi/2))^{-1} \log n)$.
\end{proof}

\begin{proof}[Proof of \autoref{pld}]
  The rate conditions on the MLP widths and depths result in order $\rho_n = n^{2\kappa} \log^5 n$ parameters for each MLP uniformly over $l$ \citep[][p.\ 187]{farrell2021deep}. Then the number of GNN parameters is $d_o + \sum_{l=1}^L (d_{0l}+d_{1l}) = O(L\rho_n)$, which is $o(\sqrt{n})$ given that $\kappa < 1/4$ and $L=O(\log n)$.
\end{proof}


\FloatBarrier
\phantomsection
\addcontentsline{toc}{section}{References}
\bibliography{GNN}{} 
\bibliographystyle{aer}


\end{document}